\documentclass[12pt,onecolumn]{IEEEtran}
\usepackage{graphicx}
\usepackage{bbm}
\usepackage{algorithmic}
\usepackage{algorithm} 
\usepackage{theorem}
\usepackage{amssymb}
\usepackage{subfigure}
\usepackage[top=1.5in, left=0.75in, right=0.75in, bottom=0.75in]{geometry}
\theoremheaderfont{\upshape\bfseries}
\theorembodyfont{\itshape}
\theoremstyle{plain} 
\newtheorem{Pro}{Proposition}

\newtheorem{Lem}{Lemma}
\newtheorem{Cor}{Corollary}
\newtheorem{Def}{Definition}
\newtheorem{Theorem}{Theorem}

\begin{document}

\abovecaptionskip5pt \belowcaptionskip5pt

\title{A Stochastic Geometry Framework for Analyzing Pairwise-Cooperative Cellular Networks}
\author{Anastasios Giovanidis$^*$ and Fran\c{c}ois Baccelli$^{*\dagger}$\\
email: anastasios.giovanidis@inria.fr, baccelli@math.utexas.edu}

\maketitle
\IEEEpeerreviewmaketitle
 
\footnotetext{Acknowlegments: The work presented in this paper has
been carried out at LINCS (http://www.lincs.fr) and has been supported by INRIA and LINCS, Paris, France.\\
$^*$ INRIA and ENS Paris, team TREC, 23 avenue d'Italie, CS 81321, 75214 Paris Cedex 13, France\\
$^\dagger$ Simons Chair in Mathematics and Electrical and Computer Engineering, UT Austin Wireless Networking and Communications group, 2501 Speedway Stop C0806, Austin, Texas 78712-1687, USA}

\begin{abstract}

Cooperation in cellular networks has been recently suggested as a promising scheme to improve system performance, especially for cell-edge users. 
In this work, we use stochastic geometry to analyze cooperation models where the positions of Base Stations (BSs) follow a Poisson point process distribution and where Voronoi cells define the planar areas associated with them. For the service of each user, either one or two BSs are involved. If two, these cooperate by exchange of user data and channel related information with conferencing over some backhaul link. Our framework generally allows variable levels of channel information at the transmitters. In this paper we investigate the case of limited channel state information for cooperation (channel phase, second neighbour interference), but not the fully adaptive case which would require considerable feedback. The total per-user transmission power is further split between the two transmitters and a common message is encoded. The decision for a user to choose service with or without cooperation is directed by a family of geometric policies depending on its relative position to its two closest base stations. An exact expression of the network coverage probability is derived. Numerical evaluation allows one to analyze significant coverage benefits compared to the non-cooperative case. As a conclusion, cooperation schemes can improve system performance without exploitation of extra network resources.
\end{abstract}


\newpage

\section{Introduction}
\label{SectionI}

The rapid advances in communications technology and mobile network market 
have urged the research community to direct its efforts towards novel communications paradigms that 
can offer a better system performance.  Already resources as time, frequency and space have been exhaustively explored, leading to the suggestion of systems with high degrees of freedom including the Random Access (RA), the Multiple-Input-Multiple-Output (MIMO) and the Orthogonal-Frequency-Division-Multiple-Access (OFDMA) channel. One of the current tendencies, which extends the above, is to understand wireless networks as cooperative entities.

Conceptually, the idea of cooperative communications is that nodes should communicate with each other and try to adapt their transmission behavior in such a way, that their own as well as the entire system performance is improved. Cooperation can be understood either as (a) coordination of node actions after certain exchange of messages between neighbours, or (b) as full cooperation by concurrent transmission of the same data towards a receiver node. The latter type of cooperation is possible, when full data and/or channel state information is exchanged between the cooperating elements.

For cellular architectures, a recent suggestion in 3G and 4G systems has been the Cooperative MultiPoint transmission (CoMP) \cite{CommMag12CoMP}, where neighbouring Base Stations (BSs) exchange data and/or channel information over backhaul links. With use of appropriate precoding and by exploiting the extra information, interference within this - so called - BS cluster can be cancelled out for the users in service \cite{KarakayaliJou}. The concept is also known in the literature as Network MIMO \cite{GesbertTutor}. Theoretically, the idea originates from the Multiple Access Channel (MAC) with conferencing, whose capacity region has been derived by Willems in \cite{WillemsConf83} for the discrete memoryless case. Its Gaussian extension is provided by Bross et al. in \cite{BrossLapWigg08}.

Grouping of BSs into cooperative clusters normally requires a centralized controller, where information over the channel quality network-wide is available \cite{PapadogiannisDyn08}. Furthermore, the full benefits of cooperation in throughput and coverage are known to be achieved when both data and channel state information are fully known to all the transmitters of each cluster \cite{GesbertTutor}. However, such information exchange requires a considerable amount of feedback. This is often unrealistic because of the system resources reserved. The research community aims nowadays at resolving this major issue \cite{TutCoMPYang13}, \cite{deKerretTutorCoMP13}. Approaches to solve this investigate cases of partial, delayed, quantized and limited channel state information exchange between the cooperating entities \cite{JungnickICC13}, \cite{deKerretTutorCoMP13} as more realistic towards a practical system implementation \cite{CoMPTutIrmer11}.

\subsection{Related Literature}
There has been a large amount of both theoretical and practical investigations of the aforementioned concepts in the literature. Considering MAC with conferencing, starting from the work by Willems and its Gaussian variant by Bross et al. numerous extensions of the basic model have been considered. In \cite{MoritzIT11}, Wiese et al. have extended 
the discrete memoryless region to the case of compound channels with an arbitrary number of channel states. Maric et al. in \cite{MaricYaKaramer05} have derived the 
region for the case of 2-transmitter MAC with conferencing and 2-receivers. Further extensions for general $p$-transmitters and $q$-receivers are provided in \cite{PQMAC11}, while Wigger and Kramer have derived the region for the 3-User MIMO Gaussian MAC with conferencing in \cite{WiggKramer09}. A work which departs from the standard model, but considers channels of 2-transmitters and 2-receivers with cooperation in either of the two or both ends of the 
channel, is found in \cite{JindalISIT03} by Jindal et al. There, the transmitters are equipped with multiple antennas and jointly encode both messages using Dirty Paper Coding. 


Related to CoMP, the authors in \cite{ZakhourGesbSP11} investigate a 2-transmitter 2-receiver model and derive from Willems' capacity region an achievable rate region, when each BS sends a private and a common message to both receivers. Further works treating $2\times 2$ models are \cite{JindalISIT03}, \cite{MaricYatesKram07}, \cite{MarschFett08}, \cite{NgGold04} and \cite{ZakhourGesbSP11}. Practical schemes to implement cooperation, combining private/common message and Dirty Paper Coding are found in \cite{fettGLOBECOM08}. Optimal methods to form cooperation clusters between BSs are investigated in \cite{GiovaWCNC12} using tools from integer programming and in \cite{deKerrGesbert12} based on the concept of sparse precoding. In \cite{HardjVuce09} Hardjawana et al. investigate transmitter structures for a system with $N$ cooperating BSs $\times$ $N$ users. Finally, many tutorial papers on implementation issues of CoMP deal with the lingering issue of channel state information feedback, that is necessary to achieve the promised cooperation benefits \cite{TutCoMPYang13}, \cite{deKerretTutorCoMP13}, \cite{JungnickICC13}, \cite{CoMPTutIrmer11}.

The existing research, as shown above, has been limited to models with a small number (usually two) of cooperating transmitters and one or more receivers. The aim is to derive 
performance benefits when cooperation is incorporated and to suggest optimal or practical transmission and reception schemes within the cluster of communicating nodes. On the other hand, the influence of cooperation and the inter-cluster interference effects on a network level is not yet thoroughly investigated. We aim here at analyzing the performance of cooperation in terms of coverage, in a network with an infinite number of nodes. A first approach to this problem (done in parallel to the authors' work) is found by Akoum and Heath in \cite{AkoumHeathSPA12}, where BSs and users are positioned on the Euclidean plane as independent Poisson point processes. The concept of random clustering is used to define cooperation clusters, whose BSs transmit using intercell interference nulling. Although this work aims at analyzing the network in its generality, it does neither consider the issue of the optimal cluster choice (because of the random way cooperation clusters are formed) nor questions of power consumption. Furthermore, the model does not deal with the novel interference field which arises when cooperation is applied network-wide.

\subsection{Work Presentation and Major Contributions}

Our work here primarily aims at providing a \textit{general framework} to treat problems of cooperation between pairs of cellular nodes at a network level. The framework utilizes stochastic geometry to derive exact closed form or integral expressions and evaluate network performance when its nodes are assumed randomly scattered on the infinite plane. We consider scenarios of cooperation with data exchange between the optimal pairs of nodes. The present paper is focused on the case where only limited information over the channel is available such as phase shift and second neighbour interference. The reason for such approach is due to the fact that full information exchange, although may allow perfect adaptation of the transmitters to the channel state, it requires a considerable - often unrealistic - amount of inter-cell feedback exchange. The current analysis can however easily be extended to the case full channel knowledge and beamformer adaptation to the instantaneous channel state, but this is left for future investigation and comparison with the current work. 
More specifically, our basic model \textbf{assumptions} in this work are as follows:

\begin{itemize}
\item The BSs are distributed on the infinite two-dimensional (2D) plane as the realization of a Poisson point process (p.p.p.) with intensity $\lambda$.
\item Each BS $\mathbf{z}$ is the center of a cell, defined as the locus of planar points closest to $\mathbf{z}$ than any other atom of the point process. This is called a 1-Voronoi cell $\mathcal{V}_1\left(\mathbf{z}\right)$.
\item One user is randomly positioned within each cell of the network.
\item At most two BSs may cooperate for the service of a single user. These are the first $\mathbf{b}_1$ and second $\mathbf{b}_2$ closest geographic neighbours of each user, who then lies within their 2-Voronoi cell $\mathcal{V}_2\left(\mathbf{b}_1,\mathbf{b}_2\right)$. In this sense, the BS cooperation pair is chosen by each user and not by a central controller entity. Practical issues on knowledge acquisition of the BS-user distance are resolved by averaging over multiple instantaneous channel realizations \cite{GiovaWCNC12} to cancel out the random effect of fast-fading.
\item The cooperation scenario is similar to the concept of Willems' coding with conferencing \cite{WillemsConf83}. This means here that the original message is split into a private and a common part sent by the cooperating pair to the user. The common part is exchanged over the backhaul.
\item The power consumed per user is fixed to $p$. No power constraint per BS is considered. 
\item Fast-fading from each BS to each user is modeled by a complex number, whose absolute value $\sqrt{G}$ is an independent random variable (r.v.) with Rayleigh distribution (consequently its power $G$ follows the exponential distribution) and whose angle $\Theta$ is an independent r.v. with uniform distribution in $\left[0,2\pi\right)$.
\item Limited channel state information is considered. Specifically the analysis assumes that the transmitters know and exchange the angular shifts resulting from complex fading, so that the reception is in-phase at the user side. The purpose of such modeling assumption is to investigate benefits from scenarios of cooperation that avoid exhaustive information exchange between transmitters within a cluster. In our case, the information exchange is limited to just the phase of the channel and not the amplitude, which is a considerable feedback gain. 

Furthermore, we analyze the case where interference from the second base station is non-causally known, and the transmission from the pair is adapted to the latter by application of Dirty Paper Coding \cite{CostaDP83}. This case does not assume further channel knowledge (i.e. does not refer to the standard combined Zero Forcing and Dirty Paper Coding scheme in \cite{CaireDPC03}) and its purpose is to give an optimistic evaluation of our limited feedback scheme. 
\item A subset of users choses No Cooperation (SISO channels) whereas its complementary subset choses Full Coopeartion ($2\times 1$ MISO channel).
\item The choice to Cooperate or Not is dictated by a family of user-optimal policies, which is one of the main novelties of our work. These policies are strictly geometric and a user choses between them (as binary actions) based on the relative value of the ratio $\frac{r_1}{r_2}$ (of distances from its first and second BS neighbor) to the design parameter $\rho$. The parameter $\rho$ influences the width of some geometric zone of cooperation between the cells of the two neighbouring BSs and is left as an optimization variable.
\end{itemize}

The major \textbf{results} can be summarized as follows:

\begin{itemize}
\item The chosen communication scenario and action policies result in a specific $\mathrm{SINR}$ expression.
\item The coverage probability of a typical planar point is derived, given an $\mathrm{SINR}$ threshold $T$. The coverage further depends on the value of geometric parameter $\rho$. Explicit expressions are provided.
\item The network-wide interference is modeled as a shot-noise field, which explicitly takes into consideration the special influence of cooperation between BSs. The interference r.v. is provided together with its Laplace transform and expected value. It can be seen that although different cooperation schemes lead to the same expected value for interference, the same does not hold for its higher moments.
\item The work considers the possibility of interference elimination of the non-beneficial signal coming from the second closest neighbour. This can be achieved when the interference created to the typical user is measured and known to both MISO transmitters. In this case Dirty Paper Coding (DPC) can be applied, described in \cite{CostaDP83}.
\item The coverage expressions for both scenarios (with and without DPC) are numerically evaluated and the benefits of cooperation with geometric policies are illustrated in plots. Simulation results are compared with the theoretical ones to guarantee the validity of the results.
\item \textbf{It is concluded that:} Cooperation (in the sense of conferencing here) is substantially beneficial for network coverage even with limited channel adaptive schemes and when cooperation is restricted between two neighbouring BSs. DPC significantly improves these benefits.
\item Coverage regions change their shape when varying parameter $\rho$. Due to the geometric policies suggested, advantage is always given to planar areas closer to the cell boundaries.
\end{itemize}

Let us stress that beyond these particular results on cooperation
with limited channel knowledge, the main contribution of the present
paper is the new framework in which these questions are analyzed.
The k-Voronoi diagram allows one to identify the region where 
certain k-tuples of antennas constitute the optimal $k$-cluster. Poisson stochastic geometry allows one to analyze e.g. the probability of coverage but potentially other performance measures as well. This framework can be extended to optimal
cooperation by more than two base stations, to a variety of 
cooperation schemes and to a variety of assumptions on channel knowledge.

\subsection{Organization of the work}

Our work is organized as follows. Section \ref{SectionII} presents the general model of cooperation. Specifically, it explains the geometric notions, gives the modeling assumptions and describes the conferencing power-split scenario. The last part presents the policies of No/Full cooperation used and the resulting cooperation zones at the cell boundaries, which are controlled by parameter $\rho$. Section \ref{SectionIII} provides the derivation of the expressions for the coverage probability using stochastic geometry. The Laplace transform of the Full cooperation signal is derived and a probabilistic comparison with the SISO signal is given. The interference is described as a r.v. and its properties and expected value are provided. The section includes the variation of the above expressions when interference from the second closest neighbour is eliminated by use of DPC. It concludes with the derivation of the expressions for coverage. Section \ref{SectionIV} presents a list of pros and cons for the total coverage, when the proposed policies are applied to the network, compared to the traditional model without cooperation. Section \ref{SectionV} illustrates the numerical evaluation of the expressions, as well as their comparison with simulation results. Finally Section \ref{SectionVI} contains the conclusions of our work. A large amount of supplementary material, as well as proofs that are omitted from the main text can be found in the Appendix.

\textbf{Notation:} We will use capital letters $G$ for random variables (r.v.'s), boldface small letters $\mathbf{x}$ for vectors (including planar points) and small letters $a$ for real/complex quantities. Capital calligraphic letters are reserved for sets (except for the notation on signal $\mathcal{S}$ and interference $\mathcal{I}$). The set of real numbers will be denoted by $\mathbb{R}$ and that of complex numbers by $\mathbb{C}$. The unit imaginary number is denoted by $j:=\sqrt{-1}$. The asterisc notation $a^*$ denotes the complex-conjugate of $a\in\mathbb{C}$. The distance between two planar points $\mathbf{z}$ and $\mathbf{u}$ is denoted by $d\left(\mathbf{z},\mathbf{u}\right)$.


\section{Cooperation in the Network under Study}
\label{SectionII}

\subsection{Geometry and Model Assumptions}
\label{SecIIA}

For the model under study, the considered Base Stations (BSs) are equipped with a single antenna and are positioned at the 
locations of atoms of the realization of a planar p.p.p. with intensity $\lambda$, denoted by $\phi=\left\{\mathbf{z}_i\right\}$. A planar tessellation, called the 1-Voronoi diagram \cite{CompGeomBook}, separates the plane into subregions with the properties that: (a) Exactly one subregion corresponds to each $\mathbf{z}_i:=\left(x_i,y_i\right)\in\phi$. The subregion with center $\mathbf{z}_i$ is denoted by $\mathcal{V}_1\left(\mathbf{z}_i\right)$ and is called a \textbf{1-Voronoi cell}. (b) Considering Euclidean distance, which we denote by $ d\left(\mathbf{z}_i,\mathbf{z}\right)$, the subregion $\mathcal{V}_1\left(\mathbf{z}_i\right)$ consists of all planar points $\mathbf{z}:=\left(x,y\right)$ closer to $\mathbf{z}_i$ than to any other atom in $\phi$, in other words $d\left(\mathbf{z}_i,\mathbf{z}\right)\leq d\left(\mathbf{z}_k,\mathbf{z}\right)$, $\forall \mathbf{z}_k\in\phi\setminus\left\{\mathbf{z}_i\right\}$.  (c) The union of all subregions is the $\mathbbm{R}^2$ plane. The intersection of every pair of cells is of $0$ Lebesgue measure. When the 1-Voronoi cells of two atoms share a common edge, one says that they are \textbf{Delaunay neighbours}. A dual graph of the 1-Voronoi tessellation called the Delaunay graph can be formed if all Delaunay neighbours are connected by an edge \cite{CompGeomBook}. 

Another tessellation of the plane, called the 2-Voronoi diagram \cite{LeekVoronoi}, separates the plane into subregions with the properties that: (a) Exactly one subregion corresponds to each unordered pair $\left(\mathbf{z}_i,\mathbf{z}_n\right)$, $i\neq n$, $\mathbf{z}_i,\mathbf{z}_n\in\phi$. This subregion is denoted by $\mathcal{V}_2\left(\mathbf{z}_i,\mathbf{z}_n\right)$ and is called a \textbf{2-Voronoi cell}. (b) $\mathcal{V}_2\left(\mathbf{z}_i,\mathbf{z}_n\right)$ consists of all planar points $\mathbf{z}$ with the property that $d\left(\mathbf{z}_i,\mathbf{z}\right)\leq d\left(\mathbf{z}_k,\mathbf{z}\right)$ and $d\left(\mathbf{z}_n,\mathbf{z}\right)\leq d\left(\mathbf{z}_k,\mathbf{z}\right)$, $\forall \mathbf{z}_k\in\phi\setminus\left\{\mathbf{z}_i,\mathbf{z}_n\right\}$. In other words, it consists of all planar points geographically closest to $\mathbf{z}_i$ and second-closest to $\mathbf{z}_n$ or the other way around. (c) The union of all subregions is the $\mathbbm{R}^2$ plane. The intersection of every pair of cells is of $0$ Lebesgue measure. Rather interestingly, the 2-Voronoi cell $\mathcal{V}_2\left(\mathbf{z}_i,\mathbf{z}_n\right)$ is non-empty if and only if the atoms $\mathbf{z}_i$ and $\mathbf{z}_n$  are Delaunay neighbours. For formal definitions, proofs and further discussions, the reader is referred to Appendix A and the publications \cite{CompGeomBook} and \cite{LeekVoronoi}. The 1- and 2-Voronoi tessellation in a square with area 1 [$m^2$] based on the uniform positioning of 10 atoms 
is shown in Fig. \ref{fig:Voronoi1} and \ref{fig:Voronoi2}.

\begin{figure}[ht]    
\centering  
\label{fig:VoronoiEXMP}
	 \subfigure[1-Voronoi tessellation example.]{          
           \includegraphics[trim = 25mm 55mm 20mm 60mm, clip, width=0.45\textwidth]{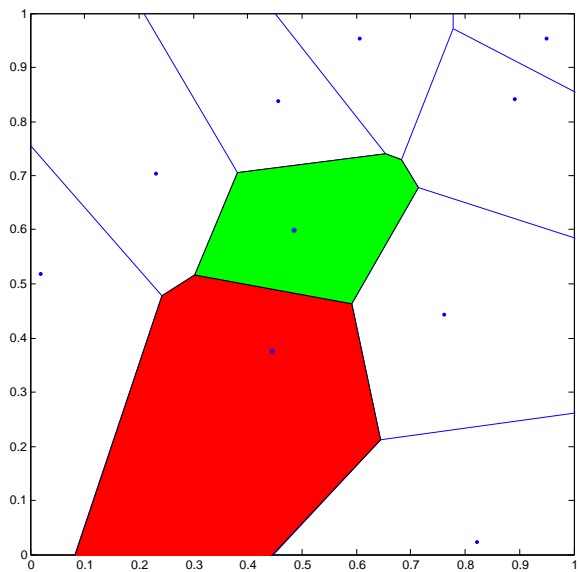}
           \label{fig:Voronoi1}
           }
            \subfigure[2-Voronoi tessellation example.]{          
           \includegraphics[trim = 25mm 55mm 20mm 60mm, clip, width=0.45\textwidth]{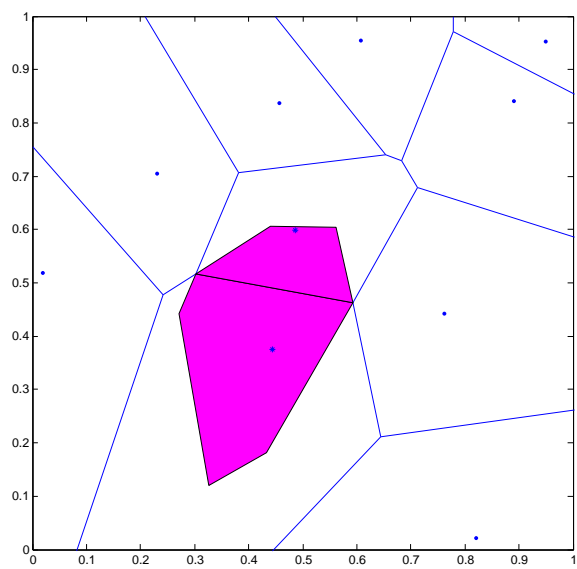}
           \label{fig:Voronoi2}
           }
           \caption{Illustration example for the 1- and 2-Voronoi tessellations of the plane. We highlight the subregions for a specific pair of atoms.}
\end{figure}

%


In the current work, we will consider a geometric cooperation scenario based on the following assumptions:

\begin{itemize}
\item The BSs are connected via a backhaul network, which consists of links of infinite capacity, one for each pair of Delaunay neighbours. The infrastructure in the network thus constitutes a Delaunay graph \cite{CompGeomBook}.

\item Exactly one user with a single antenna is 
located randomly at some point within the 1-Voronoi cell of its BS 
and we write $\mathbf{u}_i\in\mathcal{V}_1(\mathbf{z}_i)$. 
Because of this choice, the user positions do not follow a p.p.p. distribution. 

\item Each user $\mathbf{u}_i$ may be served by either one or two BSs. In the latter case, these are the atoms in $\phi$, which are its \textbf{first} and \textbf{second closest geographic neighbours}. In the former case, it is just the first closest neighbour. We use the notation $\mathbf{b}_{i1}$ and $\mathbf{b}_{i2}$ respectively when referring to these neighbours of $\mathbf{u}_i$. As already mentioned, we assume that the two BSs can communicate by information exchange over the backhaul network.

\item From the point of view of a BS located at $\mathbf{z}_i$, we refer to the user in its 1-Voronoi cell as the \textbf{primary user} $\mathbf{u}_i$ and to all other users served by it but located outside the cell as the \textbf{secondary users}. These constitute a set $\mathcal{N}^s\left(\mathbf{z}_i\right)$, with cardinality that ranges between zero and the number of Delaunay neighbours, depending on the users' position relative to $\mathbf{z}_i$.
\end{itemize}

The modeling assumption of exactly one user per BS is similar to the bipolar models suggested in \cite{BaccelliOpAloha} and \cite{JindalAndrewsComm11}. This choice is also found in the work of Venkatesan \cite{VenkaPIMRCa07} for the performance of cooperation schemes in cellular networks. It is useful because it does not pose questions about user scheduling decisions. It is further justified by the fact that, in practical communications systems, users are split among orthogonal dimensions of time and frequency (e.g. by applying OFDMA techniques) and  
each frequency slot per time is allocated to exactly one user per cell.

The distance 
between user $\mathbf{u}_i$ and its first BS neighbour $\mathbf{z}_i$ 
is equal to
$d\left(\mathbf{z}_i,\mathbf{u}_i\right) := d\left(\mathbf{b}_{i1},\mathbf{u}_i\right) = r_{i1}$. 
This is the maximum radius $r=r_{i1}$ of an open ball $\mathcal{B}\left(\mathbf{u}_i,r\right)$ with center $\mathbf{u}_i$, which is empty of atoms in $\phi$ (see Prop. \ref{CircleV1} in Appendix A). The second nearest neighbour can only be one of the atoms in $\phi$ whose associated Voronoi polygons are adjacent to $\mathcal{V}_1\left(\mathbf{z}_i\right)$. 
The distance between user $\mathbf{u}_i$ and its second BS neighbour $\mathbf{b}_{i2}$ (say $\mathbf{z}_n$),
%
$d\left(\mathbf{z}_n,\mathbf{u}_i\right) := d\left(\mathbf{b}_{i2},\mathbf{u}_i\right) = r_{i2}\geq r_{i1}$. 
This is the maximum radius $r=r_{i2}$ such that $\mathcal{B}\left(\mathbf{u}_i,r\right)$ contains $\mathbf{z}_i$ and only $\mathbf{z}_i$ in its interior and $\mathbf{z}_n$ on its boundary, or 
is empty and contains both $\mathbf{z}_i,\mathbf{z}_n$ on the boundary (In the Poisson case, the last scenario has $0$ probability. We refer also the reader to Prop. \ref{CircleV2} in Appendix A). 
An example of network for the cooperation scenario under study is illustrated in Fig. \ref{fig:ExampleNet}.

\begin{figure}[h]
\centering
\includegraphics[trim = 15mm 65mm 10mm 45mm, clip, width=0.6\textwidth]{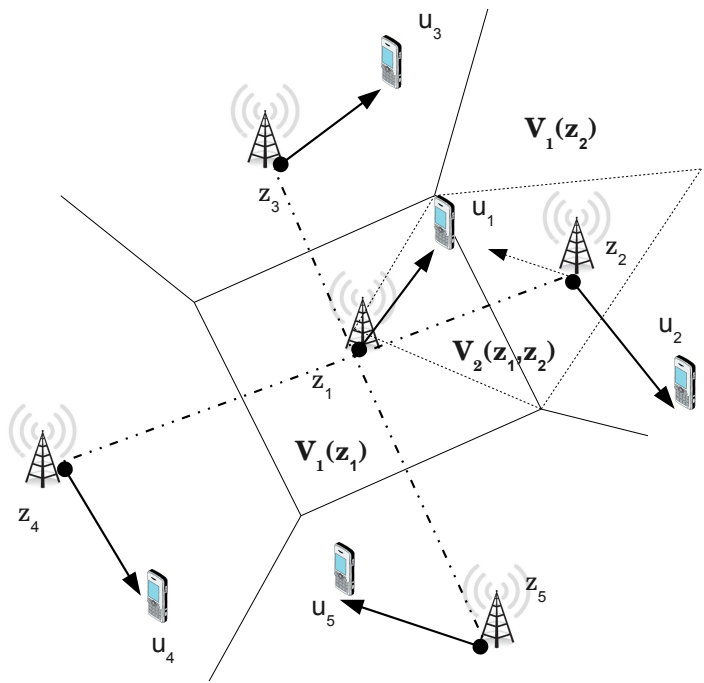}
\caption{Example of topology for a 5 BS network. BS $\mathbf{z}_1$ can communicate with all its Delaunay neighbours over backhaul links. 
Each BS is associated with exactly one user who is located 
within the corresponding 1-Voronoi cell. As shown geometrically, user $\mathbf{u}_1$ is primary user for BS $\mathbf{z}_1$ and secondary user for BS $\mathbf{z}_2$.}
\label{fig:ExampleNet}
\end{figure}


\subsection{Downlink Cooperative Transmission in Pairs with Power Splitting}
\label{SecIIB}

The communications scenario in this work uses the following idea. When two BSs cooperatively serve a user in the downlink, its signal is split into one \textbf{common part} served by both, and \textbf{private parts} served by each one of the serving BSs. The common part contains information known to both transmitters, that is exchanged between them over a reliable conferencing link. The idea is originally based on the seminal work of Willems \cite{WillemsConf83}, where the author derives the capacity region of the Multiple Access Channel (MAC) with conferencing, in the case that two transmitters cooperate over backhaul links to provide service to a single receiver. The proof of achievability is precisely based on such a splitting of the user's signal. The capacity region for the Gaussian noise case has been derived by Bross et al. for the $2\times 1$ model in \cite{BrossLapWigg08} and Maric et al for the $2\times 2$ model in \cite{MaricYaKaramer05}. In both cases the splitting idea is used for achievability (the reader is referred also to Appendix C). A variation of the idea has been applied in a network-MIMO setting by Zakhour and Gesbert for the $2\times2$ model in \cite{ZakhourGesbSP11}, where the beamforming vectors take the role of signal weights and define the power-split-ratios.

In this work we assume that only the information over channel phase shift is available at the transmission pair (and the interference from the second base station, in the DPC case later on), so that full adaptation of the signal to the instantaneous channel conditions is not possible. The purpose is to analyze schemes with limited information exchange. 


For each primary user $\mathbf{u}_i$ located within the 1-Voronoi cell of atom $\mathbf{z}_i$, consider a signal $s_i\in\mathbb{C}$ to be transmitted. The user signals are independent realizations of some random process with power $\mathbb{E}\left[\left|s_i\right|^2\right] = p>0$ and are uncorrelated with other user-signals meaning $\mathbb{E}\left[s_i\cdot s_n^*\right] = 0$, $\forall n\neq i$. The signal destined for user $\mathbf{u}_i$ in the downlink is divided into two parts:
\begin{eqnarray}
s_i & = & s_{i}^{(p)}+ s_{i}^{(c)}.
\label{SumPC}
\end{eqnarray}

\begin{itemize}
\item A \textbf{private} part sent to $\mathbf{u}_i$ from its first BS neighbour $\mathbf{b}_{i1}:=\mathbf{z}_i$, denoted by $s_{i}^{(p)}$. The second neighbour does not have a private part to send.
\item A \textbf{common} part sent by both $\mathbf{b}_{i1}$ and $\mathbf{b}_{i2}$, which is denoted by $s_{i}^{(c)}$. This part is communicated between both BSs over the backhaul links.
\end{itemize}
The two parts are uncorrelated random variables, in other words $\mathbb{E}\left[s_{i}^{(p)}\cdot {s_{i}^{(c)}}^*\right]=0$. For clarity purposes, we will use the notation
$s_{i}^{(c1)}$ for the common signal transmitted from $\mathbf{b}_{i1}$ and $s_{i}^{(c2)}$ for the common signal transmitted from $\mathbf{b}_{i2}$, although the two are actually scaled versions of the same signal.

Considering power issues, we put the constraint that the powers transmitted from both BSs to serve user $\mathbf{u}_i$ should 
sum up to $p$ to guarantee \textbf{power conservation}. This is similar to the beamforming normalization, found in \cite{ZakhourGesbSP11}. We assume that 
the common part is served by both BSs with the same power percentage $a_i\in\left[0,\frac{1}{2}\right]$, which is named the \textbf{power-split-ratio}. This ratio has in general a different value for different users $\mathbf{u}_i$. Hence,
\begin{eqnarray}
\label{powerconservA}
\mathbb{E}\left[\left|s_{i}^{(p)}\right|^2\right] = \left(1-2a_i\right)p, & \mathbb{E}\left[\left|s_{i}^{(c1)}\right|^2\right] = a_ip, & \mathbb{E}\left[\left|s_{i}^{(c2)}\right|^2\right] = a_ip. 
\end{eqnarray}
We will denote the infinite vector of these ratios by $\mathbf{a}$ and the vector of ratios omitting the entry $a_i$, by $\mathbf{a}_{-i}$.
Then BS $\mathbf{b}_{i1}$ consumes $\left(1-a_i\right)p$ in total for user $\mathbf{u}_i$, while BS $\mathbf{b}_{i2}$ consumes $a_ip$ for the same user. Each BS $\mathbf{z}_i$ transmits a total signal $x_i$. This signal is a superposition of:

\begin{itemize}
\item The private and common message for its primary user: $s_{i}^{(p)} + s_{i}^{(c1)}$.
\item The common message for all its secondary users $k\in\mathcal{N}^s\left(\mathbf{z}_i\right)$: $s_{k}^{(c2)}$.
\end{itemize}

We write altogether

\begin{eqnarray}
x_i & = & s_{i}^{(p)} + s_{i}^{(c1)} + \sum_{k\in\mathcal{N}^s\left(\mathbf{z}_i\right)} s_{k}^{(c2)}.
\label{BSsignal}
\end{eqnarray}
%

The BS signals are transmitted over the wireless medium to reach the users and experience real valued path-lost and complex valued fast-fading. For the former, the signal transmitted by BS $\mathbf{z}_n$ is received at the location of user $\mathbf{u}_i$, multiplied by the path-loss factor which depends on the distance $d\left(\mathbf{z}_n,\mathbf{u}_i\right)$.
The fast-fading is written as a product of its amplitude and the exponential shift, i.e. $\sqrt{g}e^{j\theta}$, where $j=\sqrt{-1}$. Its power $\left|\sqrt{g}e^{j\theta}\right|^2=g$ is a realization of a unit-mean exponential random variable (r.v.) $G$, whereas its phase shift $\theta$, as a realization of a unifom r.v. within the interval $\left[0,2\pi\right)$. Consequently, the total power gain from the first and second neighbour $\mathbf{b}_{i1}$ and $\mathbf{b}_{i2}$ to user $\mathbf{u}_i$ is equal to

\begin{eqnarray}
\label{FadeI}
h_{i1} := g_{i1} r_{i1}^{-\beta} &, & h_{i2} := g_{i2} r_{i2}^{-\beta}
\end{eqnarray}
while the total power gain from the first and second neighbour $\mathbf{b}_{n1}$, $\mathbf{b}_{n2}$ of some other user $\mathbf{u}_n$ to $\mathbf{u}_i$ is

\begin{eqnarray}
\label{FadeJ}
h_{n1,i} := g_{n1,i} d_{n1,i}^{-\beta} &, & h_{n2,i} := g_{n2,i} d_{n2,i}^{-\beta},
\end{eqnarray}
with $\beta>2$. Regarding the phase-shift due to fast-fading, the signal transmitted by $\mathbf{b}_{i1}$ will be multiplied by the gain amplitude $\sqrt{h_{i1}}$ and the complex phase shift $e^{j\theta_{i1}}$. Similar factors will appear at the signals from $\mathbf{b}_{i2}$, $\mathbf{b}_{n1,i}$ and $\mathbf{b}_{n2,i}$. Altogether, the signal received at user $\mathbf{u}_i$ is $y_i$ and is equal to

\begin{eqnarray}
y_{i} & = & \left(s_i^{(p)}+s_i^{(c1)}\right)e^{j\theta_{i1}}\sqrt{h_{i1}} +  s_i^{(c2)}e^{j\theta_{i2}}\sqrt{h_{i2}}+\nonumber\\
 &  & +\sum_{\mathbf{u}_n\neq \mathbf{u}_i} \left( \left(s_n^{(p)}+s_n^{(c1)}\right)e^{j\theta_{n1,i}}\sqrt{h_{n1,i}} +  s_n^{(c2)}e^{j\theta_{n2,i}}\sqrt{h_{n2,i}}\right) + \eta_{i}.
\label{RecYui}
\end{eqnarray}
The noise $\eta_{i}$ is a realization of an independent r.v. of distribution $\mathcal{N}\left(0,\sigma_{i}^2\right)$. In the above the sum of signals intended to users $\mathbf{u}_n\neq \mathbf{u}_i$ is the interference received by user $\mathbf{u}_i$. Consequently, the $\mathrm{SINR}_i^{(\mathbf{\theta})}$ for user $\mathbf{u}_i$ takes the form

\begin{eqnarray}
\label{SINRui2a1th}
\mathrm{SINR}_{i}^{(\mathbf{\theta})}\left(\mathbf{a},p\right) & = & \frac{\mathcal{S}_i^{(\mathbf{\theta})}\left(a_i,p\right)}{\sigma_i^2 + \mathcal{I}_i^{(\mathbf{\theta})}\left(\mathbf{a}_{-i},p\right)}\\
\label{SINRui2a2th}
\mathcal{S}_i^{(\mathbf{\theta})}\left(a_i,p\right) & := & h_{i1}\left(1-a_i\right)p + h_{i2}a_ip + 2a_ip\sqrt{h_{i1}h_{i2}}\cos\left(\theta_{i1}-\theta_{i2}\right)\\
\label{SINRui2a3th}
\mathcal{I}_i^{(\mathbf{\theta})}\left(\mathbf{a}_{-i},p\right) & := & \sum_{n\neq i}\mathcal{S}^{(\mathbf{\theta})}_{n,i}\left(a_n,p\right)\\
\label{SINRui2a4th}
\mathcal{S}_{n,i}^{(\mathbf{\theta})}\left(a_n,p\right) & := & h_{n1,i}\left(1-a_n\right)p + h_{n2,i}a_np + 2a_np\sqrt{h_{n1,i}h_{n2,i}}\cos\left(\theta_{n1,i}-\theta_{n2,i}\right).
\end{eqnarray}

In the above expression, an extra term, which includes the square root of the channel gain product, appears in both the beneficial signal (\ref{SINRui2a2th}) and the interfering ones (\ref{SINRui2a4th}) from all $\mathbf{u}_n$'s, $\mathbf{u}_n\neq \mathbf{u}_i$. Its existence is due to the common part transmitted from both BSs serving one user. The signals $s_n^{(c1)}$ and $s_n^{(c2)}$ are scaled versions of each other. As a result, there is correlation at the reception of these signals transmitted by two different BSs. The term would not exist in case $s_n^{(c1)}$ and $s_n^{(c2)}$ were chosen independently. 

In the $\mathrm{SINR}^{(\mathbf{\theta})}_i$ expression above, the maximum beneficial signal for $\mathbf{u}_i$ is equal to $h_{i1}\left(1-a_i\right)p + h_{i2}a_ip+2a_ip\sqrt{h_{i1}h_{i2}}$. This is achieved when the 
two angles are equal $\theta_{i1}=\theta_{i2}$ and as a result the two parts of the common signal from $\mathbf{b}_{i1}$ and $\mathbf{b}_{i2}$ are received in-phase after fading through the channel. 
This knowledge over the phases should be available at the transmitters side. Considering the interference term and the angular influence on it, control over the phase shift is not possible because transmission is usually done using local information and aims at the benefit of users within the local cell. Hence, the angles 
$\theta_{n1,i},\ \theta_{n2,i}$, $\forall n\neq i$ are realizations of independent and uniformly distributed random variables. As a result, it makes sense to consider the interference value for $\mathbf{u}_i$ in expectation over the reception angles of the signals related to all $\mathbf{u}_n\neq \mathbf{u}_i$. In such case, since $\mathbb{E}_{\theta_{n1,i},\ \theta_{n2,i}}\left[\cos\left(\theta_{n1,i} - \theta_{n2,i}\right)\right] = 0$, the extra term disappears. Altogether, we will use the following expression in place of (\ref{SINRui2a1th})-(\ref{SINRui2a4th})

\begin{eqnarray}
\label{SINRui2a1}
\mathrm{SINR}_{i}\left(\mathbf{a},p\right) & = & \frac{\mathcal{S}_i\left(a_i,p\right)}{\sigma_i^2 + \mathcal{I}_i\left(\mathbf{a}_{-i},p\right)}\\
\label{SINRui2a2}
\mathcal{S}_i\left(a_i,p\right) & := & h_{i1}\left(1-a_i\right)p + h_{i2}a_ip + 2a_ip\sqrt{h_{i1}h_{i2}}\\
\label{SINRui2a3}
\mathcal{I}_i\left(\mathbf{a}_{-i},p\right) & := & \sum_{n\neq i} h_{n1,i}\left(1-a_n\right)p + h_{n2,i}a_np.
\end{eqnarray}
It will be shown in section \ref{SectionV}, that the two $\mathrm{SINR}$ formulations in (\ref{SINRui2a1th})-(\ref{SINRui2a4th}) and (\ref{SINRui2a1})-(\ref{SINRui2a3}) provide coverage results that lie very close to each other. In particular, the suggested approximation, which considers the expectation of the interference term over ${\theta_{n1,i},\ \theta_{n2,i}}$ is valid for our model and produces results, which are almost identical with the original.

\subsection{A Family of Geometric Policies with either Full or No Cooperation}

In all above $\mathrm{SINR}$ expressions, the vector of power-split-ratios $\mathbf{a}$ is not predifined and it gives a 
continuous range of cooperation possibilities. In this work we will focus on the two extreme cases, given any fixed vector of choice $\mathbf{a}_{-i}$ for all users other than $\mathbf{u}_i$. The two cases are \textbf{user-optimal} because $a_i=0$ and $a_i=\frac{1}{2}$ are the \textbf{two possible minimizers of the signal} $\mathcal{S}_i\left(a_i,p\right)$ in (\ref{SINRui2a2}), when the choices of other users are held fixed. The choice depends on the relative quality of the channels from the two neighbouring BSs, namely the ratio $\frac{h_{i2}}{h_{i1}}$. If we ignore (or average out) the random effects of the fast fading, the critical ratio to define whether we fall in the case of No or Full cooperation is the distance ratio $\rho_i:=\frac{r_{i1}}{r_{i2}}$ (this is shown in Appendix D).
The $\mathrm{SINR}$ then takes following forms, by substitution of the appropriate values of $a_i$ in (\ref{SINRui2a1}).

\begin{itemize}
\item \textbf{No Cooperation} ($\mathrm{No\ Coop}$) for $a_i= 0$, which gives
\begin{eqnarray}
\label{NoC}
\mathrm{SINR}_i\left(0,\mathbf{a}_{-i}, p\right) & = & \frac{h_{i1} p}{\sigma_i^2 + \mathcal{I}_i\left(\mathbf{a}_{-i},p\right)}.
\end{eqnarray}
\item \textbf{Full Cooperation} ($\mathrm{Full\ Coop}$) for $a_i= \frac{1}{2}$, which gives
\begin{eqnarray}
\mathrm{SINR}_i\left(\frac{1}{2},\mathbf{a}_{-i},p\right) & = & \frac{\frac{\left(\sqrt{h_{i1}}+\sqrt{h_{i2}}\right)^2}{2}p}{\sigma_i^2 + \mathcal{I}_i\left(\mathbf{a}_{-i},p\right)}.
\label{FC}
\end{eqnarray}
\end{itemize}

\begin{Def}
\label{PolicyGeomDef}
The \textbf{user-optimal geometric policies} considered in this work are the family of policies with \textbf{global} parameter $\rho\in\left[0,1\right]$ such that
\begin{eqnarray}
\label{PolicyGeom}
a_i & = & \left\{
\begin{tabular}{l l l}
$0$ & $(\mathrm{No\ Coop})$ & , if $r_{i1}\leq \rho r_{i2}$\\
$\frac{1}{2}$ & $(\mathrm{Full\ Coop})$ & , if $r_{i1}> \rho r_{i2}$
\end{tabular}
\right..
\label{CaseOPTstoch}
\end{eqnarray}
The policies are called "user-optimal" and "geometric" because the optimal choice to cooperate or not depends on the relative position of each user to its two closest BSs. The ratio $\rho\in\left[0,1\right]$ defines the planar cooperation regions, has a global value network-wide and is left as an optimization variable.
\end{Def}
The $\mathrm{SINR}_i$ in (\ref{SINRui2a1}) can be rewritten using the policies in Def. \ref{PolicyGeomDef} and the expressions in (\ref{NoC}) and (\ref{FC}) as

\begin{eqnarray}
\mathrm{SINR}_i\left(\rho,r_{i1},r_{i2},\mathbf{a}_{-i},p\right) & = & \mathrm{SINR}_i\left(0,\mathbf{a}_{-i},p,r_{i1},r_{i2}\right) \mathbbm{1}_{\left\{r_{i1}\leq \rho r_{i2}\right\}}\nonumber\\ 
& + & \mathrm{SINR}_i\left(\frac{1}{2},\mathbf{a}_{-i}, p, r_{i1},r_{i2}\right) \mathbbm{1}_{\left\{r_{i1}> \rho r_{i2}\ \&\ r_{i1}\leq r_{i2}\right\}}.
\label{SINRo2}
\end{eqnarray}
The indicator functions depend on the relative user distance between the two BSs and the parameter $\rho$, as long as user $\mathbf{u}_i$ is a planar point within the 2-Voronoi cell of $\mathbf{b}_{i1}$ and $\mathbf{b}_{i2}$. The expression in (\ref{SINRo2}) takes the following two values at the lower and upper limits of $\rho$:
\begin{itemize}
\item For $\rho=0$ $\mathrm{Full\ Coop}$ covers the entire 2-Voronoi cell.
\item For $\rho=1$ $\mathrm{No\ Coop}$ is applied in the entire 2-Voronoi cell.
\end{itemize}
Considering intermediate values of $\rho$, we provide the geometric locus of points on the plane 
where $\mathrm{Full\ Coop}$ is applied, in the following Lemma (for a proof see Appendix E).

\begin{Lem}
Given a fixed parameter $\rho$ and locations of $\mathbf{b}_{i1}$ and $\mathbf{b}_{i2}$ on the plane, 
the geometric locus of these points that satisfy $r_{i1}\leq \rho r_{i2}$ for $\rho\in\left[0,1\right)$ 
is a \textbf{disc}. For $\rho=1$ the locus of points degenerates to the line which passes over the 1-Voronoi boundary of the two cells.
\label{ProGL}
\end{Lem}

In Fig. \ref{fig:ExampleCoop04} and Fig. \ref{fig:ExampleCoop09} we illustrate the geometric locus of points where $\mathrm{Full\ Coop}$ is applied, that is all points for which $\mathbbm{1}_{\left\{r_{i1}> \rho r_{i2}\ \&\ r_{i1}\leq r_{i2}\right\}} = \mathbbm{1}_{\left\{r_{i1}> \rho r_{i2}\ \&\ \mathcal{V}_2\left(\mathbf{b}_{i1},\mathbf{b}_{i2}\right)\right\}} =1$, for a given realization of BS positions. This is the subregion of the
2-Voronoi cell which lies outside the discs described above. The larger the value of $\rho$, the "thinner" the $\mathrm{Full\ Coop}$ region. We provide two examples, one for $\rho=0.4$ and another one for $\rho=0.9$. Observe that in the case of small $\rho$ there can appear subregions of the 2-Voronoi cell far away from the boundary and specifically even "behind" the BS where $\mathrm{Full\ Coop}$ can be applied. Furthermore, in the case of values of $\rho$ close to $1$ observe how the $\mathrm{Full\ Coop}$ region tends to disappear, when the 
two circles reach the 1-Voronoi edge.

\begin{figure}[t]    
\centering  
\label{fig:VoronoiREG}
	 		\subfigure[Cooperation Regions for $\rho=0.4$.]{          
           \includegraphics[trim = 15mm 55mm 10mm 45mm, clip, width=0.45\textwidth]{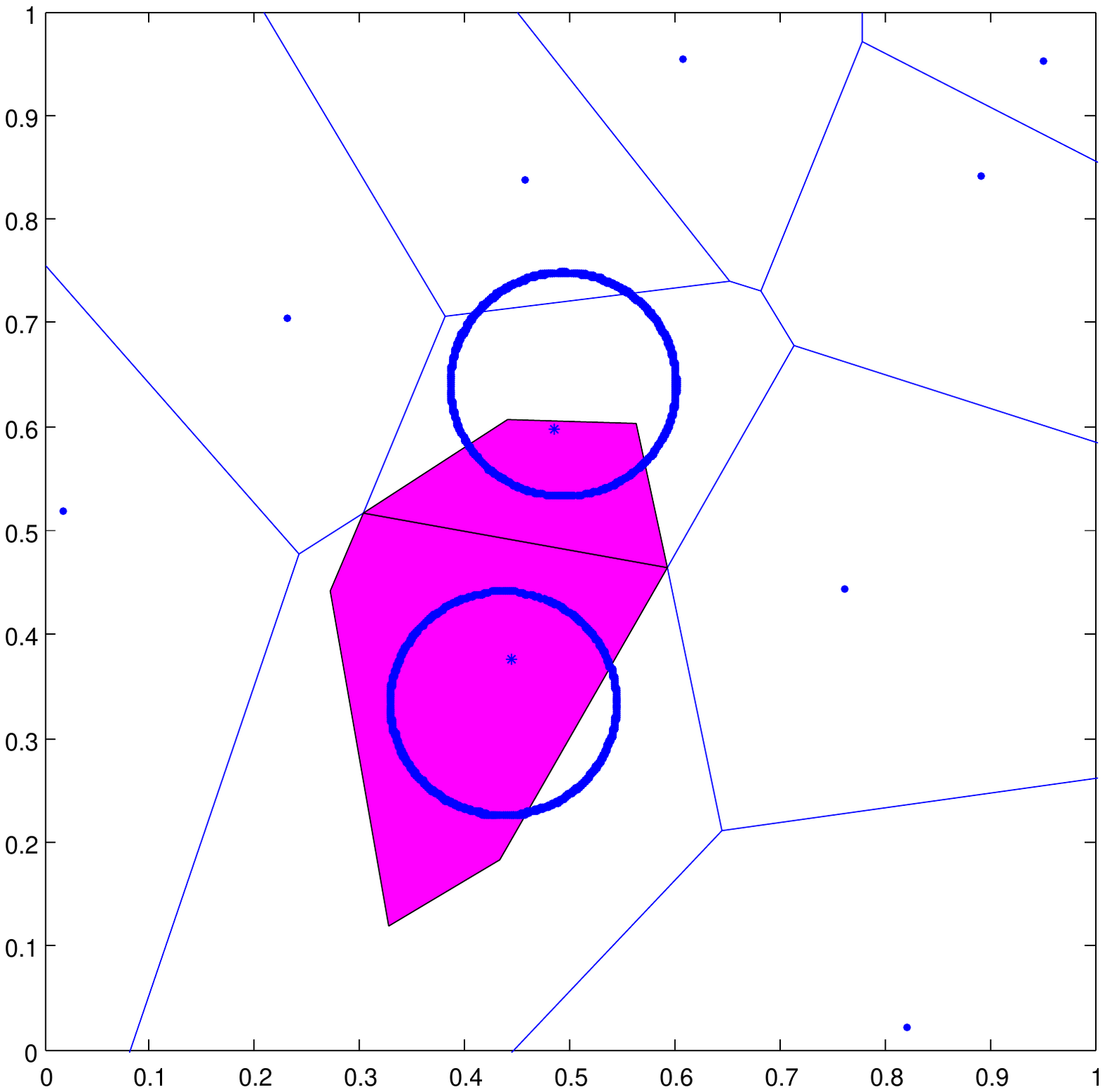}
           \label{fig:ExampleCoop04}
           }
            \subfigure[Cooperation Regions for $\rho=0.9$]{          
           \includegraphics[trim = 15mm 55mm 10mm 45mm, clip, width=0.45\textwidth]{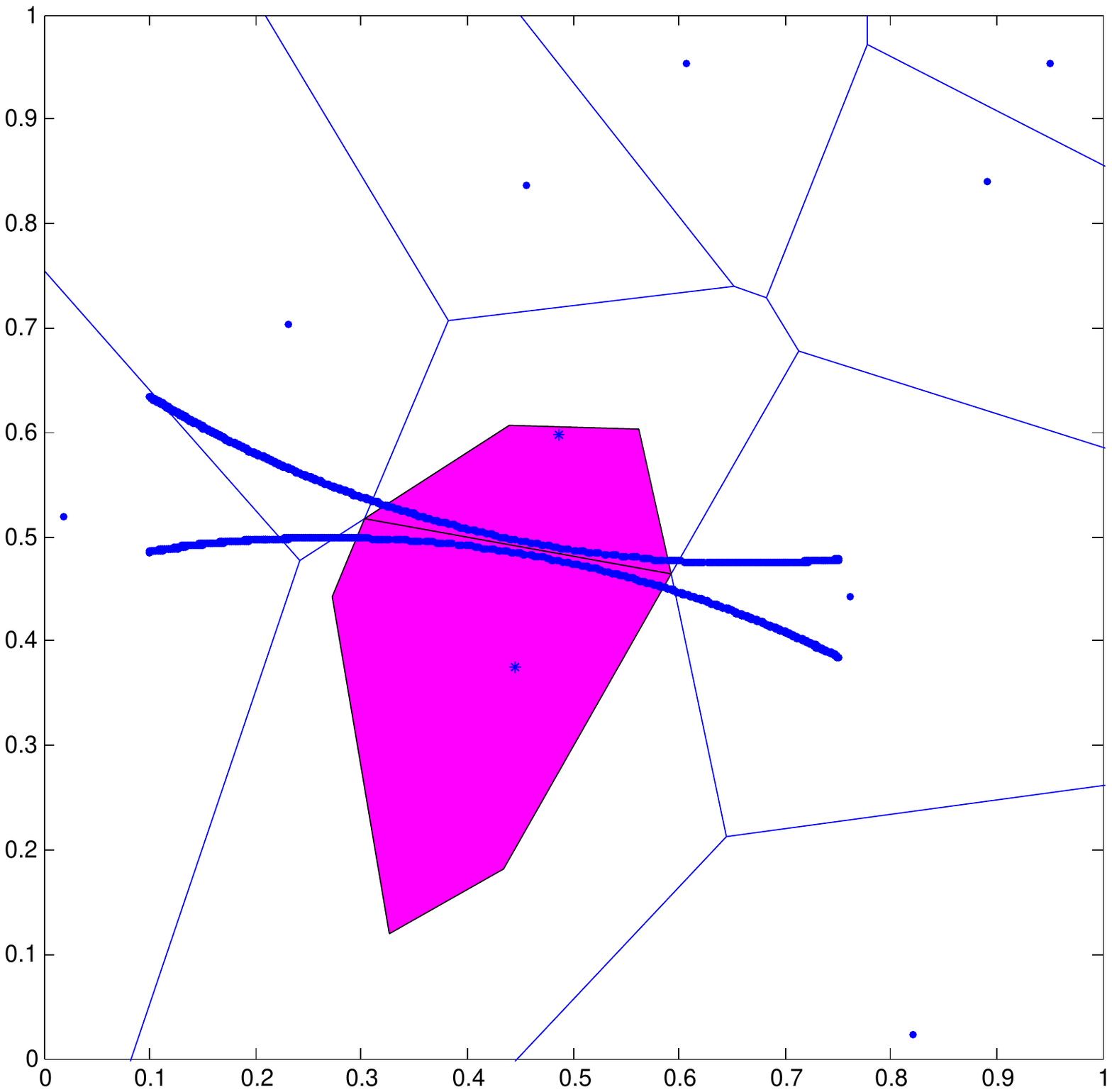}
           \label{fig:ExampleCoop09}
           }
           \caption{Cooperation regions between two neighbouring BSs, in an example topology of 10 uniformly scattered atoms and two different values of global $\rho$.}
\end{figure}

%


\section{Stochastic Geometry and Coverage}
\label{SectionIII}

We formulate the problem within the framework of Stochastic Geometry \cite{BaccelliBookStoch}. We focus on a random location on the plane and set its coordinates as 
the origin $\left(0,0\right)$. We assume that a user is positioned at this point, whom we denote by $\mathbf{u}_o$. The two closest BSs are $\mathbf{b}_1:=\mathbf{b}_{o1}$, $\mathbf{b}_2:=\mathbf{b}_{o2}$ with distances $r_1=r_{o1}$, $r_2:=r_{o2}$ from the typical location and 
gains $h_{1}:= h_{o1} = g_{1} r_{1}^{-\beta}$, $h_{2}:=h_{o2} = g_{2} r_{2}^{-\beta}$, 
$h_{n1}:=h_{n1,o} = g_{n1} d_{n1}^{-\beta}$ and $h_{n2}:=h_{n2,o} = g_{n2} d_{n2}^{-\beta}$.
Furthermore, $g_1$, $g_2$, $g_{n1}$ and $g_{n2}$ are realizations of independent exponential r.v.'s with mean $p$ (the individual per-user power) $G_{ni},G_i\sim \exp\left(1/p\right)$, $i=1,2$. Consequently, $\sqrt{G}$ follows the Rayleigh distribution.

The aim of this work is to derive the coverage of the cooperating system by calculating the 
coverage probability 

\begin{eqnarray}
q_c\left(T,\lambda,\beta,p,\rho\right):=\mathbb{P}\left[\mathrm{SINR}>T\right],
\label{CovDef}
\end{eqnarray}
where the $\mathrm{SINR}$ is measured at a \textbf{typical location} of the plane. Notice that it would be more natural to look at the SINR at the location of a \textbf{typical user} and that the two definitions do not coincide here because the point process of BSs and that of users are not independent. The coverage is a function of the threshold $T$, the p.p. intensity $\lambda$, the path-loss exponent $\beta$, the user power $p$ and the policy parameter $\rho$. 
From here on, the parameter set $\left\{T, \lambda, \beta, p\right\}$, which does not influence the analysis, is omitted 
from the arguments. The problem position extends the relevant work of Andrews et al. in \cite{AndrewsCoverage} for cellular networks. The $\mathrm{SINR}$ for the typical location can now be written using the policies in Definition \ref{PolicyGeomDef} and the expressions (\ref{NoC}) and (\ref{FC}).

\begin{eqnarray}
\mathrm{SINR}\left(\rho,r_1,r_2\right) & = & \frac{g_1 r_1^{-\beta}}{\sigma^2+ \mathcal{I}\left(\rho,r_2\right)} \mathbbm{1}_{\left\{r_1\leq \rho r_2\right\}} + \frac{\frac{\left(\sqrt{g_1 r_1^{-\beta}}+\sqrt{g_2 r_2^{-\beta}}\right)^2}{2}}{\sigma^2 + \mathcal{I}\left(\rho,r_2\right)} \mathbbm{1}_{\left\{r_1> \rho r_2\ \&\ r_1\leq r_2\right\}}
\label{SINRo2}
\end{eqnarray}

The interference $\mathcal{I}\left(\rho,r_2\right)$ takes into consideration all the power splitting decisions of primary users related to BSs with distance $d_{n1}, d_{n2}\geq r_2$ from the origin - the reader can refer to (\ref{SINRui2a3}). The decisions are determined by the value of $\rho$.

\subsection{Distribution of distance to the two closest neighbours}

For the stochastic geometry analysis, it will be useful to derive the probability density function (p.d.f.) of the 
r.v.'s of the distances from the typical location to the first and second closest BS neighbour. The proofs of the lemmas that follow together with supplementary material are provided in Appendix F.

\begin{Lem}
Given a p.p.p. of intensity $\lambda$, the joint p.d.f. of the distances $\left(r_1,r_2\right)$ between the typical location $\mathbf{u}_o$ and its first and second closest neighbour, is equal to
\begin{eqnarray}
f_{r_1,r_2}\left(r_1,r_2\right) & = & \left(2\lambda\pi\right)^2 r_1 r_2e^{-\lambda \pi r_2^2}.
\label{pdfN2f}
\end{eqnarray}
The expected value of the distance $r_2$ is equal to
\begin{eqnarray}
\mathbb{E}\left[r_2\right] & = & \frac{3}{4\sqrt{\lambda}}.
\label{ExpR2}
\end{eqnarray}
\label{LemJointD}
\end{Lem}

An important lemma below, provides the expression for the probability for a user not 
to demand for cooperation, based on the region $r_1\leq \rho r_2$. The interesting observation is that the parameter $\rho\in\left[0,1\right]$ 
fully determines this probability.

\begin{Lem}
\label{Lemrho}
The probability of a user on the 2D plane, not to demand for BS cooperation for its service is
\begin{eqnarray}
\mathbb{P}\left[\mathrm{No\ Coop}\right] = \mathbb{P}\left[r_1\leq \rho r_2\right] = \rho^2.
\end{eqnarray}
\end{Lem}


\subsection{Cooperative Channel Fading Distribution and Properties}

In the current subsection we will derive a general expression for the distribution 
of the signal power when full cooperation is applied. Furthermore, certain properties of this 
distribution will be provided. The proofs of the results in this paragraph are provided in Appendix G.

\begin{Lem}
\label{LemZ}
Given the r.v.'s $G_i\sim \exp\left(1/p_i\right)$, $i=1,2$, the Laplace Transform (LT) of the r.v. 
\begin{eqnarray}
\label{RVz}
Z_{r_1,r_2} := \left(\sqrt{G_1 r_1^{-\beta}}+\sqrt{G_2 r_2^{-\beta}}\right)^2
\end{eqnarray}
is equal to
\begin{eqnarray}
\label{LTz}
\mathcal{L}_Z\left(s,\mu_1,\mu_2\right) = \frac{-s\sqrt{\frac{1}{\mu_1\mu_2}}\pi +s\sqrt{\frac{1}{\mu_1\mu_2}} \arctan\left(\sqrt{\frac{\mu_1}{\mu_2}}g\left(s\right)\right)+s\sqrt{\frac{1}{\mu_1\mu_2}} \arctan\left(\sqrt{\frac{\mu_2}{\mu_1}}g\left(s\right)\right)+ g\left(s\right)}{g\left(s\right)^3},
\end{eqnarray}
where $g\left(s\right) := \sqrt{1+\left(\frac{1}{\mu_1}+\frac{1}{\mu_2}\right) s}$ and $\mu_i:=\frac{r_i^{\beta}}{p_i}$ for $i=1,2$.

Furthermore, its p.d.f. is square integrable and its expectation is equal to
\begin{eqnarray}
\label{EZpdf}
\mathbb{E}\left[Z_{r_1,r_2}\right] & = & \sqrt{\frac{1}{\mu_1\mu_2}}\left(\frac{\pi}{2} + \frac{\mu_1+\mu_2}{\sqrt{\mu_1\mu_2}} \right).
\end{eqnarray}
\end{Lem}

We further provide an interesting property of the r.v. $\frac{Z_{r,r}}{2}$ of the fading for $\mathrm{Full\ Coop}$, in the special case of $\mu_1=\mu_2=\mu:=\frac{r^{\beta}}{p}$ (the expression given in (\ref{RVz})), with relation to the r.v. $G\sim\exp{\left(\mu\right)}$, which is the case for no cooperation. We use the notion of \textit{stochastic ordering} \cite[A4. p.411]{AsmBook} based on which the r.v. $Y$ stochastically dominates the r.v. $X$ and we write $X\leq_{st}Y$, if $\mathbb{P}\left[X>t\right]\leq \mathbb{P}\left[Y>t\right]$ for all $t$.

\begin{Lem}
\label{StochOrd1}
Given $\mu_1=\mu_2=\mu:=\frac{r^{\beta}}{p}$, and the two r.v.'s $G\sim\exp\left(\mu\right)$ and $Z_{r,r}/2$ from (\ref{RVz}), the following stochastic ordering inequality holds
\begin{eqnarray}
\label{SOineq}
\begin{tabular}{l l l l}
If $\mu< 1$ & $\Rightarrow$ & $G\leq_{st}\frac{Z_{r,r}}{2}$. & 
\end{tabular}
\end{eqnarray}
\end{Lem}

We further use the notion of the \textit{Laplace-Stieltjes transform ordering} \cite{LaplaceOrd91} based on which, the r.v. $Y$ dominates the r.v. $X$ and we write $X\leq_L Y$, if 
$\mathcal{L}_X\left(s\right)\geq \mathcal{L}_Y\left(s\right)$, for all $s\geq 0$.

\begin{Lem}
\label{LaplaceOrd}
Given $\mu_1=\mu_2=\mu:=\frac{r^{\beta}}{p}$, and the two r.v.'s $G\sim\exp\left(\mu\right)$ and $Z_{r,r}/2$ from (\ref{RVz}), the following Laplace-Stieltjes transform ordering inequality holds
\begin{eqnarray}
\label{LaplOrdIneq}
G\leq_L \frac{Z_{r,r}}{2}.
\end{eqnarray}
\end{Lem}

\subsection{Interference as Shot Noise}
\label{ParagInterference}
In the current subsection, we provide an expression for the random variable of the interference experienced at the typical location, together with its Laplace transform. To do this we describe the interference as a shot-noise field \cite[Ch.2]{BaccelliBookStoch} 
generated by a point process outside a ball of radius $r_2$. The details of the analysis together with proofs and extra material is available in Appendix H. More specifically, given a realization of the point process, 
the interference received at the origin is equal to the expression in (\ref{SINRui2a3})

\begin{eqnarray}
\label{GenINtera}
\mathcal{I}\left(\mathbf{a}_{-o},r_2\right) & := & \sum_{n:\mathbf{u}_n\in\mathcal{V}_1\left(\mathbf{z}_n\right),\mathbf{z}_n\in\phi\setminus\left\{\mathbf{b}_1\right\}}h_{n1}\left(1-a_n\right) + h_{n2}a_n.\nonumber
\end{eqnarray}
The downlink signals to all users other than the tagged one are considered as interference. So, the summation above is done over all primary users, related to some BS other than the first closest BS to the origin. Furthermore, for each user $\mathbf{u}_n$, the choice of the cooperation parameter $a_n$ splits the signal power $P$ between the two closest BSs.

As already mentioned, the user-optimal decision policy suggested is binary ($\mathrm{No\ Coop}$ or $\mathrm{Full\ Coop}$) for the entire set of users. 
The geometry is illustrated in Fig. \ref{fig:Intref}. If we substitute for $a_n=0$ and $a_n=\frac{1}{2}$ respectively for these two actions, the interference takes the expression

\begin{figure}[h]
\centering
\includegraphics[trim = 10mm 100mm 10mm 10mm, clip, width=0.55\textwidth]{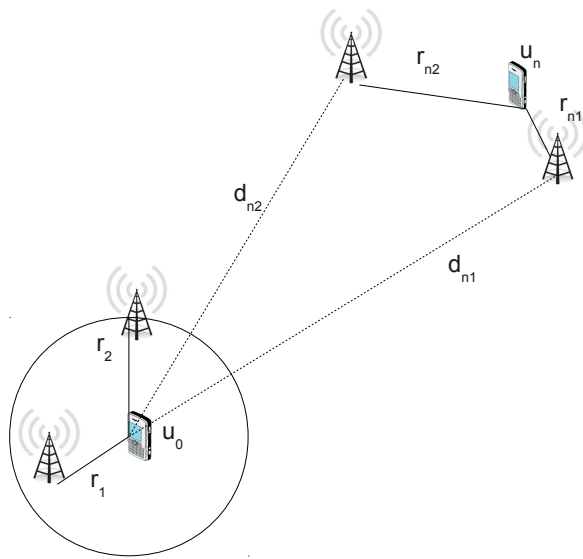}
\caption{Illustration of the interference model with a single interfering pair lying outside the ball of radius $r_2$.}
\label{fig:Intref}
\end{figure}

\begin{eqnarray}
\label{INtrfRHO}
\mathcal{I}\left(\rho,r_2\right) & = & \sum_{n:\mathbf{u}_n\in\mathcal{V}_1\left(\mathbf{z}_n\right),\mathbf{z}_n\in\phi\setminus\left\{\mathbf{b}_1\right\}}h_{n1} \mathbbm{1}_{\left\{r_{n1}\leq \rho r_{n2}\right\}} + \frac{h_{n1}+h_{n2}}{2} \mathbbm{1}_{\left\{r_{n2}\geq r_{n1}>\rho r_{n2}\right\}}.\nonumber
\end{eqnarray}

To derive a more tractable expression for the interference, which takes into account the random positioning of the users, suppose that each atom $\mathbf{z}_n$
of the BS point process has a primary user $\mathbf{u}_n$ who, with probability $\mathbb{P}\left[r_{n1}\leq \rho r_{n2}\right]$ lies within the $\mathrm{No\ Coop}$ region. From Lemma \ref{Lemrho} this probability is equal to $\rho^2$, whereas with probability $1-\rho^2$ the user requests for $\mathrm{Full\ Coop}$. Then we consider that each BS is associated with a \textbf{binary} r.v. $B_n\sim \mathrm{Bernoulli(\rho^2)}$ such that

\begin{eqnarray}
\label{Bernou}
B_n & = & \left\{
\begin{tabular}{l l l}
$1$ & with prob. $\rho^2$ & ($\mathrm{No\ Coop}$)\\
$0$ & with prob. $1-\rho^2$ & ($\mathrm{Full\ Coop}$)
\end{tabular}
\right..
\end{eqnarray}
Each BS is related to an independent $B_n$. Then:

\begin{itemize}
\item If $B_n=1$ an independent mark $\mathcal{M}_n$ is associated with the BS. The mark is equal to $\mathcal{M}_n:= d_{n}^{-\beta}G_{n}$, and this signal is treated as interference at the typical location, where $G_n\sim\exp\left(1/p\right)$. The signal has to traverse a distance of $d_n=d_{n1}$ from the closest neighbour of user $\mathbf{u}_n$, as shown in Fig. \ref{fig:Intref}.
\item If $B_n=0$ an independent mark $\mathcal{N}_n$ is associated with the BS. This is the case of $\mathrm{Full\ Coop}$, where the interfering signal due to user $\mathbf{u}_n$ is jointly emitted from its two closest neighbours. Here, we make the \textbf{far-field approximation} $d_{n1}\approx d_{n2} = d_n$, where the distances of the two cooperating atoms to the typical location are treated as equal. Such a heuristic is allowed, since the two closest neighbours of user $\mathbf{u}_n$ are expected to 
lie ¨close¨ to each other and at the same time, outside of the ball of radius $r_2$. In this sense we can expect that the error is very small. Based on this approximation, BSs with primary users requiring $\mathrm{Full\ Coop}$, are considered to emit the entire signal $\mathcal{N}_n := d_n^{-\beta}\frac{G_{n1}+G_{n2}}{2}$, $G_{n1},G_{n2}\sim\exp\left(1/p\right)$.

\end{itemize}
%
It is interesting to notice that, the r.v. $G_n$, related to the mark $\mathcal{M}_n$ follows the exponential distribution, or equivalently the $\Gamma\left(1,p\right)$ distribution, with expected value $p$, whereas the r.v. $\frac{G_{n1}+G_{n2}}{2}$ related to $\mathcal{M}_n$, follows the $\Gamma\left(2,\frac{p}{2}\right)$, with the same expected value $p$. In other words, the path-loss of the interfering signals is in expectation equal in both cases (of either $\mathrm{Full\ Coop}$ or $\mathrm{No\ Coop}$ for the interfering user). In this far-field approximation, the interference random variable takes the expression

\begin{eqnarray}
\label{Interference1}
\mathcal{I}\left(\rho,r_2\right) & := & r_{2}^{-\beta} G_{2}B_2 + r_2^{-\beta}\frac{G_1+G_2}{2} \left(1-B_2\right) \nonumber\\
& + & \sum_{\mathbf{z}_n\in\phi\setminus\left\{\mathbf{b}_1,\mathbf{b}_2\right\}} d_{n}^{-\beta}G_{n} B_n + d_n^{-\beta}\frac{G_{n1}+G_{n2}}{2} \left(1-B_n\right),
\end{eqnarray}
where the first two terms come from the interference created by the second neighbour lying on the 
boundary of the ball $\mathcal{B}\left(\mathbf{u}_o,r_2\right)$. Finally we derive the LT of the Interference variable $\mathcal{L}_{\mathcal{I}}\left(s,\rho,r_2\right)$.

\begin{Theorem}
\label{THlaplI}
The LT of the Interference random variable for the model under study, with exponential fading channel power (Rayleigh fading), is equal to
\begin{eqnarray}
\label{LTinterfD}
\mathcal{L}_{\mathcal{I}}\left(s,\rho,r_2\right) & = & \mathcal{L}_{\mathcal{J}}\left(s,\rho,r_2\right)\cdot e^{-2\pi\lambda\int_{r_2}^{\infty} \!  \left(1-\mathcal{L}_{\mathcal{J}}\left(s,\rho,r\right)\right)r \, \mathrm{d} r },
\end{eqnarray}
where
\begin{eqnarray}
\label{EachID}
\mathcal{L}_{\mathcal{J}}\left(s,\rho,r\right)& = & \rho^2\frac{1}{1+sd_n^{-\beta}p}+\left(1-\rho^2\right)\frac{1}{\left(1+sd_n^{-\beta}\frac{p}{2}\right)^2}.
\end{eqnarray}
\end{Theorem}
It can be shown that the LT of the Interference random variable $\mathcal{L}_{\mathcal{I}}\left(s,\rho,r_2\right)$ in (\ref{LTinterfD}) is a non-decreasing function in $\rho$ and $r_2$ and non-increasing function in $s$. A direct consequence of this is that the LT of the interference takes its maximal value for $\rho=1$. This is reasonable based on the Laplace-Stieltjes transform ordering. The argument is that for any $0\leq\rho_a<\rho_b\leq 1$, $\mathcal{L}_{\mathcal{I}}\left(s,\rho_a,r_2\right)\leq \mathcal{L}_{\mathcal{I}}\left(s,\rho_b,r_2\right)$ $\Rightarrow$ $\mathcal{I}\left(\rho_a,r_2\right)\geq_L \mathcal{I}\left(\rho_b,r_2\right)\geq_L\mathcal{I}\left(1,r_2\right)$. The larger $\rho$ the smaller the interference. Similarly, the larger the distance to the second neighbour $r_2$, the smaller the interference, because a larger empty ball of interferers around the typical location is guaranteed. Finally, we give the expression of the expected value for the Interference r.v.

\begin{eqnarray}
\label{EIgeneral}
\mathbb{E}\left[\mathcal{I}\left(\rho,r_2,\beta,p,\lambda \right)\right] & = & \frac{p}{\left(\beta-2\right)r_2^{\beta}}\left(\beta-2+2\pi\lambda r_2^2\right).
\end{eqnarray} 
The expected value of the interference r.v. is \textbf{independent of the parameter $\rho$}. This means that all scenarios of cooperation regions for any $\rho\in\left[0,1\right]$ are compared to each other, with reference to an interference field with the same expected value given in (\ref{EIgeneral}).

\subsection{Second Neighbour Interference Elimination - Dirty Paper Coding}

The second geographic BS neighbour is the most influential on the interference power, due to 
its proximity to the typical location. In this subsection we will consider an ideal scenario, where the interference created by the second closest BS can be cancelled out perfectly in the case of full cooperation, by means of coding. This requires precise \textbf{knowledge} by the first neighbour of the interfering signal from the primary user and all possible secondary users served by $\mathbf{b}_2$, which is extra information for the system. If such information is available, the encoding procedure for the signal related to the typical location can be projected on the orthogonal signal space of $\mathbf{b}_2$, achievable by Dirty Paper Coding so that the effect of $\mathbf{b}_2$ on interference is eliminated (see \cite{CostaDP83} - however notice that we do not refer here to the Zero Forcing and Dirty Paper coding scheme in \cite{CaireDPC03} because such scheme would require full channel state information feedback). If the elimination is perfect, the $\mathrm{SINR}$ for the typical location should be rewritten as follows

\begin{eqnarray}
\mathrm{SINR}\left(\rho,r_1,r_2\right) & = & \frac{g_1 r_1^{-\beta}}{\sigma^2+ \mathcal{I}\left(\rho,r_2\right)} \mathbbm{1}_{\left\{r_1\leq \rho r_2\right\}} + \frac{\frac{\left(\sqrt{g_1 r_1^{-\beta}}+\sqrt{g_2 r_2^{-\beta}}\right)^2}{2}}{\sigma^2 + \mathcal{I}_{DPC}\left(\rho,r_2\right)} \mathbbm{1}_{\left\{r_1> \rho r_2\ \&\ r_1\leq r_2\right\}},\nonumber
\label{SINRo2a}
\end{eqnarray}
where the interference in the case of full cooperation has been substituted by the r.v. $\mathcal{I}_{DPC}$. The latter is derived exactly as the r.v. $\mathcal{I}$ by just omitting the interference part from the second closest BS neighbour. Its expression and Laplace Transform are thus

\begin{eqnarray}
\label{Interference1}
\mathcal{I}_{DPC}\left(\rho,r_2\right) & := & \sum_{\mathbf{z}_n\in\phi\setminus\left\{\mathbf{b}_1,\mathbf{b}_2\right\}} d_{n}^{-\beta}G_{n} B_n + d_n^{-\beta}\frac{G_{n1}+G_{n2}}{2} \left(1-B_n\right),\nonumber
\end{eqnarray}
and
\begin{eqnarray}
\label{LTinterfDcancel}
\mathcal{L}_{\mathcal{I}_{DPC}}\left(s,\rho,r_2\right) & = &  e^{-2\pi\lambda\int_{r_2}^{\infty} \!  \left(1-\mathcal{L}_{\mathcal{J}}\left(s,\rho,r\right)\right)r \, \mathrm{d} r },\nonumber
\end{eqnarray}
where $\mathcal{L}_{\mathcal{J}}\left(s,\rho,r\right)$ is given in (\ref{EachID}).  Notice that if the user chooses $\mathrm{No\ Coop}$, the elimination is not possible. The expected value of the interference without the influence of $\mathbf{b}_2$ is smaller than the expression in (\ref{EIgeneral}) by $\frac{\beta-2+2\pi\lambda r_2^2}{2\pi\lambda r_2^2}>1$, since $\beta>2$, as the following formula shows

\begin{eqnarray}
\label{EIgeneralcan}
\mathbb{E}\left[\mathcal{I}_{DPC}\left(\rho,r_2,\beta,p,\lambda \right)\right] & = & \frac{p}{\left(\beta-2\right)r_2^{\beta}}2\pi\lambda r_2^2.
\end{eqnarray} 
The expected value of DPC interference is as well independent of the cooperation parameter $\rho$.

\subsection{General Coverage Probability}

In the current subsection we derive the coverage probability expression for the model under study. The proof of the related theorem that follows can be found in Appendix I.

\begin{Theorem}
\label{CovProb}
For the cooperation scenario under study and for a given set of system values $\left\{T,\lambda,\beta,P\right\}$, the coverage probability at the typical location as a function of the parameter $\rho\in\left[0,1\right]$ is equal to
\begin{eqnarray}
q_c\left(\rho\right) 	& = & 	q_{c,1}\left(\rho\right) + q_{c,2}\left(\rho\right)\nonumber\\
\label{Tint1}
						& = & 	\int_0^{\infty} \! \int_{\frac{r_1}{\rho}}^{\infty} \!  \left(2\lambda\pi\right)^2 r_1 r_2e^{-\lambda \pi r_2^2} \cdot e^{-\frac{r_1^{\beta}}{p} T\sigma^2}\mathcal{L}_{\mathcal{I}}\left(\frac{r_1^{\beta}}{p} T,\rho,r_2\right)\,  \mathrm{d} r_2 \,  \mathrm{d} r_1\\
\label{Tint2}
						& + & 	\int_0^{\infty} \! \int_{r_1}^{\frac{r_1}{\rho}} \!  \left(2\lambda\pi\right)^2 r_1 r_2e^{-\lambda \pi r_2^2} \int_{-\infty}^{\infty} \! e^{-2j\pi\sigma^2 s} \mathcal{L}_{\mathcal{I}}\left(2j\pi s,\rho,r_2\right)\frac{\mathcal{L}_{Z}\left(-j\pi s/T,\frac{r_1^{\beta}}{p},\frac{r_2^{\beta}}{p}\right)-1}{2j\pi s}\,  \mathrm{d} s\,  \mathrm{d} r_2 \,  \mathrm{d} r_1\nonumber\\
\end{eqnarray}
In this expression, $\mathcal{L}_{\mathcal{I}}\left(s,\rho,r_2\right)$ is the LT of $\mathcal{I}$ given in (\ref{LTinterfD}), (\ref{EachID}) and $\mathcal{L}_Z\left(s,\mu_1,\mu_2\right)$ is the LT of the general fading r.v. $Z_{r_1,r_2}$ given in (\ref{LTz}). For the case of Dirty Paper Coding, $\mathcal{L}_{\mathcal{I}}$, should be replaced by $\mathcal{L}_{\mathcal{I}_{DPC}}$ given in (\ref{LTinterfDcancel}).
\end{Theorem}


\section{Pros and Cons of BS-Pair Conferencing}
\label{SectionIV}

%
%
%


The coverage probability expression in (\ref{Tint1})+(\ref{Tint2}) can provide intuition on the change in coverage probability when applying different $\rho$-dependent 2-cell cooperation policies. In the following we summarize the pros and cons of cooperation, given the model of the previous sections. We would like to emphasize here, that all benefits result from pairwise cooperation with use of limited feedback exchange between the transmitters (only channel phase).

\textbf{Pros:} The gains of the cooperation scheme result from four reasons.
\begin{itemize}
\item The specific coding scheme enabled by conferencing, provides an extra term of $2a_i p\sqrt{h_{i1}h_{i2}}$ $\stackrel{a_i=1/2}{=}$ $p\sqrt{h_{i1}h_{i2}}$ for the beneficial signal in (\ref{SINRui2a2}). Given that the channel fast fading r.v.'s follow exponential distribution, and for the symmetric case $r_1=r_2=r$ the resulting $\mathrm{Full\ Coop}$ variable $Z_{r,r}/2$ is larger than $G$ in the stochastic ordering, as shown in Lemmas \ref{StochOrd1} and \ref{LaplaceOrd}. From an engineering perspective, the user transmission power is equally divided between the two closest BS neighbours, who transmit the same signal with power $p/2$. Since the signals add-up coherently at the receiver at the typical location, this can provide an increase in the received signal power, but depends also on the distances of the two BSs from the user. Such benefit can be achieved only when the two cooperating BSs transmit exactly the same data in a synchronous way to the receiver. The codeword of the second neighbour is linearly dependent to the codeword of the first neighbour.

Evaluating the expression in (\ref{EZpdf}) for $p=1$ and $r_1=r_2=1$, we get an expected value $\frac{\mathbb{E}\left[Z_{r,r}\right]}{2}=1+\pi/4\approx 1.8>1=\mathbb{E}\left[G\right]$. However, the opposite effect will occur when the second neighbour is very far from the origin ($r_2>r_1$), due to the path-loss exponent $\beta>2$. In such a case, the power split between the two neighbours will result in a received signal less than that of $\mathrm{No\ Coop}$. 

\item The knowledge of the second closest neighbour exact position (related to the typical location), guarantees an open ball of radius $r_2>r_1$ to be interference free.
The expected radius $\mathbb{E}\left[r_2\right]=\frac{3}{4\sqrt{\lambda}}$ is strictly larger than $\mathbb{E}\left[r_1\right]=\frac{1}{2\sqrt{\lambda}}$ so that the area of the open ball is guaranteed by Lemmas \ref{Lemr1} and \ref{LemJointD} to be larger in expectation than the one defined from just the first closest BS where the typical user is assigned. This means that there will be a factor of improvement on interference, related to the case of no cooperation.

\item Cooperation is in favor of the cell-edge areas and offers coverage, when $\mathrm{No\ Coop}$ fails to do so. Fig. \ref{fig:CovEXE1a} illustrates the coverage areas for an instance of BS positions and when no channel fading is taken into account (so that smooth contours are visible). Observe furthermore, the obvious increase of the coverage lobe in all cells when $\rho=0$ and $\mathrm{Full\ Coop}$ everywhere is applied.

The parameter $\rho$, which varies within the interval $\left[0,1\right]$,  can define the shape of the coverage areas. A change of the parameter results in a change of the shape, so that certain planar points not covered when $\rho=1$ can be covered for some other value $\rho<1$ and the other way round. In this way, by adapting the parameter $\rho$, we can shape the coverage of the cell depending on the service needs. 

\item The parameter $\rho$ can be chosen optimally as $\rho^*$, so that the geometric policy maximizes coverage. The optimal value of $\rho$ depends on the set of model parameters $\left(T,\lambda,\beta,p\right)$. The expression in (\ref{Tint1})+(\ref{Tint2}) cannot be shown to be concave in $\rho$, so that existence of local optima cannot be excluded. The global optimum can be found by numerical evaluation, as shown in Fig. \ref{fig:CovT01full} for a specific example. Fig. \ref{fig:Zona06} exhibits the optimal planar cooperation policy for the considered example.

\end{itemize}

\begin{figure}[th]
\centering
\includegraphics[trim = 0mm 40mm 0mm 30mm, clip, width=0.6\textwidth]{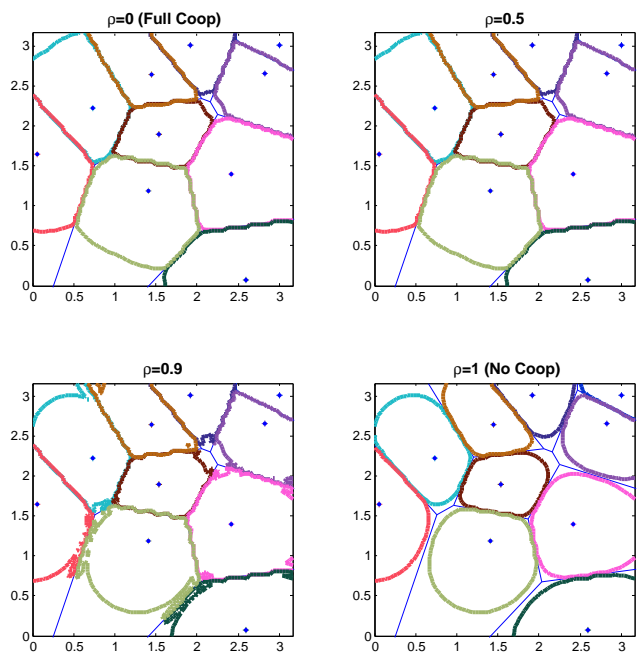}
\caption{Example of coverage regions for different values of the parameter $\rho$ and $T=0.8$. Case without fast fading. For the chosen value of threshlod, the 
maximum coverage is achieved in the case of $\rho=0.5$, but the difference with the case of $\mathrm{Full\ Coop}$ (everywhere) for $\rho=0$ is trivial. For these two values, it can be seen that most areas at the cell edges are covered, whereas the case of $\mathrm{No\ Coop}$ with $\rho=1$ fails to do so. In general, intermediate values of the cooperation paremeter, e.g. for $\rho=0.5$ or $\rho=0.9$ as shown in the above plots, allow a combination of policies $\mathrm{Full\ Coop}$ and $\mathrm{No\ Coop}$ on the plane. In this way, a higher coverage benefit can be achieved, by serving points closer to base stations without cooperation and transmitting cooperatively for the points at the cell borders.}
\label{fig:CovEXE1a}
\end{figure}


\begin{figure}[ht]    
\centering  
\label{fig:VoronoiREG}
	 	\subfigure[Probability of coverage for different values of $\rho$ and $T=0.8$. The optimal value $\rho^*=0.5$.]{          
           \includegraphics[trim = 0mm 40mm 0mm 0mm, clip, width=0.35\textwidth]{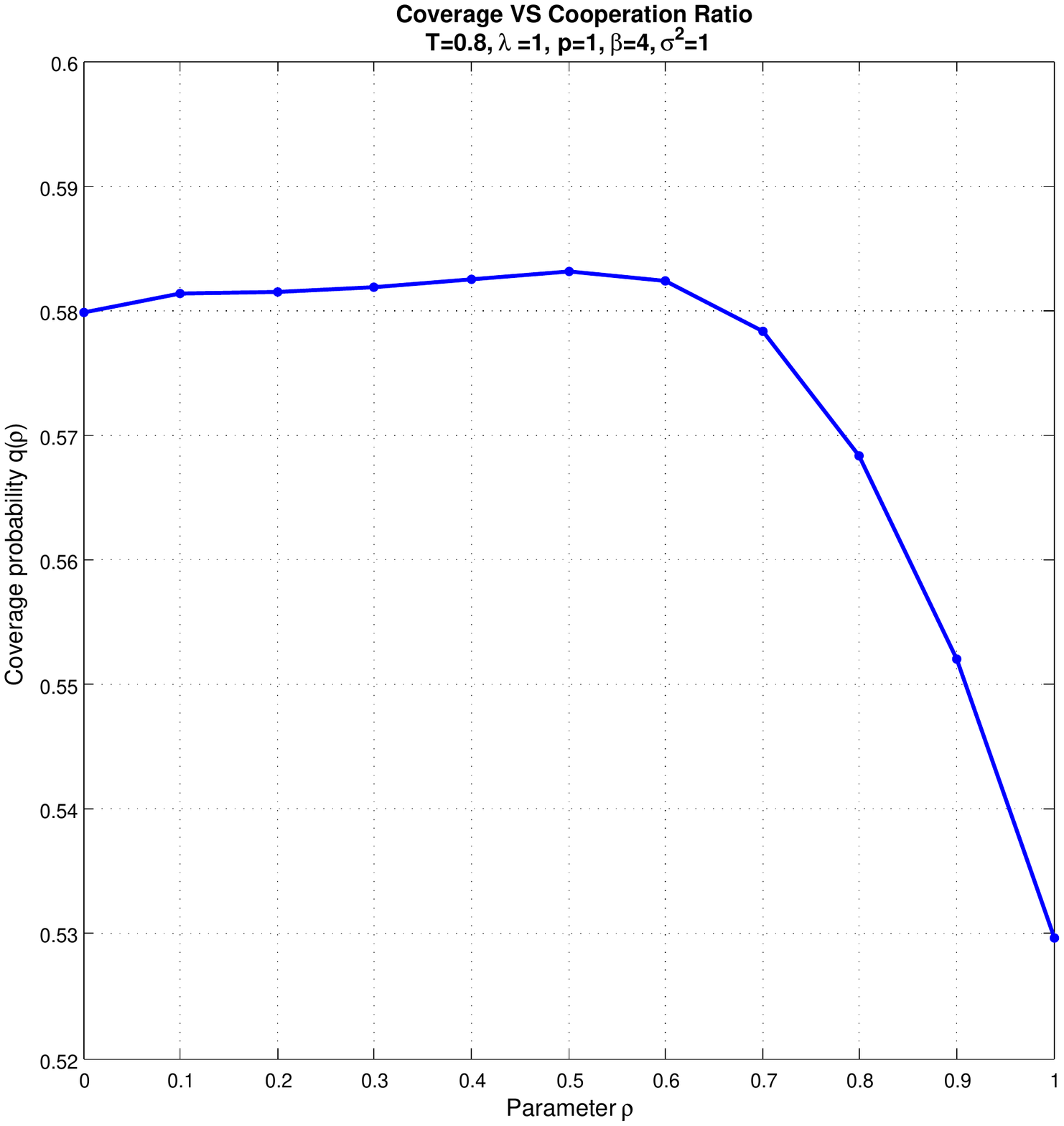}
           \label{fig:CovT01full}
           }
            \subfigure[Cooperation regions for $T=0.8$ and $\rho^*= 0.5$.]{          
           \includegraphics[trim = 0mm 50mm 0mm 0mm, clip, width=0.4\textwidth]{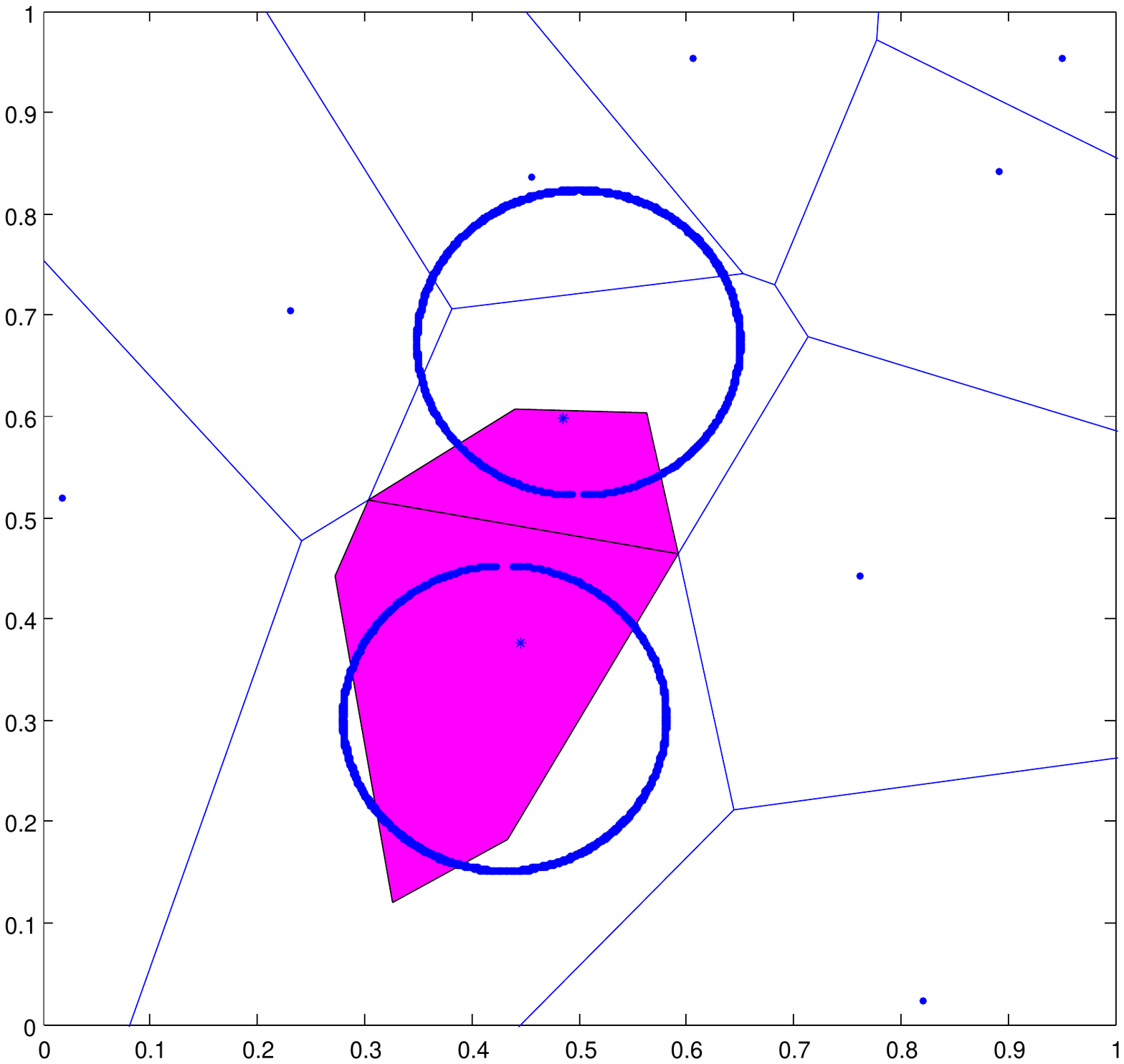}
           \label{fig:Zona06}
           }
           \caption{(a) Evaluation of the integrals in (\ref{Tint1})+(\ref{Tint2}) with threshold $T=0.8$ and varying the cooperation parameter $\rho$. It can be seen that the average coverage probability reaches a maximum for $\rho^*= 0.5$, although the maximum achievable coverage is almost the same, for the case $\rho=0$ related to $\mathrm{Full\ Coop}$ everywhere. Hence, the benefits of cooperation can be achieved, even when cooperation is applied to a smaller percentage of the served users. The resulting optimal $\mathrm{Full\ Coop}$ regions are shown in (b) which include the locus of planar points inside the 2-Voronoi cell but outside the discs.}
\end{figure}

%
%
%

\textbf{Cons:} The negative effects of the cooperation scheme result from two reasons.
\begin{itemize}
\item The parameter $\rho$ of the geometric policy may enforce a user at a planar point $\mathbf{z}$ to ask for cooperation, even when $g_1r_1^{-\beta}>\frac{\left(\sqrt{g_1r_1^{-\beta}}+\sqrt{g_2r_2^{-\beta}}\right)^2}{2}$. The reason is that geometric policies do not adapt to the actual fast-fading realization and depend only on the ratio $\frac{r_1}{r_2}$. Since $\rho$ is a design parameter, the optimal $\rho^*$ is expected to adapt the policy to such events.

\item By comparing the tail probabilities of the $\mathrm{Full\ Coop}$ random variable $\frac{Z_{r_1,r_2}}{2}$ and the exponential $\mathrm{No\ Coop}$ random variable $G$ we can find many examples of the pair $\left(r_1,r_2\right)$, such that there exist values of the signal threshold $T$, above which it holds
\begin{eqnarray}
\label{Tail1}
\mathbb{P}\left[Z_{r_1,r_2}/2>T\right]<\mathbb{P}\left[G>T\right].\nonumber
\end{eqnarray}
In other words, above a certain threshold, the contribution of the exponential $\mathrm{No\ Coop}$ signal power to the coverage probability, may be more considerable than the $\mathrm{Full\ Coop}$ signal power. This is very important if we add also the fact that, as shown in subsection \ref{ParagInterference}, the interference r.v. has an expected value which is independent of the cooperation parameter $\rho$. A consequence is that, there are values of the threshold $T$, above which the coverage probability of the $\mathrm{No\ Coop}$ everywhere scheme is optimal and cooperation turns out to be suboptimal. An example plot over the threshold $T$, with parameter values $p=1$, $r_1^{\beta}=1$ and $r_2^{\beta}=\frac{1}{2}$ is shown in Fig. \ref{fig:TailProb}. It is clearly shown how the tail probability of the $\mathrm{No\ Coop}$ signal power is higher than the $\mathrm{Full\ Coop}$ one, for $T>3$. 

\begin{figure}[h]
\centering
\includegraphics[trim = 0mm 0mm 0mm 0mm, clip, width=0.4\textwidth]{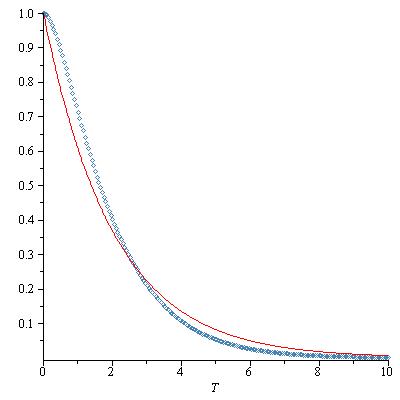}
\caption{Tail probability of the $\mathrm{Full\ Coop}$ and $\mathrm{No\ Coop}$ signal power, with $p=1$, $r_1^{\beta}=1$ and $r_2^{\beta}=\frac{1}{2}$. The figure, resulting by the chosen values of user-power and distances, illustrates that there can be values of the threshold $T$ above which the tail probability of the $\mathrm{No\ Coop}$ signal power is higher than that of the $\mathrm{Full\ Coop}$ signal power.}
\label{fig:TailProb}
\end{figure}

\end{itemize}

\section{Numerical Evaluation and Simulations}
\label{SectionV}

The coverage probability given in Theorem \ref{CovProb} as a sum of the two integrals in (\ref{Tint1}) and (\ref{Tint2}) can be numerically evaluated using MATLAB. 
The evaluation uses routines of nested integration for triple and quadraple integrals. 

To guarantee the validity of the analysis and of the derived results, we have further developed a simulator,
which builds a p.p.p. of intensity $\lambda_{sim} = 1\ \mathrm{atom/m^2}$ for the positions of the BSs, 
within a finite square observation window of surface $W=20\ \mathrm{m^2}$ (size of the x axis $\sqrt{20}\ \mathrm{m}$). For each realization,
a Poisson generator produces a random number of atoms $N$ with expected value $\mathbb{E}\left[N\right]=\lambda_{sim} W = 20\ \mathrm{atoms}$, which are further positioned 
uniformly over the square area. This configuration of points coincides with the distribution of points of the p.p.p. on a bounded 
window as shown in \cite[pp.14-15]{BaccelliBookStoch}. The approximations and analytical evaluation of the integrals is validated by these simulations. A total number of ten thousand uniformly drawn scenarios is averaged to derive the estimation on the simulated coverage probability. 

A first round of comparisons between analytical and simulation results is performed by assuming the 
approximation over the interference - presented in section \ref{ParagInterference}.
In other words, it is assumed in Fig. \ref{fig:SimulB4} ($\beta=4$), Fig. \ref{fig:SimulB2p5} ($\beta=3.2$) and Fig. \ref{fig:SimulB3p2} ($\beta=2.5$) that the distances 
$d_{n1}\approx d_{n2}$ between the typical location and the first and second closest Base Station neighbour of some interfering user are approximately equal. In these figures, the three curves ($\mathrm{No\ Coop}$, $\mathrm{Full\ Coop}$, Optimal Cooperation) derived by numerically evaluating 
the integrals are compared to the three curves from the simulations with the approximation.

The evaluation shows that the analytical results match relatively well the simulated ones. A gap between analytical and simulated curves can be observed. The gap gradually increases for the case of $\beta=3.2$ and $\beta=2.5$ compared to $\beta =4$. This gap occurs due to the fact that the 
expected number of simulated points ($\mathbb{E}\left[N\right]=20$) is low, so that the simulations do not consider the interference created by atoms outside the 
window of size $W=20\ \mathrm{m^2}$. The lower the path-loss exponent $\beta$, the stronger the interfering signal for large distances from the emitting Base Station. 
In this sense, a small simulation window underestimates the total interference 
created. When the size of the window - and consequently the expected number of point process atoms - increases it can be verified that the gap diminishes. 

In Fig. \ref{fig:SimulTHETAcurves} the coverage probability curves using the two different $\mathrm{SINR}$ models (\ref{SINRui2a1th})-(\ref{SINRui2a4th}) and (\ref{SINRui2a1})-(\ref{SINRui2a3}) are compared. The curves are produced by the simulator, in order to evaluate the quality of the approximation, when the expectation over $\theta$ is taken at the interference term. It is evident from the figure that the two models produce almost identical results and hence the approximation is valid for the model under study, since it can simplify the analysis, without considerably affecting the results.

In Fig. \ref{fig:SimulB24noApproax}, we further evaluate, for $\beta=4$, the difference between numerical coverage curves and simulated without this approximation, in order to get a comparison with the coverage of the original model. It can be observed that although the curve for $\mathrm{No\ Coop}$ is the same as the curve with the approximation in Fig. \ref{fig:SimulB4} (as expected) this 
does not hold for the Full and Optimal Cooperation simulation curves. The coverage probability deteriorates faster than the analytical results for higher thresholds, which suggests that  
the approximation underestimates the effect of interference. The main reason is the fact that, in the case $d_{n1}\neq d_{n2}$ for some user $n$ not at the typical location, it can occur in our model, that \textbf{it can be served as secondary user by the BS closest to the typical location ($\mathbf{b}_1$)}. In such a case, its signal creates the highest possible interference at the typical location, since it stems from the closest possible distance, i.e. $d_{n2}=r_1=d\left(\mathbf{b}_{1},\mathbf{u}_0\right)$. Since the first neighbour interference has such a negative effect, it is therefore reasonable to prefer cooperation policies, which diminishes it. This can be achieved by increasing the information at the transmitter side and applying Dirty Paper Coding (DPC) techniques.

\subsection{Coverage $q\left(\rho\right)$ Versus Threshold $T$, $\beta=4$}

In the plots of Fig. \ref{fig:SimulB4} and Fig. \ref{fig:SimulB24noApproax}, the coverage probability is plotted with respect to the threshold value $T$. The threshold $T$ varies between $\left[0.01,\  10\right]$. The absolute difference percentage is shown in Fig. \ref{fig:CovDiffGain} and the optimal value of 
$\rho^*$ for each threshold $T$ in Fig. \ref{fig:OptThresh}. From the above plots the gains from geometric cooperation policies between pairs of neighbouring BSs are already significant. The gain reaches a maximum of over $10\%$ between $T=0.1$ to $0.5$, while the gains disappear for values of threshold $T>2$ and the optimal policy is $\mathrm{No\ Coop}$ in the entire plane. The reason that $\mathrm{Full\ Coop}$ is not optimal for high $T$ can be found in the previous section of Pros and Cons and specifically Fig. \ref{fig:TailProb} which shows that the tail probability of the $\mathrm{No\ Coop}$ signal can be greater than that of the $\mathrm{Full\ Coop}$, from some value of $T$ and above. This depends also on the specific values of $r_1$ and $r_2<r_1$.


\begin{figure}[h!]
\centering
\includegraphics[trim = 0mm 30mm 0mm 20mm, clip, width=0.5\textwidth]{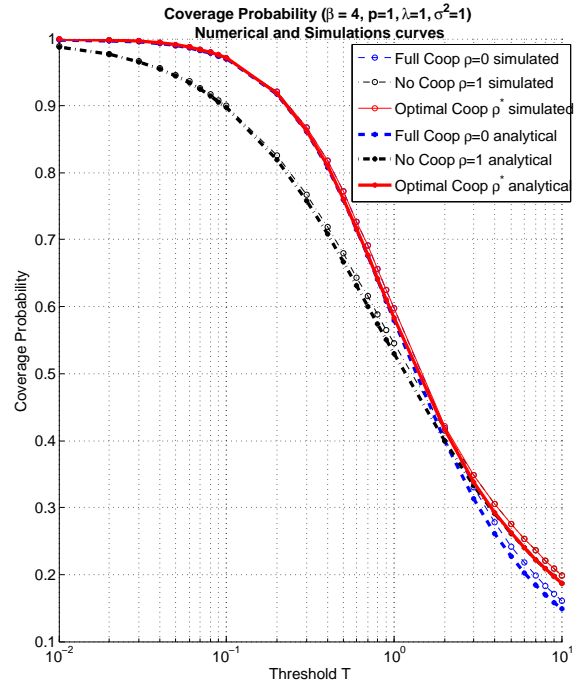}
\caption{Comparison between analytical and simulation results. Expected number of uniformly produced atoms $\mathbb{E}\left[N\right]=20$ in an area of $20\ \mathrm{m^2}$ with density $\lambda=1\ \mathrm{atom/m^2}$. Path-loss exponent $\beta=4$.}
\label{fig:SimulB4}
\end{figure}

\begin{figure}[h!]    
\centering  
\label{fig:B4gain}
	 	\subfigure[Coverage difference between Optimal and $\mathrm{No\ Coop}$.]{          
           \includegraphics[trim = 0mm 28mm 0mm 20mm, clip, width=0.4\textwidth]{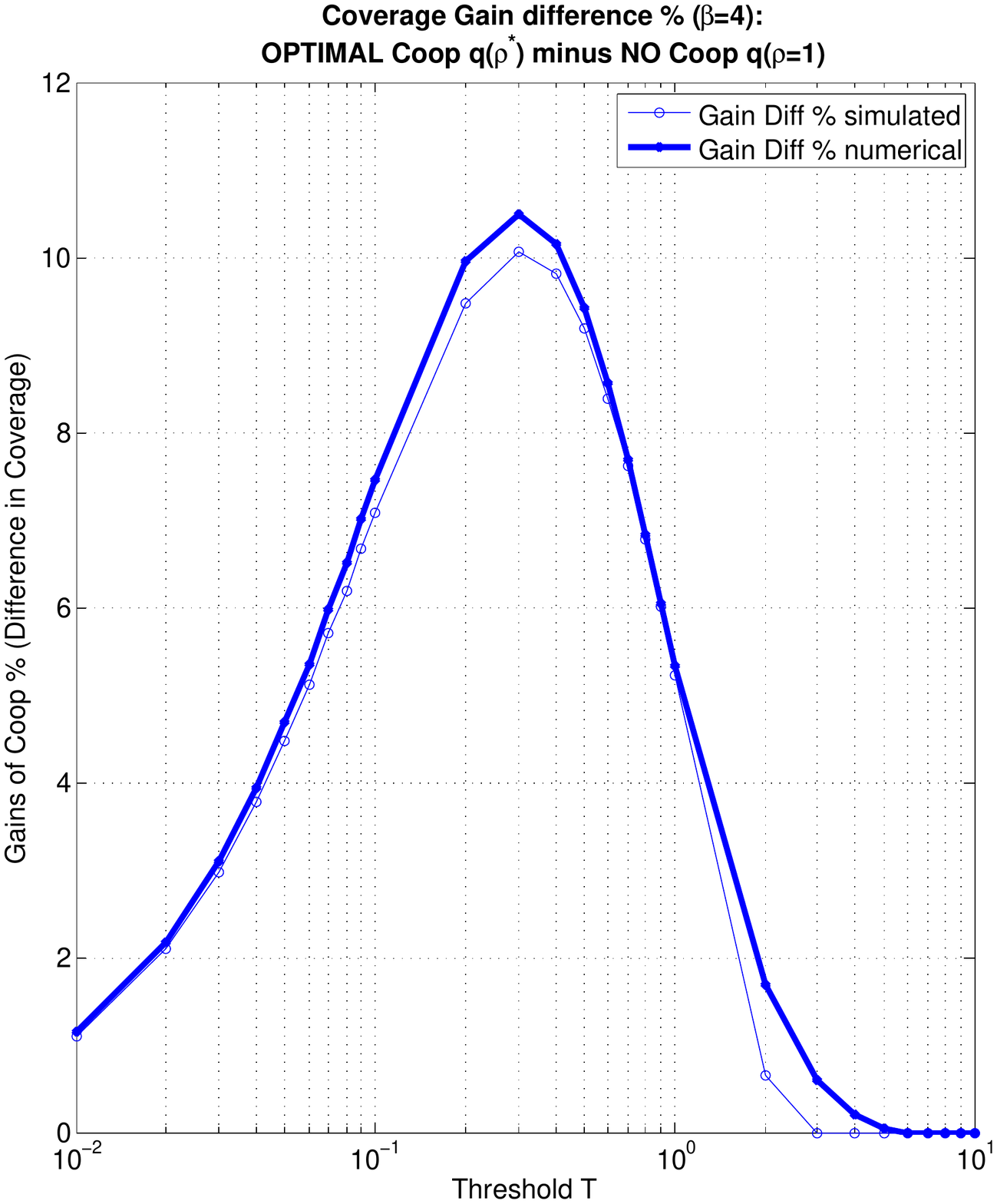}
           \label{fig:CovDiffGain}
           }
            \subfigure[Optimal value of parameter $\rho^*$ versus value $T$.]{          
           \includegraphics[trim = 0mm 30mm 0mm 25mm, clip, width=0.4\textwidth]{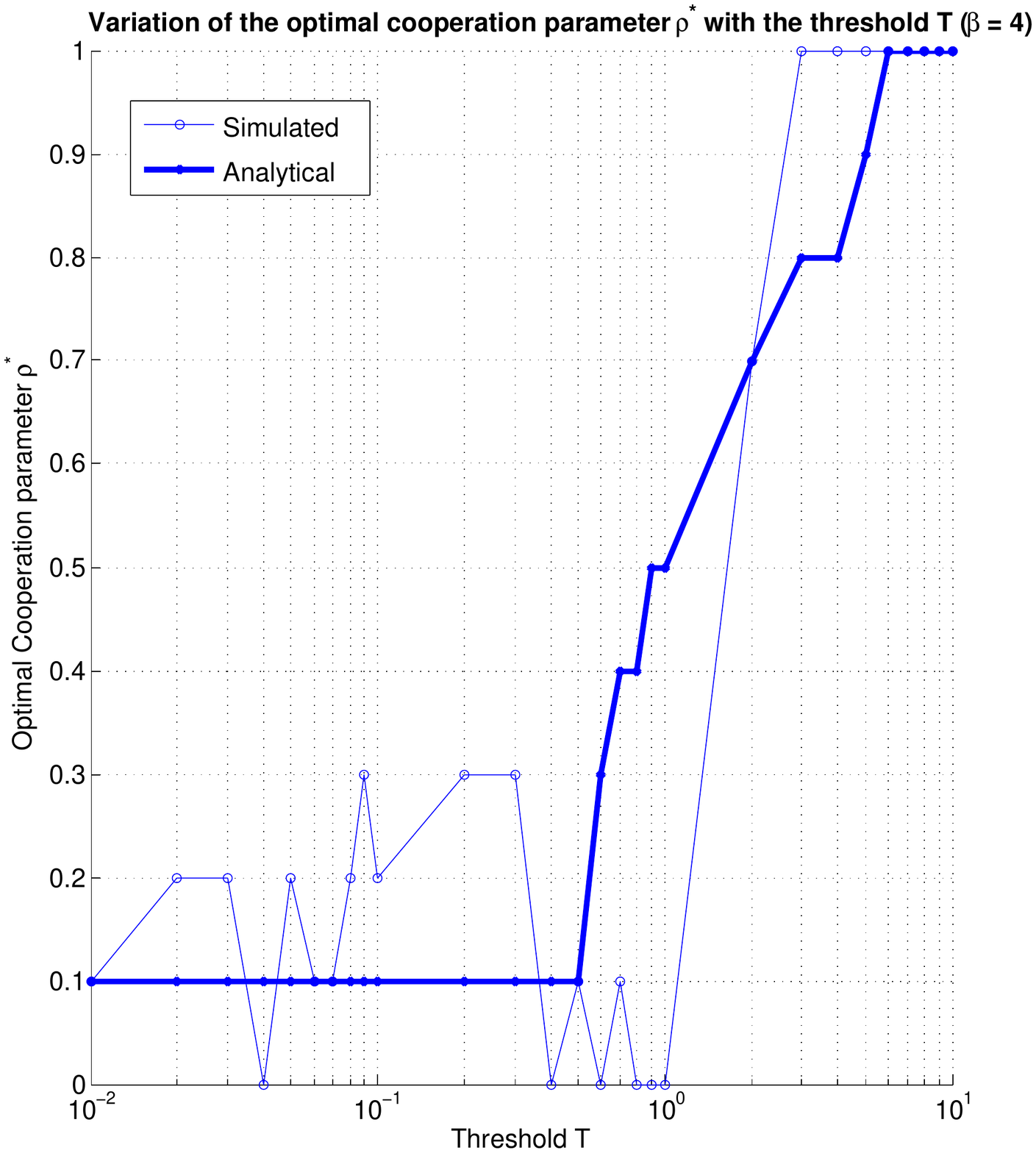}
           \label{fig:OptThresh}
           }
           \caption{Coverage improvement by use of cooperation geometric policies for $\beta=4$. Optimal value of design parameter $\rho^*$ for varying $T$.}
\end{figure}

\begin{figure}[h!]    
\centering  
\label{fig:B4gain}
	 	\subfigure[Comparison of the $\mathrm{SINR}$ model approximation with and without expectation over $\theta$ at the interference.]{          
  \includegraphics[trim = 0mm 45mm 0mm 20mm, clip, width=0.45\textwidth]{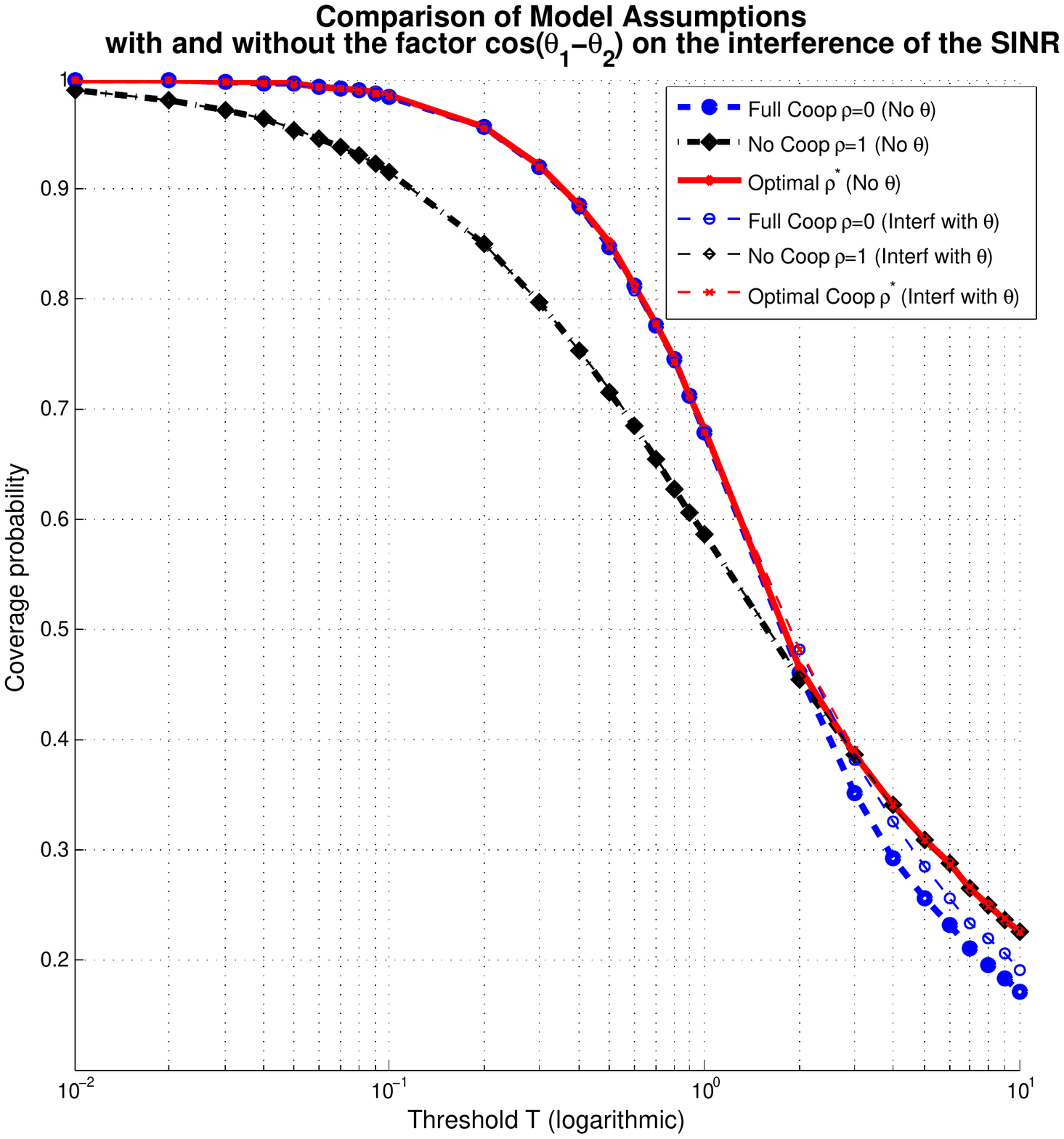}          
\label{fig:SimulTHETAcurves}
           }
            \subfigure[Comparison of the $\mathrm{SINR}$ model with and without the far-field approximation.]{          
\includegraphics[trim = 0mm 35mm 0mm 20mm, clip, width=0.4\textwidth]{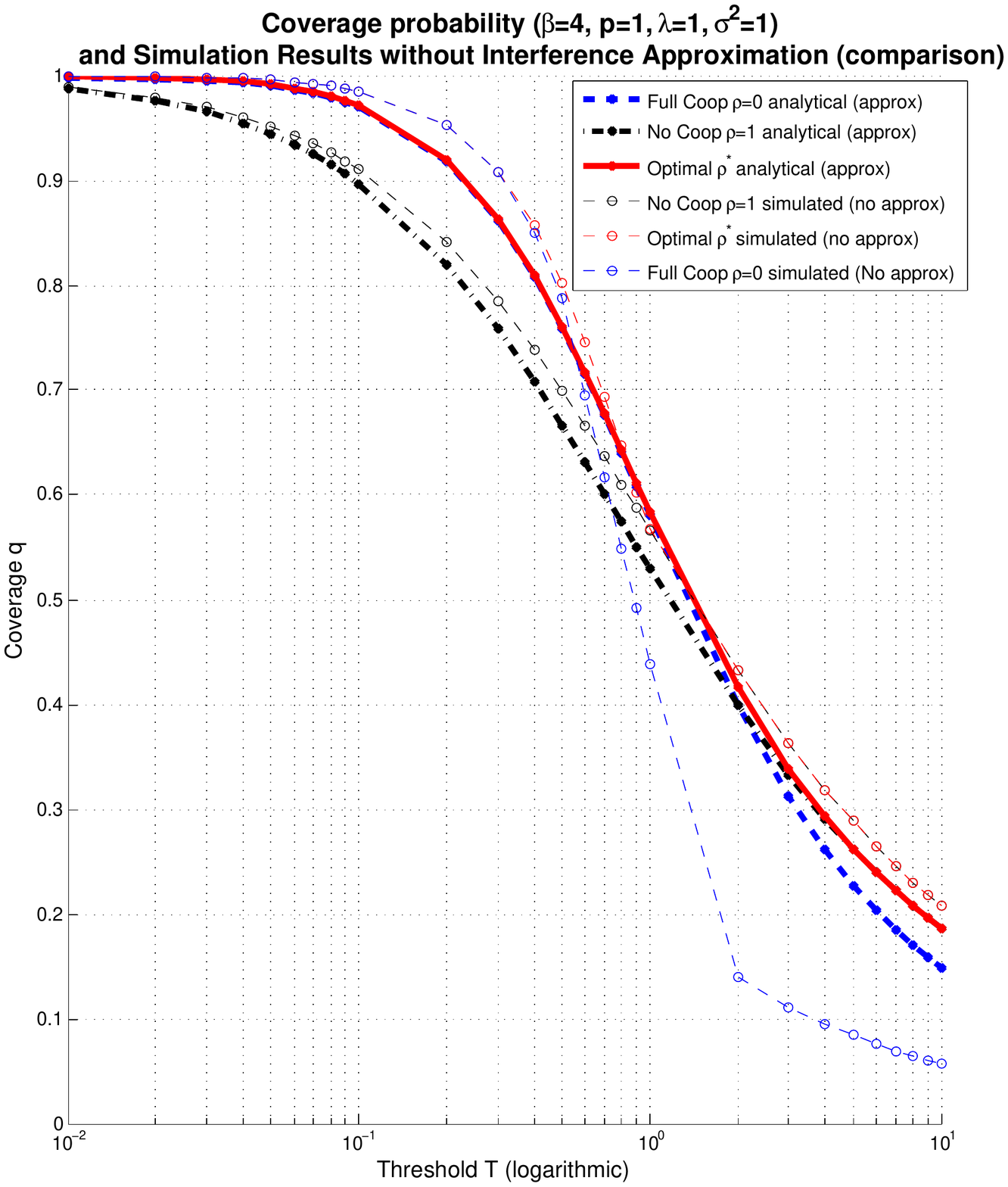}
\label{fig:SimulB24noApproax}
           }
           \caption{(a) Comparison of the two $\mathrm{SINR}$ models in  (\ref{SINRui2a1th})-(\ref{SINRui2a4th}) and (\ref{SINRui2a1})-(\ref{SINRui2a3}) by simulation results for the coverage probability. The two models produce curves that are almost identical to each other. The plots justify the chosen approximation on the interference. (b) Comparison between analytical and simulation results for the case of no approximation ($d_{n1}\neq d_{n2}$) on the interference. Expected number of uniformly produced atoms $\mathbb{E}\left[N\right]=20$ in an area of $20\ \mathrm{m^2}$ with density $\lambda =1\ \mathrm{atom/m^2}$. Fading exponent $\beta=4$.}
\end{figure}

\subsection{Coverage $q\left(\rho\right)$ Versus Threshold $T$, $\beta=3.2$ $\beta=2.5$}
In Fig. \ref{fig:SimulB3p2}, the coverage probability is plotted in function of $T$ for $\beta=3.2$ and in Fig. \ref{fig:SimulB2p5} for $\beta=2.5$. The results are similar to the case of $\beta=4$. The difference between the analytical and simulated values in both is due to the reduced value of the path-loss exponent $\beta$. 
This suggests that the number $\mathbb{E}\left[N\right]=20$ of simulated atoms and the given planar area is not enough to sufficiently approximate tha analytical curves, and a larger area should be considered. This gives a further arguement in favor of the use of point processes and analytical/numerical evaluation, compared to simulations, since the latter may provide erroneous results due to edge effects.

\begin{figure}[h!]    
\centering  
\label{fig:B2p5}
	 	\subfigure[Coverage with interference approximation. Case $\beta=3.2$.]{          
           \includegraphics[trim = 0mm 35mm 0mm 20mm, clip, width=0.45\textwidth]{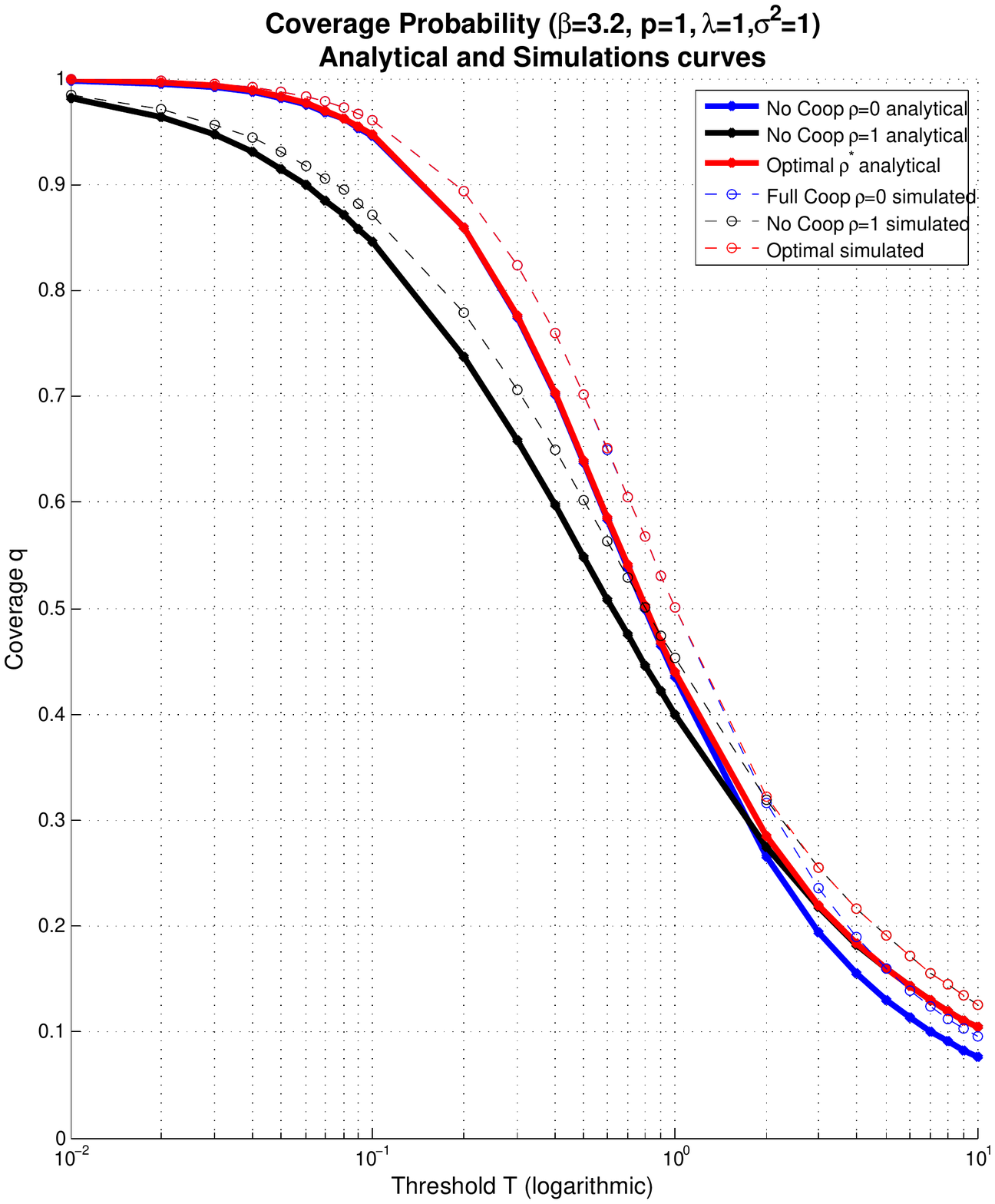}
           \label{fig:SimulB2p5}
           }
            \subfigure[Coverage with interference approximation. Case $\beta=2.5$.]{          
           \includegraphics[trim = 0mm 35mm 0mm 20mm, clip, width=0.45\textwidth]{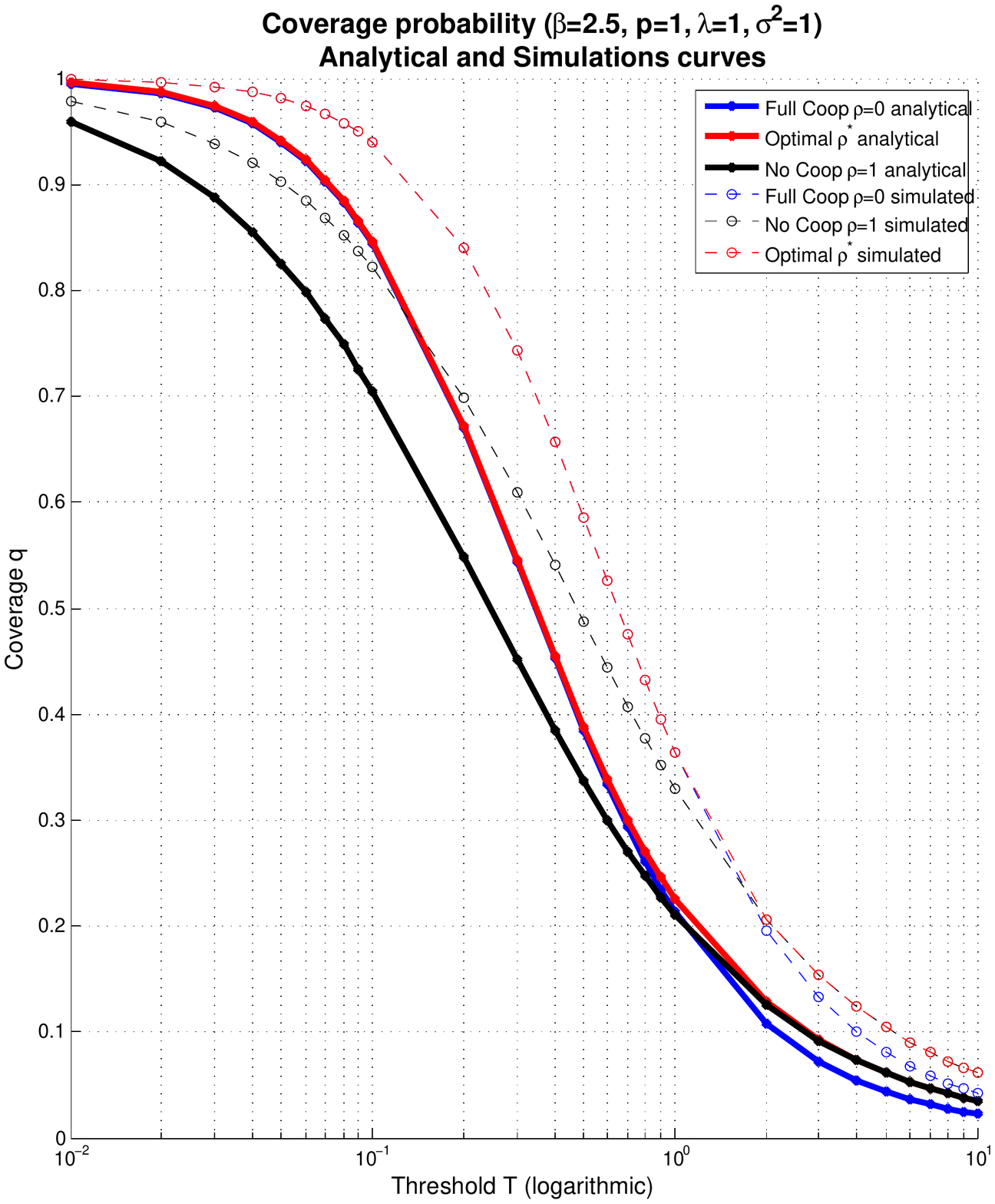}
           \label{fig:SimulB3p2}
           }
           \caption{(a) Comparison between analytical and simulation results. Expected number of uniformly produced atoms $\mathbb{E}\left[N\right]=20$ in an area of $10\ \mathrm{m^2}$ with density $\lambda =1\ \mathrm{atom/m^2}$. Path-loss exponent $\beta=3.2$. (b) Path-loss exponent $\beta=2.5$.}
\end{figure}

%

\subsection{Second Neighbour Interference Elimination - Dirty Paper Coding (DPC)}

Significantly higher gains are obtained in the case where some further system knowledge is available at the BS side. Specifically, the information over second neighbour interference may eliminate the negative effects from $\mathbf{b}_2$ (and obviously $\mathbf{b}_1$) and result in high gains in coverage, in the entire domain of $T$. This is illustrated in the plot of coverage in Fig. \ref{fig:CovGainInterf}. From the plot we get a maximum gain in coverage of over 16\% between $T=0.2$ to $1$. The absolute difference percentage is shown in Fig. \ref{fig:CovDiffGainInterf} and the optimal value of $\rho^*$ for each threshold $T$ in Fig. \ref{fig:OptThreshInterf}.

\begin{figure}[h!]
\centering
\includegraphics[trim = 0mm 30mm 0mm 20mm, clip, width=0.50\textwidth]{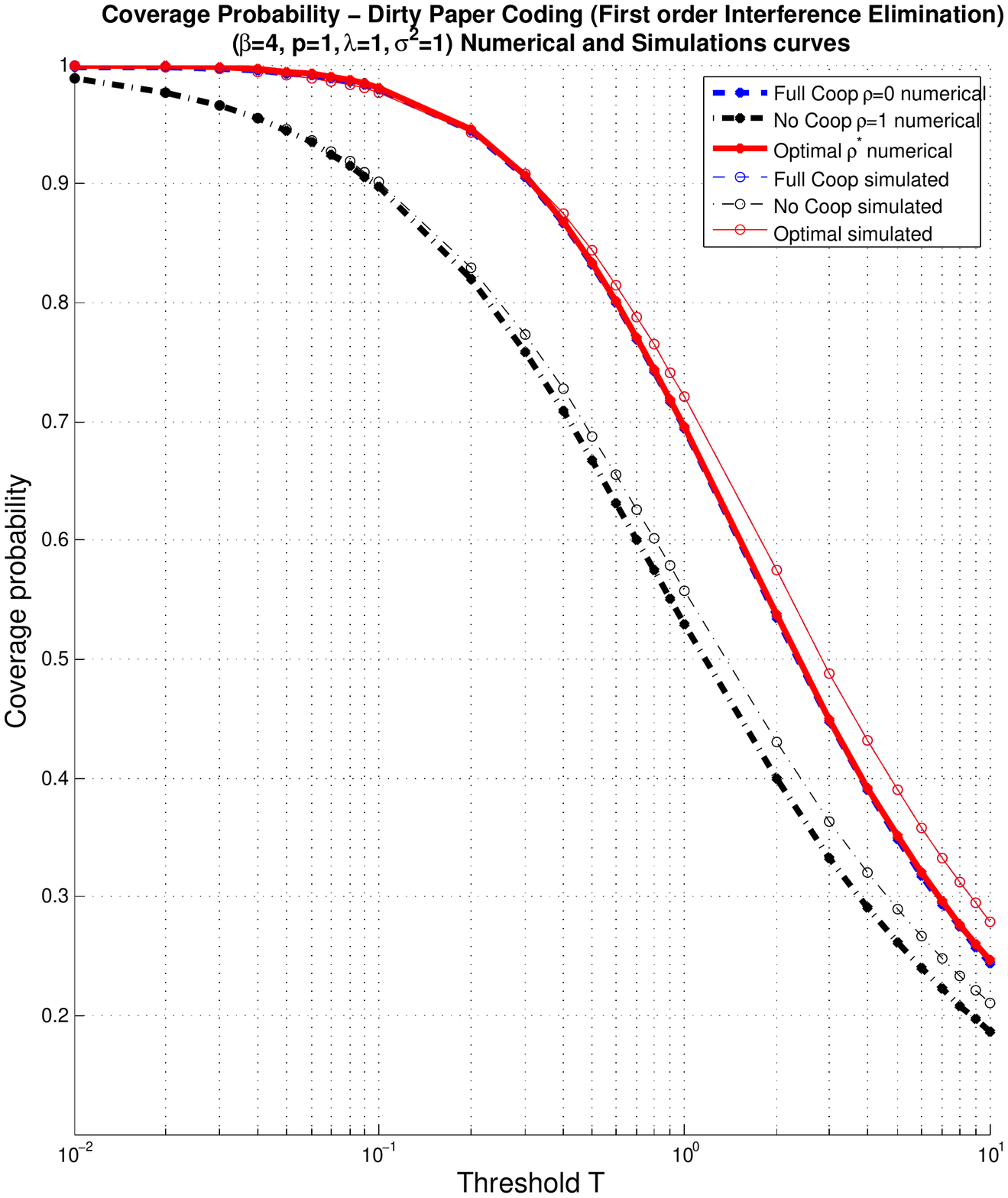}
\caption{Coverage probability using Optimal cooperation, compared to the case with $\mathrm{Full\ Coop}$ and $\mathrm{No\ Coop}$. The case of second neighbour interference cancellation. Path-loss exponent $\beta=4$.}
\label{fig:CovGainInterf}
\end{figure}

\begin{figure}[h!]    
\centering  
\label{fig:DPCgain}
	 	\subfigure[Coverage difference between Optimal and $\mathrm{No\ Coop}$. DPC and $\beta=4$.]{          
           \includegraphics[trim = 0mm 30mm 0mm 20mm, clip, width=0.40\textwidth]{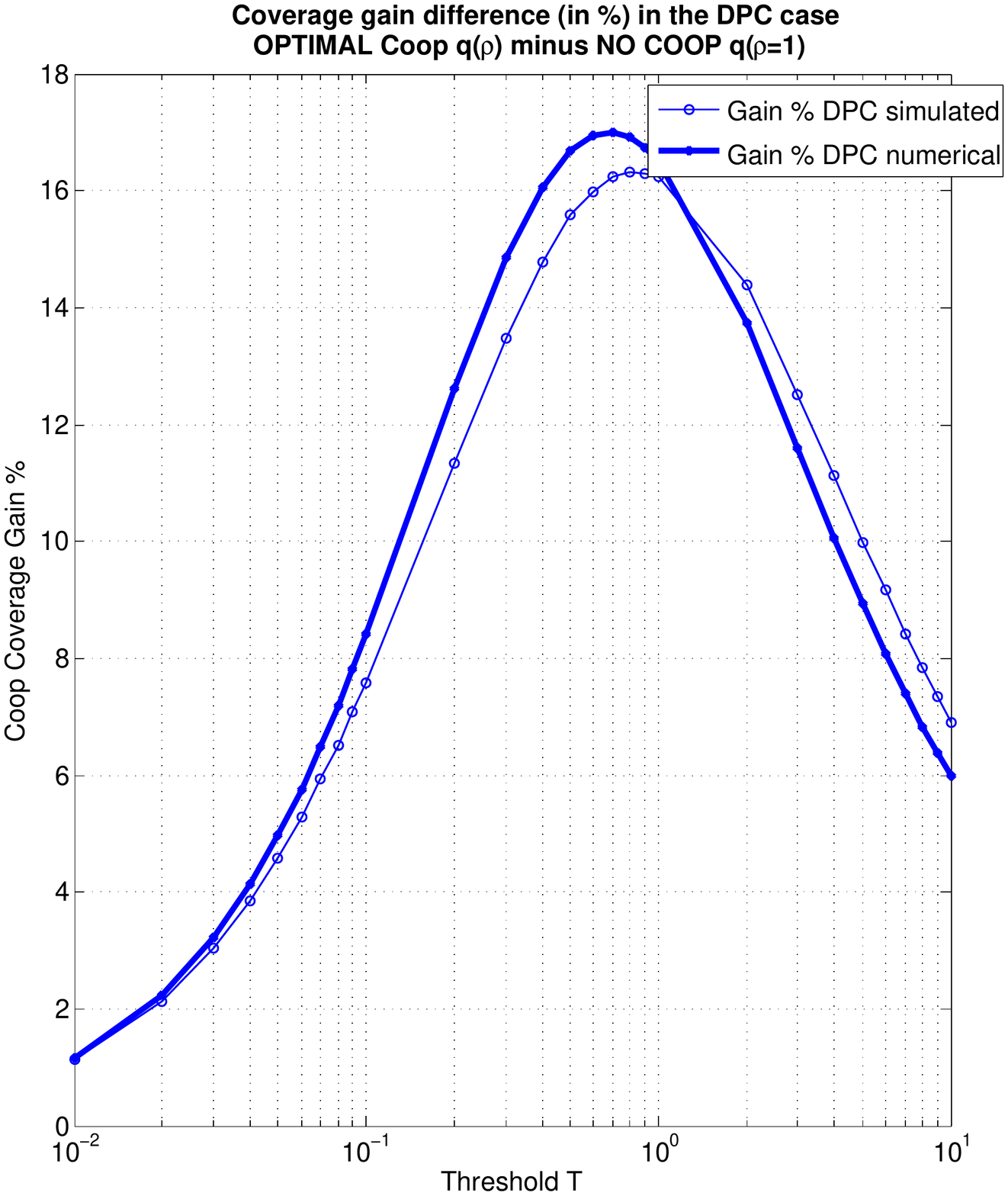}
           \label{fig:CovDiffGainInterf}
           }
            \subfigure[Optimal value of parameter $\rho^*$ versus threshold value $T$. DPC and $\beta=4$.]{          
           \includegraphics[trim = 0mm 30mm 0mm 20mm, clip, width=0.40\textwidth]{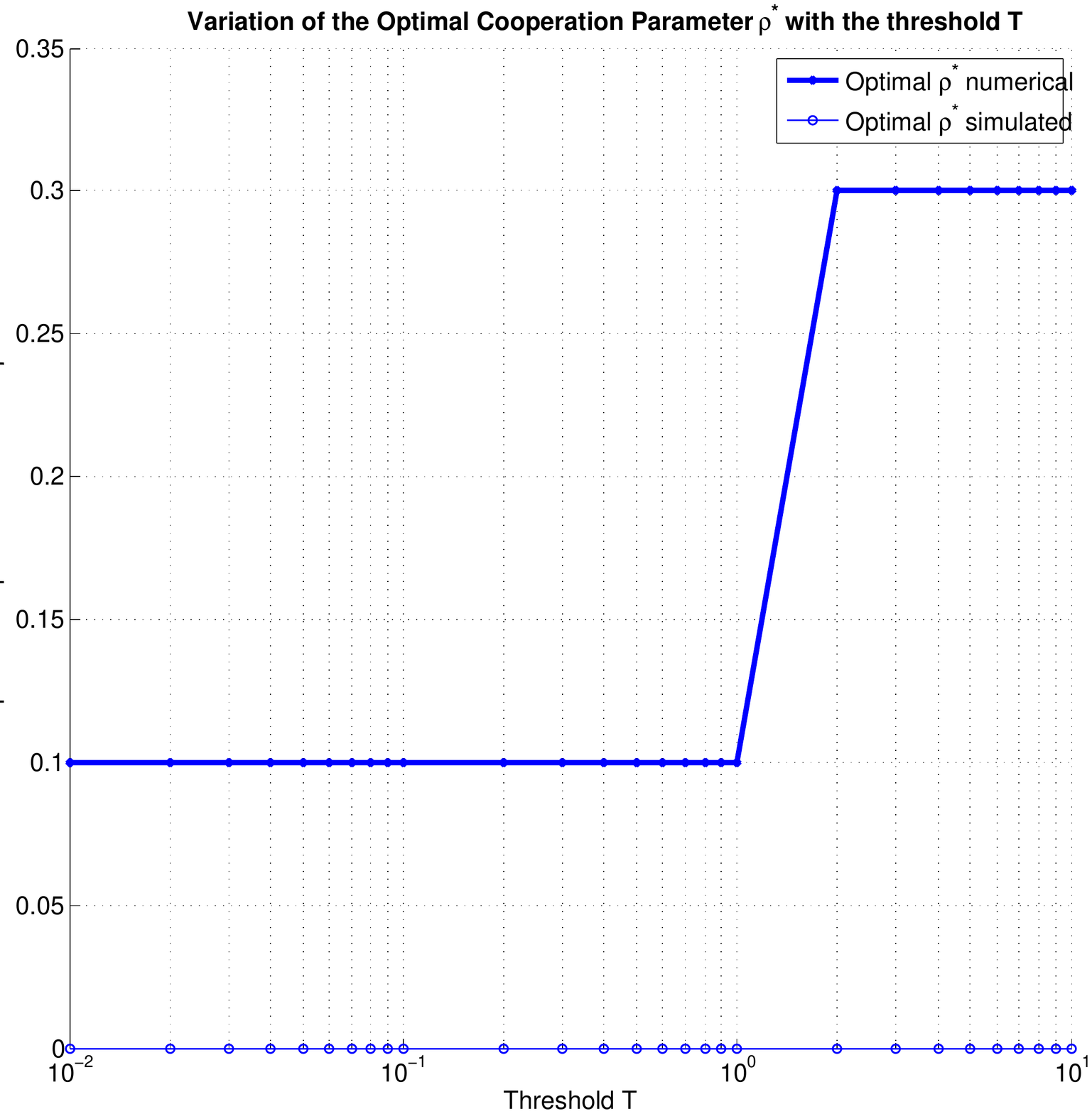}
           \label{fig:OptThreshInterf}
           }
           \caption{Coverage improvement by use of cooperation geometric policies for $\beta=4$. Optimal value of design parameter $\rho^*$ over $T$.}
\end{figure}

\section{Conclusions and Future Work}
\label{SectionVI}

In this paper, we have proposed a general methodology to treat problems of cooperation in cellular networks, in the case where data exchange is allowed only between pairs of nodes.
The framework developed uses tools from stochastic geometry to model the random positions of nodes on the plane and calculate the network performance with reference to 
a typical location. 

In our work the focus has been on the evaluation of the coverage probability, using a specific cooperation scheme. This scheme is based on the idea of conferencing by Willems, where two transmitters encode a common message 
after exchange of information over the backhaul. We give to the service of each user two possibilities. Either to be served by its closest BS or 
to ask from the two closest ones to cooperate. The latter choice splits the total user power between these two BSs, which encode the same common message. The scheme results in 
an additional term on the beneficial received signal due to the correlated nature of the two transmissions. 

The treated cooperation scheme assumes only the knowledge of the fast-fading angular shift at the cooperation pair, so that transmission is done in a way that the two signals add up in-phase at the receiver and the maximum signal power benefit is obtained. The purpose of this assumption on limited channel state information exchange was to investigate possible coverage benefits, that can result without the necessity of exhaustive channel state feedback and full-channel beamformer adaptation. 

The stochastic geometry framework introduce here, can be applied to the evaluation of other types of cooperation, as is the case where the transmitters have full knowledge of the instantaneous channel gains and adapt their functionality to it. Treatment of such cases is very important, since it may reveal higher performance benefits, at the cost of additional information at the transmitter side, and is our future research direction.

To investigate cooperation, we have defined a set of geometric user-defined policies for the user actions. Depending on the ratio $\frac{r_1}{r_2}$ and some design parameter $\rho\in\left[0,1\right]$, a user may ask for BS cooperation or not. It was shown that the cooperation regions are regions close to the cell borders. The width of such zones is determined by the parameter $\rho$. Outside such zones, the user always chooses not to cooperate.

The evaluation of coverage, given that the positions of BS atoms follow a p.p.p. and the fast fading from a single BS to a user follows an exponential distribution, has shown gains when geometric policies of cooperation with conferencing are used, compared to the traditional cellular cell-centered architectures. The expression of $q\left(\rho\right)$ in (\ref{Tint1})+(\ref{Tint2}) with interference r.v. $\mathcal{I}$ and $\mathcal{I}_{DPC}$ respectively, for the cases without and with interference elimination from the second closest BS (by application of Dirty Paper Coding) quantifies the increase in gains with the increase in available information. Without any information and based just on power splitting, the gains are visible for $T<2$ and reach a maximum of over $10\%$ improvement in coverage area. The situation improves considerably when DPC is applied  (but no Zero Forcing) to adapt to the interfering signals coming from the serving pair,  and the gains are visible for the entire domain of $T$. An improvement of up to $17\%$ is apparent in coverage.

Another important benefit has been on the change of the shapes of coverage areas by an appropriate choice of the parameter $\rho$. Cell boundary regions are favored by cooperation since the absolute gain in coverage is concentrated there. Furthermore, the geometric, user-located policies suggested here are not necessarily the optimal ones. Better policies can most probably be obtained when the instantaneous channel fast-fading is known or when the relative influence of the action choice ($\mathrm{Full\ Coop}$ or $\mathrm{No\ Coop}$) on the neighbouring cells is taken into consideration. Optimal policies are expected to be non-geometric.

The issue which requires intensive research and is the direction for future investigations, is the choice of the number of cooperating BSs as well as the type of cooperation. In our work we have chosen to apply cooperation in the sense of power-splitting of the user message between two cooperating transmitters, as well as data exchange and linearly dependent encoding after conferencing. 
DPC is applied when the knowledge over second neighbour BS interference is available. However, the cooperation will have a different form if complex beamformers are available. Knowledge over the channel quality can allow for adaptive transmission by use of a Zero-Forcing or other type of precoder. Additional to that, a larger number of cooperating BSs (greater than two), will as well have different results on the received $\mathrm{SINR}$ at the receiver side.

The above discussions show that there is a potential for improvement on the results presented here and many further steps can be taken to analyze alternative cooperation schemes. The basic methodology presented in this work can provide the analytical framework to facilitate later extensions to the directions mentioned above. The results derived from our model, which assumes limited channel knowledge at the transmission pair, show that cooperation can be a very profitable scheme, which improves system performance, without exploitation of further network resources (frequency, time or power) by appropriate information data exchange between network nodes. The coverage gains from such cooperation strictly between pairs of neighbouring nodes range up to $17\%$ and favor the areas close to the cell edge.

\newpage

%

\bibliographystyle{unsrt}

\newpage


\section*{Appendix}
\footnotesize

\subsection{Geometry aspects}
\label{AppA}

In the current work, the locations of base stations (BSs) of a wireless cellular network on the $\mathbb{R}^2$ plane are determined by 
some realization of a Poisson point process (Poisson p.p.) $\Phi$ with constant density $\lambda$, 
denoted by $\phi=\left\{\mathbf{z}_i\right\}$. Once, the positions are fixed, the model under study is deterministic. We will refer to the elements of $\phi$ as \textit{atoms}.

Each atom $\mathbf{z}_i\in\phi$ is represented as an ordered pair $\left(x_i,y_i\right)$. 
The Euclidean distance between $\mathbf{z}_i$ and any point $\mathbf{z}=\left(x,y\right)\in\mathbb{R}^2$ of the 
plane is defined as the $\ell^2$-norm of their difference

\begin{eqnarray}
\label{distanceM}
d\left(\mathbf{z}_i,\mathbf{z}\right) := \left|\mathbf{z}_i-\mathbf{z}\right|_2 = \sqrt{\left(x_i-x\right)^2+\left(y_i-y\right)^2}.
\end{eqnarray}

Based on the positions of the atoms, the plane can be subdivided geometrically into 
cells, one for each $\mathbf{z}_i$, such that each cell contains all planar points closer to 
its atom than to any other atoms of $\phi$.

\begin{Def}{(\textbf{1-Voronoi cell} \cite{LeekVoronoi}, \cite[pp.148-149]{CompGeomBook})}\\
The 1-Voronoi cell $\mathcal{V}_1\left(\mathbf{z}_i\right)$ associated with $\mathbf{z}_i$ is the locus of all points in $\mathbb{R}^2$ which are closer to $\mathbf{z}_i$ 
than to any other atom of $\phi$. Then, 

\begin{eqnarray}
\label{Voronoi1Set}
\mathcal{V}_1\left(\mathbf{z}_i\right)  = \left\{\mathbf{z}\in\mathbb{R}^2|\ d\left(\mathbf{z}_i,\mathbf{z}\right)\leq d\left(\mathbf{z}_k,\mathbf{z}\right), \forall \mathbf{z}_k\in\phi\setminus\left\{\mathbf{z}_i\right\}\right\}.\nonumber\\
\end{eqnarray}


Alternatively, given any two distinct atoms $\mathbf{z}_i,\mathbf{z}_k\in\phi$, the locus of points closer to $\mathbf{z}_i$ than $\mathbf{z}_k$ 
denoted by $h\left(\mathbf{z}_i,\mathbf{z}_k\right)$, is the one of two closed half-planes, determined by the perpendicular bisector 
of the line segment $\overline{\mathbf{z}_i\mathbf{z}_k}$, which contains $\mathbf{z}_i$. The 1-Voronoi cell can be defined as the intersection of all half-planes associated with $\mathbf{z}_i$
\begin{eqnarray}
\label{Voronoi1HP}
\mathcal{V}_1\left(\mathbf{z}_i\right) & = & \bigcap_{k:k\neq i} h\left(\mathbf{z}_i,\mathbf{z}_k\right).
\end{eqnarray}
\label{DefV1}
\end{Def}
The collection of all cells $\mathcal{V}_1\left(\mathbf{z}_i\right)$ constitutes 
the 1-Voronoi diagram (or planar tessellation) $\mathcal{V}_1\left(\phi\right)$ associated with the set of points $\phi$. The Voronoi diagram is a graph and its dual can be constructed, which is the so called Delaunay graph (or triangulation) of the set $\phi$.

\begin{Pro}(\textbf{1-Voronoi vertices and edges} \cite[pp.150-151, Th.7.4]{CompGeomBook})\\
A planar point $\mathbf{z}$ is a \textit{vertex} of the diagram if and only if the largest empty open ball with $\mathbf{z}$ as center contains at least three atoms on its boundary. 

An edge of $\mathcal{V}_1\left(\mathbf{z}\right)$ is related to exactly two atoms $\mathbf{z}_i$ and $\mathbf{z}_n$ and is the set of planar points
\begin{eqnarray}
\label{edgeV1}
e\left(\mathbf{z}_i,\mathbf{z}_n\right) & = & \left\{\mathbf{z}\in\mathbb{R}^2|\ d\left(\mathbf{z}_i,\mathbf{z}\right)=d\left(\mathbf{z}_n,\mathbf{z}\right) < d\left(\mathbf{z}_k,\mathbf{z}\right),\right.\nonumber\\
& & \left.\mathbf{z}_k\in\phi\setminus\left\{\mathbf{z}_i,\mathbf{z}_n\right\}\right\}
\end{eqnarray}
The edge lies on the bisector of the line segment $\overline{\mathbf{z}_i\mathbf{z}_n}$ and each $\mathbf{z}\in e\left(\mathbf{z}_i,\mathbf{z}_n\right)$ - conditioned that the above set is not empty - 
is the center of an empty open ball with the two related atoms and no other on its boundary.
\label{VertexEdgeV1}
\end{Pro}

\begin{Def} (\textbf{Delaunay graph} \cite[p.196]{CompGeomBook}))\\
The Delaunay graph $\mathcal{D}\left(\phi\right)$ is the dual of the 1-Voronoi diagram $\mathcal{V}_1\left(\phi\right)$. In the Delaunay graph each atom $\mathbf{z}_i$ takes the role of a vertex and two vertices $\mathbf{z}_i$, $\mathbf{z}_n$ are connected by an edge $\mathbf{z}_i\mathbf{z}_n$ if the set $e\left(\mathbf{z}_i,\mathbf{z}_n\right)$ in (\ref{edgeV1}) is non-empty, i.e. if the two 1-Voronoi cells $\mathcal{V}_1\left(\mathbf{z}_i\right)$ and $\mathcal{V}_1\left(\mathbf{z}_n\right)$ share an edge.
\label{Delaunay}
\end{Def}

\begin{Pro} (\textbf{Delaunay edges} \cite[p.198, Th.9.6]{CompGeomBook})\\

Two points $\mathbf{z}_i,\mathbf{z}_n\in\phi$ form an edge of the Delaunay graph if and only if there exists an empty ball with only $\mathbf{z}_i$ and $\mathbf{z}_n$ on its boundary.

Three atoms $\mathbf{z}_i$, $\mathbf{z}_n$ and $\mathbf{z}_m$ are vertices of the same face of the Delaunay graph if and only if the circle which passes through them is the boundary of a ball, empty of other atoms in $\phi$.
\label{DelaunayCond}
\end{Pro}

For the description of the topology for the communications scenario under study, we will also make use of the 2-order Voronoi diagram $\mathcal{V}_2\left(\phi\right)$, which is the collection of all non-empty 2-Voronoi cells.

\begin{Def} (\textbf{2-Voronoi cell} \cite{LeekVoronoi})\\
The 2-Voronoi cell $\mathcal{V}_2\left(\mathbf{z}_i,\mathbf{z}_n\right)$ associated with $\mathbf{z}_i,\mathbf{z}_n\in\phi$, $i\neq n$ is the locus of all points in $\mathbb{R}^2$ closer to $\left\{\mathbf{z}_i,\mathbf{z}_n\right\}$ than to any other atom in $\phi$
\begin{eqnarray}
\label{Voronoi2Set}
\mathcal{V}_2\left(\mathbf{z}_i,\mathbf{z}_n\right) & = & \left\{\mathbf{z}\in\mathbb{R}^2|\ d\left(\mathbf{z}_i,\mathbf{z}\right)\leq d\left(\mathbf{z}_k,\mathbf{z}\right) \  \&\ \right.\nonumber\\
								 &  & \left.d\left(\mathbf{z}_n,\mathbf{z}\right)\leq d\left(\mathbf{z}_k,\mathbf{z}\right), \forall \mathbf{z}_k\in\phi\setminus\left\{\mathbf{z}_i,\mathbf{z}_n\right\}\right\}.\nonumber\\
\end{eqnarray}
The 2-Voronoi cell can be defined as the intersection of all half-planes associated with $\left\{\mathbf{z}_i,\mathbf{z}_n\right\}$, that is
\begin{eqnarray}
\label{Voronoi2}
\mathcal{V}_2\left(\mathbf{z}_i,\mathbf{z}_n\right) & = & \bigcap_{\mathbf{z}_l\in\left\{\mathbf{z}_i,\mathbf{z}_n\right\},\mathbf{z}_k\in\phi\setminus\left\{\mathbf{z}_i,\mathbf{z}_n\right\}} h\left(\mathbf{z}_l,\mathbf{z}_k\right).
\end{eqnarray}
\label{DefV2}
\end{Def}

Given any point $\mathbf{z}\in\mathbb{R}^2$, we can find its first closest neighbouring atom and 
the corresponding 1-Voronoi cell it belongs to, by enlarging the radius $r$ of an empty ball $\mathcal{B}\left(\mathbf{z},r\right)$ 
with center $\mathbf{z}$, until at least one atom meets its boundary. Given that a single atom $\mathbf{z}_i$ was met, by further expanding 
the radius until exactly one atom lies in the interior and at least another one meets the boundary, we find the 
second closest neighbouring atom and the corresponding 2-Voronoi cell. This is explicitly stated in the following Propositions.

\begin{Pro}
\label{CircleV1}
A point $\mathbf{z}\in\mathbb{R}^2$ lies within the 1-Voronoi cell $\mathbf{z}\in\mathcal{V}_1\left(\mathbf{z}_i\right)$ if and only if it is the center of an empty ball of radius $r_{i1} = d\left(\mathbf{z}_i,\mathbf{z}\right)$ with $\mathbf{z}_i$ on its boundary. The ball is denoted by $\mathcal{B}\left(\mathbf{z},r_{i1}\right)$.
\end{Pro}

\begin{proof}
If $r_{i1}$ is such that $\mathcal{B}\left(\mathbf{z},r_{i1}\right)$ is empty, then $r_{i1} \leq d\left(\mathbf{z},\mathbf{z}_k\right)$, $\forall \mathbf{z}_k\in\phi\setminus\left\{\mathbf{z}_i\right\}$ and by (\ref{Voronoi1Set}) in Def. \ref{DefV1} we conclude $\mathbf{z}\in\mathcal{V}_1\left(\mathbf{z}_i\right)$. If $r_{i1}$ is such that $\mathcal{B}\left(\mathbf{z},r_{i1}\right)$ contains in its interior at least one atom $\mathbf{z}_k\neq \mathbf{z}_i$, while $\mathbf{z}_i$ lies on 
its boundary, then  $r_{i1} > d\left(\mathbf{z},\mathbf{z}_k\right)$ and $\mathbf{z}\notin\mathcal{V}_1\left(\mathbf{z}_i\right)$.
\end{proof}

\begin{Pro}
\label{CircleV2}
A point $\mathbf{z}\in\mathbb{R}^2$ lies within the 2-Voronoi cell $\mathbf{z}\in\mathcal{V}_2\left(\mathbf{z}_i,\mathbf{z}_n\right)$ if and only if 
it is the center of a ball $\mathcal{B}\left(\mathbf{z},r_{i2}\right)$ of radius $r_{i2} = \max\left\{d\left(\mathbf{z}_i,\mathbf{z}\right), d\left(\mathbf{z}_n,\mathbf{z}\right)\right\}$, 
with the property that it contains one of the two atoms $\mathbf{z}_i,\mathbf{z}_n$ in the interior (say $\mathbf{z}_n$) and no other, while the second atom (say $\mathbf{z}_j$) lies on its boundary. Alternatively, iff both atoms lie on the boundary and the open ball is empty in the case $d\left(\mathbf{z}_i,\mathbf{z}\right)= d\left(\mathbf{z}_n,\mathbf{z}\right)$.
\end{Pro}

\begin{proof}
If the ball $\mathcal{B}\left(\mathbf{z},r_{i2}\right)$ has the properties of the proposition, then obviously $d\left(\mathbf{z}_i,\mathbf{z}\right)\leq d\left(\mathbf{z}_n,\mathbf{z}\right)\leq d\left(\mathbf{z}_k,\mathbf{z}\right)$, 
$\forall \mathbf{z}_k\in\phi\setminus\left\{\mathbf{z}_i,\mathbf{z}_n\right\}$ and by (\ref{Voronoi2Set}) in Def. \ref{DefV2} we conclude that $\mathbf{z}\in\mathcal{V}_2\left(\mathbf{z}_i,\mathbf{z}_n\right)$. 

Conversely, if $\mathbf{z}\in\mathcal{V}_2\left(\mathbf{z}_i,\mathbf{z}_n\right)$, then $d\left(\mathbf{z}_i,\mathbf{z}\right)\leq r_{i3}$ and $d\left(\mathbf{z}_n,\mathbf{z}\right)\leq r_{i3}$, where $r_{i3} = \min_{\mathbf{z}_k\in\phi\setminus\left\{\mathbf{z}_i,\mathbf{z}_n\right\}}d\left(\mathbf{z}_k,\mathbf{z}\right)$. Without loss of generality, let $r_{i2} = d\left(\mathbf{z}_n,\mathbf{z}\right)$ and $r_{i2}\leq r_{i3}$, 
hence either $\mathbf{z}_i$ - and no other atom - lies within the ball $\mathcal{B}\left(\mathbf{z},r_{i2}\right)$ and the $\mathbf{z}_n$ lies on the boundary, or both $\mathbf{z}_i,\mathbf{z}_n$ 
lie on the boundary and the ball is empty. (Note: In the first case, at most two $\mathbf{z}_k$, $k\neq i,n$ can also lie on the boundary and in the second, at most one $\mathbf{z}_k$, $k\neq i,n$ can also be on the boundary. The reason is that the atoms of a realization of a p.p.p. are in the so called "general quadratic position" \cite{LeekVoronoi}, which allows at most three atoms to be co-circular almost surely.)
\end{proof}

The second closest neighbouring atom to any planar point $\mathbf{z}\in\mathbb{R}^2$ 
can also be found by artificially removing its first closest neighbour and recalculating the Voronoi diagram. Based on this idea, a 2-Voronoi cell can be constructed using set operations on 1-Voronoi cells.

\begin{Pro} (\cite[p.481, Lemma 1]{LeekVoronoi})
Let $\mathcal{V}_1\left(\mathbf{z}_i;\tilde{\phi}\right)$ be the 1-Voronoi cell of 
atom $\mathbf{z}_i$, when the Voronoi diagram is constructed by the subset $\tilde{\phi}\subseteq\phi$. Given two atoms  
$\mathbf{z}_i,\mathbf{z}_n\in \phi$, their 2-Voronoi cell can be expressed as 
\begin{eqnarray}
\label{ReconV2fromV1}
\mathcal{V}_2\left(\mathbf{z}_i,\mathbf{z}_n\right) & = & \left(\mathcal{V}_1\left(\mathbf{z}_n;\phi\right)\cap\mathcal{V}_1\left(\mathbf{z}_i;\phi\setminus\left\{\mathbf{z}_n\right\}\right)\right)\nonumber\\
& & \bigcup\left(\mathcal{V}_1\left(\mathbf{z}_i;\phi\right)\cap\mathcal{V}_1\left(\mathbf{z}_n;\phi\setminus\left\{\mathbf{z}_i\right\}\right)\right).
\end{eqnarray}
\label{ProV2V1}
\end{Pro}
Thus, the 2-Voronoi cell of $\mathbf{z}_i,\mathbf{z}_n$ can be expressed as the union of two subregions. 
The first subregion can be found by taking the 1-Voronoi cell of $\mathbf{z}_i$ 
and by further considering its intersection with the 1-Voronoi cell of the other atom $\mathbf{z}_n$, when the diagram is 
drawn after removal of $\mathbf{z}_i$. The second subregion is similar by interchanging the roles of $\mathbf{z}_i$ and $\mathbf{z}_n$.

%
%
%
%

We will close this section by proving an important theorem which relates the 1-Voronoi diagram (and its dual) with the 2-Voronoi diagram.

\begin{Pro}
The 2-Voronoi cell $\mathcal{V}_2\left(\mathbf{z}_i,\mathbf{z}_n\right)$ is non-empty if and only if the atoms $\mathbf{z}_i$ and $\mathbf{z}_n$ are Delaunay neighbours. Alternatively iff the 1-Voronoi cells of $\mathbf{z}_i$ and $\mathbf{z}_n$ share a 1-Voronoi edge.
\label{ProV12}
\end{Pro}

\begin{proof}
If $\mathbf{z}_i$ and $\mathbf{z}_n$ are Delaunay neighbours then $e\left(\mathbf{z}_i,\mathbf{z}_n\right)\neq \emptyset$ in (\ref{edgeV1}) and there exists at least one 
planar point $\mathbf{z}$ with the property $d\left(\mathbf{z}_i,\mathbf{z}\right)=d\left(\mathbf{z}_n,\mathbf{z}\right)<d\left(\mathbf{z}_k,\mathbf{z}\right)$, $\forall k\neq i,n$. By (\ref{Voronoi2Set}) in Def. \ref{DefV2} this 
implies that $\mathbf{z}\in\mathcal{V}_2\left(\mathbf{z}_i,\mathbf{z}_n\right)\neq \emptyset$.

For the converse we will show that if $\mathbf{z}_i$ and $\mathbf{z}_n$ are not Delaunay neighbours then $\mathcal{V}_2\left(\mathbf{z}_i,\mathbf{z}_n\right)= \emptyset$. If the assumption 
holds, then $e\left(\mathbf{z}_i,\mathbf{z}_n\right) = \emptyset$ and there exists no point $\mathbf{z}$ with the property $d\left(\mathbf{z}_i,\mathbf{z}\right)=d\left(\mathbf{z}_n,\mathbf{z}\right)<d\left(\mathbf{z}_k,\mathbf{z}\right)$, $\forall k\neq i,n$. Take the definition of the 2-Voronoi cell in (\ref{Voronoi2Set}), Def. \ref{DefV2}. With the above restriction, 

\begin{eqnarray}
\label{Voronoi2SetR1}
\mathcal{V}_2\left(\mathbf{z}_i,\mathbf{z}_n\right) & 	= & \left\{\mathbf{z}\in\mathbb{R}^2|\ d\left(\mathbf{z}_i,\mathbf{z}\right) > d\left(\mathbf{z}_n,\mathbf{z}\right)\ or\ d\left(\mathbf{z}_i,\mathbf{z}\right)< d\left(\mathbf{z}_n,\mathbf{z}\right) \right.\nonumber\\
											&  & \ \&\  d\left(\mathbf{z}_i,\mathbf{z}\right)\leq d\left(\mathbf{z}_k,\mathbf{z}\right) \nonumber\\
								 			&  & \ \&\ \left.d\left(\mathbf{z}_n,\mathbf{z}\right)\leq d\left(\mathbf{z}_k,\mathbf{z}\right), \forall \mathbf{z}_k\in\phi\setminus\left\{\mathbf{z}_i,\mathbf{z}_n\right\}\right\}.\nonumber
\end{eqnarray}
Since the 2-Voronoi cell is defined in (\ref{Voronoi2}) as intersection of half-planes, it is a convex set, hence exactly one of the two inequalities $d\left(\mathbf{z}_i,\mathbf{z}\right)< d\left(\mathbf{z}_n,\mathbf{z}\right) \ or\  d\left(\mathbf{z}_i,\mathbf{z}\right)> d\left(\mathbf{z}_n,\mathbf{z}\right)$ can hold, say $d\left(\mathbf{z}_i,\mathbf{z}\right)< d\left(\mathbf{z}_j,\mathbf{z}\right)$ without loss of generality. Then

\begin{eqnarray}
\label{Voronoi2SetR2}
\mathcal{V}_2\left(\mathbf{z}_i,\mathbf{z}_n\right) 	& = & \left\{\mathbf{z}\in\mathbb{R}^2|\ d\left(\mathbf{z}_i,\mathbf{z}\right)< d\left(\mathbf{z}_n,\mathbf{z}\right) \right.\nonumber\\
											&   & \ \&\  d\left(\mathbf{z}_i,\mathbf{z}\right)\leq d\left(\mathbf{z}_k,\mathbf{z}\right) \nonumber\\
								 			&   & \ \&\ \left.d\left(\mathbf{z}_n,\mathbf{z}\right)\leq d\left(\mathbf{z}_k,\mathbf{z}\right), \forall \mathbf{z}_k\in\phi\setminus\left\{\mathbf{z}_i,\mathbf{z}_n\right\}\right\}.\nonumber
\end{eqnarray}
However, the second inequality of the above set cannot hold with equality for some $k\neq i,n$ since in such case $d\left(\mathbf{z}_i,\mathbf{z}\right)= d\left(\mathbf{z}_k,\mathbf{z}\right)< d\left(\mathbf{z}_n,\mathbf{z}\right)$ and the $\mathbf{z}_i,\mathbf{z}_k$ are closer to $\mathbf{z}$ than $\mathbf{z}_n$. Hence the set of inequalities can 
be re-written more compactly as $d\left(\mathbf{z}_i,\mathbf{z}\right)< d\left(\mathbf{z}_k,\mathbf{z}\right)$, $\forall \mathbf{z}_k\in\phi\setminus\left\{\mathbf{z}_i\right\}$ and $d\left(\mathbf{z}_n,\mathbf{z}\right)\leq d\left(\mathbf{z}_k,\mathbf{z}\right)$, $\forall \mathbf{z}_k\in\phi\setminus\left\{\mathbf{z}_i,\mathbf{z}_n\right\}$.

Using the Definition of the 1-Voronoi cell (\ref{Voronoi1Set}) and the expression in (\ref{ReconV2fromV1})
\begin{eqnarray}
\label{Voronoi2SetR5}
\mathcal{V}_2\left(\mathbf{z}_i,\mathbf{z}_n\right) 	& = & \textbf{int}(\mathcal{V}_1\left(\mathbf{z}_i\right))\bigcap\mathcal{V}_1\left(\mathbf{z}_n;\phi\setminus\left\{\mathbf{z}_i\right\}\right).
\end{eqnarray}
Here $\textbf{int}(\mathcal{V}_1\left(\mathbf{z}_i\right))$ denotes the interior of the 1-Voronoi cell of $\mathbf{z}_i$ (i.e. without including any edges or vertices resulting from $d\left(\mathbf{z}_i,\mathbf{z}\right)= d\left(\mathbf{z}_k,\mathbf{z}\right)$ for some $\mathbf{z}_k$). 

The above (\ref{Voronoi2SetR5}) implies that the 2-Voronoi cell:
\begin{itemize}
\item either is strictly contained in the interior of $\mathcal{V}_1\left(\mathbf{z}_i\right)$
\item or the intersection of the two sets is empty.
\end{itemize}

Assume that the 2-Voronoi cell is not empty and consider a point on its \textit{boundary} $\mathbf{z}\in\partial{\mathcal{V}_2\left(\mathbf{z}_i,\mathbf{z}_j\right) }$. 
There are 4 cases:

\begin{enumerate}
\item There exists an atom $\mathbf{z}_m\in\phi$, $m\neq i,n$ such that
\begin{eqnarray}
\label{BCaseI}
d\left(\mathbf{z}_i,\mathbf{z}\right) = d\left(\mathbf{z}_m,\mathbf{z}\right) & \& & d\left(\mathbf{z}_n,\mathbf{z}\right) = d\left(\mathbf{z}_m,\mathbf{z}\right).\nonumber
\end{eqnarray}
Then the atoms $\mathbf{z}_i$, $\mathbf{z}_n$ and $\mathbf{z}_m$ are co-circular with center $\mathbf{z}$ and the circle is empty of other atoms in $\phi$, hence there exists by Prop. \ref{DelaunayCond} 
an edge connecting them in the Delaunay graph, which contradicts the assumption that $\mathbf{z}_i$ and $\mathbf{z}_n$ are not Delaunay neighbours.
\item There exist two atoms $\mathbf{z}_m,\mathbf{z}_l\in\phi$, $m,l\neq i,n$ such that
\begin{eqnarray}
\label{BCaseII}
d\left(\mathbf{z}_i,\mathbf{z}\right) = d\left(\mathbf{z}_m,\mathbf{z}\right) & \& & d\left(\mathbf{z}_n,\mathbf{z}\right) = d\left(\mathbf{z}_l,\mathbf{z}\right).\nonumber
\end{eqnarray}
Then the following inequality holds $d\left(\mathbf{z}_i,\mathbf{z}\right) = d\left(\mathbf{z}_m,\mathbf{z}\right) < d\left(\mathbf{z}_n,\mathbf{z}\right) = d\left(\mathbf{z}_l,\mathbf{z}\right)$, which contradicts the assumption $\mathbf{z}\in\partial{\mathcal{V}_2\left(\mathbf{z}_i,\mathbf{z}_n\right) }$. (We consider strict inequality, since at most 3 atoms can be co-circular - due to the "general quadratic position" of the atoms - and the case of equality is not allowed.) 
\item There exists one atom $\mathbf{z}_m\in\phi$, $m\neq i,n$ such that
\begin{eqnarray}
\label{BCaseIII}
d\left(\mathbf{z}_i,\mathbf{z}\right) = d\left(\mathbf{z}_m,\mathbf{z}\right) & \& & d\left(\mathbf{z}_n,\mathbf{z}\right) < d\left(\mathbf{z}_m,\mathbf{z}\right).\nonumber
\end{eqnarray}
which leads to $d\left(\mathbf{z}_i,\mathbf{z}\right) > d\left(\mathbf{z}_n,\mathbf{z}\right)$, that is $\mathbf{z}\in\mathcal{V}_1\left(\mathbf{z}_n\right)$, which contradicts (\ref{Voronoi2SetR5}).
\item There exists one atom $\mathbf{z}_m\in\phi$, $m\neq i,n$ such that
\begin{eqnarray}
\label{BCaseIV}
d\left(\mathbf{z}_i,\mathbf{z}\right) < d\left(\mathbf{z}_m,\mathbf{z}\right) & \& & d\left(\mathbf{z}_n,\mathbf{z}\right) = d\left(\mathbf{z}_m,\mathbf{z}\right).\nonumber
\end{eqnarray}
Then $\mathbf{z}\in e\left(\mathbf{z}_n,\mathbf{z}_m\right)\subset\mathcal{V}_1\left(\mathbf{z}_n\right)\ \& \ \mathbf{z}\in\mathbf{int}\left(\mathcal{V}_1(\mathbf{z}_i)\right)$, which is impossible.
\end{enumerate}
The above show that $\partial{\mathcal{V}_2\left(\mathbf{z}_i,\mathbf{z}_n\right)}=\emptyset\Rightarrow\mathcal{V}_2\left(\mathbf{z}_i,\mathbf{z}_n\right)=\emptyset$.
\end{proof}

\newpage
\subsection{Special Case for User Position}
\label{AppB}

A special case of the general cooperation scenario described in Section \ref{SecIIA} is when the cardinality of the secondary user set 
per BS is exactly one for each $\mathbf{z}_i\in\phi$, that is 

\begin{eqnarray}
\left|\mathcal{N}^s\left(\mathbf{z}_i\right)\right| = 1, & & \forall \mathbf{z}_i\in\phi.
\label{CoopSpecial}
\end{eqnarray}
The above simply means that the following special scenario is considered

\begin{Def} (\textbf{Special Cooperation Scenario})\\
\begin{itemize}
\item The infrastructure for the BSs as well as the cooperation scenario follows Section \ref{SecIIA}.
\item Each BS $\mathbf{z}_i\in\phi$ has exactly one primary user who lies within its 1-Voronoi cell and exactly one secondary user 
who lies within the 1-Voronoi cell of one of its Delaunay neighbours. 
\end{itemize}
\label{CoopDefS}
\end{Def}
In other words, the positions of users are such that only one out of all possible 2-Voronoi cells related to $\mathbf{z}_i$ is 
nonempty. In what follows we will discuss questions regarding the 
feasibility of such a Special Cooperation Scenario. To better pose this question mathematically we will introduce the notion of direction on the Delaunay graph.

\begin{Def}
We say than an edge $\mathbf{z}_i\mathbf{z}_n$ of the Delaunay diagram has \textbf{direction} from vertex $\mathbf{z}_i$ to vertex $\mathbf{z}_n$
 if the primary user of $\mathbf{z}_i$ is the secondary user for $\mathbf{z}_n$. We denote the ordered pair by  
$\left(\mathbf{z}_i,\mathbf{z}_n\right)$ and it is called an arc on the directed Delaunay graph.
\label{DirectDel}
\end{Def}


\begin{Pro}
The Special Cooperation Scenario in Def. \ref{CoopDefS} is feasible if each vertex of the Delaunay graph 
has exactly one incoming (secondary user) and one outgoing (primary user) arc. The problem coincides 
with the \textbf{Travelling Salesman Problem}, where a node visited cannot be revisited and the salesman 
should proceed necessarily to one of its Delaunay neighbours.
\label{FeasibilityQ}
\end{Pro}
A relevant work in the field from Dillencourt \cite{Dill94} shows that, \textbf{finding Hamiltonian Cycles in Delaunay triangulations is NP-Complete}. A Hamiltonian path 
is a path on an undirected graph, which visits each node exactly once and is also a cycle.


\newpage

\subsection{The Multiple Access Channel with Conferencing Encoders}
\label{AppC}

Let us first study a model in isolation, where two transmitters $\mathbf{z}_1$ and $\mathbf{z}_2$ of a discrete memoryless (d.m.) Multiple Access Channel (MAC)
encode messages to send to a single receiver $\mathbf{u}$. The two transmitters are connected via two separate finite 
capacity directed connections, which can transport at most $C_{12}$ [bits/sec] of information from $\mathbf{z}_1$ to $\mathbf{z}_2$ and 
$C_{21}$ [bits/sec] from $\mathbf{z}_2$ to $\mathbf{z}_1$. \textit{Conferencing} denotes the 
exchange of information between the two encoders. The procedure is summarized as follows:\\
As a first step, each source $i\in\left\{1,2\right\}$ generates a message $M_i$ 
independently and with uniform distribution from a finite set $\mathcal{M}_i \in \left\{1,\ldots,e^{N\cdot R_i}\right\}$. The 
rate of transmission per BS equals $R_i$ nats per channel use.
The encoder considers the information from the other transmitter to produce 
the codeword $\mathbf{x}_i\in\mathcal{X}_i$ from a finite set. The codelength is $N\geq 1$ and thus the 
codewords for channel input are $N$-length vectors $\mathbf{x}_i = \left(x_{i,1},\ldots,x_{i,N}\right)$. The messages are transmitted over the wireless medium and the 
probability of output $\mathbf{y}\in\mathcal{Y}$, $\mathbf{y} = \left(y_1,\ldots,y_N\right)$ at the receiver $\mathbf{u}$ is given by $\mathbb{P}\left(\mathbf{y}|\mathbf{x}_1,\mathbf{x}_2\right)$. The decoding 
function at the receiver uses the information over the output vector $\mathbf{y}$ to produce the message estimates $\hat{M}_i$, $i\in\left\{1,2\right\}$. The capacity 
region of the described channel has been derived by Willems in \cite{WillemsConf83}. Random variables (r.v.'s) are denoted 
by capital letters: $X_1,X_2,Y$ and $Q$.

\begin{Theorem}
(\textbf{Capacity Region of d.m. MAC with Conferencing} \cite{WillemsConf83})\\
For the d.m. MAC denoted by $\left(\mathcal{X}_1\times \mathcal{X}_2,\mathbb{P}\left(\mathbf{y}|\mathbf{x}_1,\mathbf{x}_2\right),\mathcal{Y}\right)$ 
with encoders connected by communication links with capacities $C_{12}$ and $C_{21}$, let us consider a distribution
$\mathbb{P}\left(\mathbf{q},\mathbf{x}_1,\mathbf{x}_2,\mathbf{y}\right)=\mathbb{P}\left(\mathbf{q}\right)\mathbb{P}\left(\mathbf{x}_1|\mathbf{q}\right)\mathbb{P}\left(\mathbf{x}_2|\mathbf{q}\right)\mathbb{P}\left(\mathbf{y}|\mathbf{x}_1,\mathbf{x}_2\right)$, with 
$|\mathcal{Q}| \leq \min\left\{|\mathcal{X}_1|\cdot |\mathcal{X}_2|+2,|\mathcal{Y}|+3\right\}$. 
The capacity region equals the union over all rate regions 
\begin{eqnarray}
\mathcal{C}_{conf}\left(C_{12},C_{21}\right) = \bigcup_{\mathbb{P}\left(\mathbf{q}\right),\mathbb{P}\left(\mathbf{x}_1|\mathbf{q}\right),\mathbb{P}\left(\mathbf{x}_2|\mathbf{q}\right)}\mathcal{R}_{conf}\left(C_{12},C_{21}\right)
\end{eqnarray}
where, 
\begin{eqnarray}
\mathcal{R}_{conf}\left(C_{12},C_{21}\right) & := & \left\{\left(R_1,R_2\right): 0\leq R_1\leq I\left(X_1;Y|X_2,Q\right)+ C_{12},\right.\nonumber\\
& & 0\leq R_2 \leq I\left(X_2;Y|X_1,Q\right)+ C_{21},\nonumber\\
& & 0\leq R_1+R_2\leq \min\left\{ I\left(X_1,X_2;Y|Q\right)+C_{12}+C_{21},\right.\nonumber\\
& & \left.\left.I\left(X_1,X_2;Y\right)\right\}\right\}.
\label{CapRegConf}
\end{eqnarray}
$I\left(X;Y|Z\right) = H\left(X|Z\right)-H\left(X|Y,Z\right)$ is the mutual information of r.v.'s $X$ and $Y$ conditioned on $Z$ and $H\left(X\right)$ is the 
entropy of $X$.
\label{WillemsConfR}
\end{Theorem}

The following \textbf{three special cases} are of particular interest:

\begin{enumerate}
\item The case where $C_{12} = C_{21} = 0$, when there is no cooperation between encoders, results in the capacity region of the MAC \cite[pp.388-407]{ThomasCoverIT} and the 
rate region in (\ref{CapRegConf}) reduces to the region in \cite[p.397, Th.14.3.3]{ThomasCoverIT}.
\item The case where infinite capacity conferencing links are installed between the two transmitters $C_{12} = C_{21} = \infty$, leads to the simplified capacity region 
\begin{eqnarray}
\label{C12C21infCR}
0\leq & R_1+R_2 & \leq \max_{\mathbb{P}\left(\mathbf{q}\right),\mathbb{P}\left(\mathbf{x}_1|\mathbf{q}\right),\mathbb{P}\left(\mathbf{x}_2|\mathbf{q}\right)} I\left(X_1,X_2;Y\right).
\end{eqnarray}
Here, each transmitter can be fully informed over the message of the second source and transmit at a maximal rate equal to the mutual information of $\left(X_1,X_2\right)$ and $Y$ - when the other transmitter sends at zero rate.
\item The case of channels with additive mean zero Gaussian noise at the receiver $\mathbf{u}$.
\end{enumerate}
Case 3 is of particular interest for wireless telecommunication systems. Specifically, consider the channel fading power gains $h_{1}$ and $h_{2}$ from BS $\mathbf{z}_1$ and $\mathbf{z}_2$ respectively, to user $\mathbf{u}$. The received signal at $\mathbf{u}$ equals

\begin{eqnarray}
\label{TransmSign}
\mathbf{y} & = & \sqrt{h_{1}} \cdot \mathbf{x}_{1} + \sqrt{h_{2}} \cdot \mathbf{x}_{2} + \eta
\end{eqnarray}
where the noise $\eta$ is a realization of the r.v. $\eta \sim \mathcal{N}\left(0,\sigma^2\right)$ follows the normal distribution. 
We impose an average power constraint per BS on the transmitted codeword $\mathbf{x}_i\in\mathcal{X}_i$

\begin{eqnarray}
\frac{1}{N}\sum_{n=1}^N x_{i,n}^2 & \leq & P_i
\label{AvPconstr}
\end{eqnarray}
The capacity region of the channel has been derived in \cite{BrossLapWigg08} and is presented in the following:

\begin{Theorem}
(\textbf{Capacity Region of Gaussian MAC with Conferencing \cite{BrossLapWigg08}})\\
The capacity region of the Gaussian MAC with conferencing encoders and power constraints given in (\ref{AvPconstr})
is equal to 

\begin{eqnarray}
\mathcal{C}_{conf,G}\left(C_{12},C_{21}\right) & = & \bigcup_{0\leq \kappa_1,\kappa_2\leq 1} \Bigg\{\left(R_1,R_2\right):\nonumber\\
& & 0\leq R_1 \leq \frac{1}{2}\log\left(1+\frac{h_{1}\kappa_1P_1}{\sigma^2}\right)+ C_{12},\nonumber\\
& & 0\leq R_2 \leq \frac{1}{2}\log\left(1+\frac{h_{2}\kappa_2P_2}{\sigma^2}\right)+ C_{21},\nonumber\\
& & 0\leq R_1+R_2\leq \min\Big\{\frac{1}{2}\log\left(1+\frac{h_{1}\kappa_1P_1+h_{2}\kappa_2P_2}{\sigma^2}\right)+ C_{12}+C_{21},\nonumber\\
& & \frac{1}{2}\log\left(1+\frac{h_{1}P_1+h_{2}P_2+2\sqrt{h_{1}\left(1-\kappa_1\right)P_1\cdot h_{2}\left(1-\kappa_2\right)P_2}}{\sigma^2}\right) \Big\} \Bigg\}.
\label{CapRegConfG}
\end{eqnarray}
\label{GaussConf}
\end{Theorem}

The above theorem indicates something very interesting. In the case of Gaussian channels, splitting the 
power of each transmitter into two parts: (a) A part for a private message and (b) A part for a common message, with a ratio 
defined by the parameter $\kappa_i$, results in rate pairs on the boundary of the MAC region with conferencing. 
This idea of private/common message was also used in the achievability proof of the region in (\ref{CapRegConf}). 
The following proposition shows how the 
Gaussian capacity region can be derived from the d.m. region, in the special Case 2 of 
infinite capacity links between the two encoders and using the above technique.

\begin{Pro}
In the Gaussian channel, for the case $C_{12}=C_{21}=\infty$, the process of encoding a private message and 
a common message at each transmitter by splitting the total transmission power based on the ratio $\kappa_i$ and using superposition coding, 
is capacity achieving. Specifically, the rates satisfy the inequality (\ref{CapRegConfG}, $C_{12}=C_{21}=\infty$):
\begin{eqnarray}
0\leq R_1+R_2\leq \frac{1}{2}\log\left(1+\frac{h_{1}P_1+h_{2}P_2+2\sqrt{h_{1}\left(1-\kappa_1\right)P_1\cdot h_{2}\left(1-\kappa_2\right)P_2}}{\sigma^2}\right)
\end{eqnarray}
\label{CtoGC}
\end{Pro}

\begin{proof}
Let us start by the capacity region of the d.m. channel which is shown in (\ref{C12C21infCR}) 
to be $0\leq R_1+R_2 \leq \max I\left(X_1,X_2;Y\right)$. The maximization is 
taken over all probability distributions, considering also the transmission power constraint in (\ref{AvPconstr}). 
We expand the expression for the information $I$.

\begin{eqnarray}
I\left(X_1,X_2;Y\right) & = & H\left(Y\right) - H\left(Y|X_1,X_2\right)\nonumber\\
						& \stackrel{(i)}{=} 	& H\left(Y\right) - H\left(\sqrt{h_{1}} \cdot X_{1} + \sqrt{h_{2}} \cdot X_{2} + \eta|X_1,X_2\right)\nonumber\\
						& \stackrel{(ii)}{=} 	& H\left(Y\right) - H\left(\eta\right)\nonumber\\
						& \stackrel{(iii)}{=} 	& H\left(Y\right) - \frac{1}{2}\log\left(2\pi e\sigma^2\right)
\label{IG}
\end{eqnarray}
where in (i) we have replaced with (\ref{TransmSign}), in (ii) the conditioning removes the parts $X_1$ and $X_2$ which are fully known, 
while it does not influence the entropy of noise (independent r.v.), while in (iii) the term $H\left(\eta\right)$ is replaced by its 
expression for the differential entropy of the normal distribution. We now calculate the expected received signal power.

\begin{eqnarray}
\label{MomentY}
\mathbb{E}\left[Y^2\right]  & \stackrel{(\ref{TransmSign})}{=}  & \mathbb{E}\left[\left(\sqrt{h_{1}} \cdot X_{1} + \sqrt{h_{2}} \cdot X_{2} + \eta\right)^2\right]\nonumber\\
							& \stackrel{(iv)}{=}								& h_{1}\mathbb{E}\left[X_{1}^2\right] + h_{2}\mathbb{E}\left[X_{2}^2\right]  + \mathbb{E}\left[\eta^2\right] + 2\mathbb{E}\left[\sqrt{h_{1}} X_{1} \cdot \sqrt{h_{2}} X_{2}\right]\nonumber\\
							& \stackrel{(\ref{AvPconstr})}{\leq}				& h_{1}P_1 + h_{2}P_2  + \sigma^2 + 2\sqrt{h_{1}}\sqrt{h_{2}} \cdot \mathbb{E}\left[ X_{1} \cdot X_{2}\right]\nonumber\\
							& \stackrel{(v)}{=}				& h_{1}P_1 + h_{2}P_2  + \sigma^2 + 2\sqrt{h_{1}}\sqrt{h_{2}} \cdot \mathbb{E}\left[ 
							\left(X_{1,p}+X_{1,c}\right) \cdot \left(X_{2,p}+X_{2,c}\right)\right]\nonumber\\
							& \stackrel{(vi)}{=}				& h_{1}P_1 + h_{2}P_2  + \sigma^2 + 2\sqrt{h_{1}}\sqrt{h_{2}} \cdot \mathbb{E}\left[ 
							X_{1,c}\cdot X_{2,c}\right]\nonumber\\
							& \stackrel{(vii)}{\leq}				& h_{1}P_1 + h_{2}P_2  + \sigma^2 + 2\sqrt{h_{1}}\sqrt{h_{2}} \cdot \sqrt{\mathbb{E}\left[ 
							X_{1,c}^2\right]\cdot \mathbb{E}\left[X_{2,c}^2\right]}\nonumber\\
							& \stackrel{(viii)}{=}				& \sigma^2 + h_{1}P_1 + h_{2}P_2 + 2\sqrt{h_{1}\left(1-\kappa_1\right)P_1\cdot h_{2}\left(1-\kappa_2\right)P_2} =: c
\end{eqnarray}
where in (iv) we have expanded the term in parenthesis and considered the fact that $\mathbb{E}\left[X_i\cdot \eta\right] = 0$ since the noise and the 
encoded signal are uncorrelated independent r.v.'s, in (v) we consider the codeword $X_i$ to be the sum of the \textbf{private} and \textbf{common} message 
r.v.'s by applying superposition coding (see also \cite{GantiSuperPo10}, \cite{SuperPoCodACM07} and \cite{MarschFett08}), in (vi) we further expand the 
product in the expectation and consider that the primary and common messages are all uncorrelated r.v.'s, (vii) applies the Cauchy-Schwarz inequality \cite[p.399]{RossBook96} 
and (viii) splits the maximum transmission power, such that $\kappa_i$ is given for the private and $1-\kappa_i$ for the common message. An important notice is that the 
inequality in the application of the Cauchy-Schwarz holds with equality, if and only if $X_{1,c}$ and $X_{2,c}$ are linearly dependent. To guarantee this, we consider $X_{1,c}=\mu \cdot X_{2,c}$, where $\mu = \sqrt{\frac{\left(1-\kappa_1\right)P_1}{\left(1-\kappa_2\right)P_2}}$.

Given (\ref{MomentY}) the entropy of $Y$ is bounded by $H\left(Y\right)\leq \frac{1}{2}\log\left(2\pi e\cdot c\right)$ using \cite[Th.9.6.5]{ThomasCoverIT}, which states 
that the normal distribution maximizes the entropy for a given variance. Applying this inequality to (\ref{IG}) above, we reach the result for some fixed $\kappa_1,\kappa_2$.
\end{proof}

\newpage

\subsection{Motivation on User-optimal Geometric Cooperation Policies}
\label{AppD}

\begin{Pro}
Given a reference user $\mathbf{u}_i$ and a fixed cooperation pattern for the rest of the network $\mathbf{a}_{-i}$, 
it is optimal for the two nearest BSs either to fully cooperate or not to cooperate at all for the user service, depending on the value 
\begin{eqnarray}
\rho & := & \frac{h_{i2}}{h_{i1}}\in\left[0,1\right]
\label{RatioL}
\end{eqnarray}
of their relative channel power ratio. More specifically, 
\begin{eqnarray}
a_i^* & = & \left\{
\begin{tabular}{l l l}
$0$ & $(\mathrm{No\ Coop})$ & , if $\rho\leq 0.1716$\\
$\frac{1}{2}$ & $(\mathrm{Full\ Coop})$ & , if $\rho\geq 0.1716$
\end{tabular}
\right.
\label{CaseOPT}
\end{eqnarray}
\label{Regions}
\end{Pro}

\begin{proof}
To prove the first part of the statement it is easily seen that the signal $S_i\left(a_i,p\right)$ from (\ref{SINRui2a2}) is either monotone increasing or decreasing. 
The first partial derivative of the SINR over $a_i$ gives
\begin{eqnarray}
\frac{\partial \mathrm{SINR}_i\left(a_i,p\right)}{\partial a_i} & = & \left(- h_{i1} + h_{i2} + 2\sqrt{h_{i1}h_{i2}}\right)\frac{p}{\sigma_i^2+ \mathcal{I}_{i}\left(\mathbf{a}_{-i},p\right)}
\end{eqnarray}
and the monotonicity depends on the sign of the expression in parenthesis. For positive sign, the SINR is increasing in $a_i$ and the 
full cooperation is optimal, otherwise no cooperation is optimal. 
\begin{eqnarray}
- h_{i1} + h_{i2} + 2\sqrt{h_{i1}h_{i2}} & \geq 0 & \Rightarrow\nonumber\\
- \frac{1}{\sqrt{\rho}} + \sqrt{\rho} + 2 & \geq 0 & \Rightarrow \nonumber\\
\rho + 2\sqrt{\rho} -1 & \geq 0 &
\label{calculations}
\end{eqnarray}
The roots of the above quadratic polynomial are $-2.4142$ and $0.4142$ and the polynomial is negative inbetween and positive otherwise on $\mathbb{R}$. 
Furthermore, only the positive roots can be accepted since $\sqrt{\rho}$ is non-negative. 
Hence $\rho\geq 0.4142^2 = 0.1716$ results in full cooperation being optimal. Optimality of no cooperation is proved for the negative sign of the 
expression in parenthesis.
\end{proof}

Let us use the above result with a specific expression for the total gain. Specifically, we consider no fast-fading, so that $g\left(\mathbf{z}_n,\mathbf{u}_i\right) = 1$, and the exponential expression for the path-loss, so that $h\left(\mathbf{z}_n,\mathbf{u}_i\right) = d\left(\mathbf{z}_n,\mathbf{u}_i\right)^{-\beta}$, $\beta>2$. The distances of the first and second neighbour to user $\mathbf{u}_i$ are $r_{i1}$ and $r_{i2}$ respectively.

\begin{Cor}
Given that the total gain takes the expression $h\left(\mathbf{z}_n,\mathbf{u}_i\right) = d\left(\mathbf{z}_n,\mathbf{u}_i\right)^{-\beta}$, $\beta>2$, following cooperation result holds

\begin{eqnarray}
a_i^* & = & \left\{
\begin{tabular}{l l l}
$0$ & $(\mathrm{No\ Coop})$ & , if $r_{i1}\leq 0.1716^{1/\beta} \cdot r_{i2}$\\
$\frac{1}{2}$ & $(\mathrm{Full\ Coop})$ & , if $r_{i1}\geq 0.1716^{1/\beta} \cdot r_{i2}$
\end{tabular}
\right.
\label{CaseOPT}
\end{eqnarray}
For e.g. $\beta = 4$, full cooperation is optimal when $r_{i1}\geq 0.6436\cdot r_{i2}$.
\label{coroll1}
\end{Cor}

It is rather interesting in the above that the results hold irrespective of the cooperation strategy chosen by the other users. The decision to cooperate or not, is 
dependent on the user position within the related 2-Voronoi cell $\mathbf{u}_i\in \mathcal{V}_2\left(\mathbf{b}_{i1},\mathbf{b}_{i2}\right)$ and the fading characteristics. 
This can be true, only when we assume that a BS has an unbounded power budget to consume. However, in cases where the number of secondary users $\mathcal{N}^s\left(\mathbf{z}_i\right)$ is large, we cannot expect that the BS can use full cooperation modus for all such users, due to high power consumption expenses. 
It is reasonable then to ask, which is the optimal BS cooperation scheme under specific power constraints per BS.


\newpage

\subsection{Geometric Locus of the Cooperation Policies}
\label{AppE}

\begin{Lem}
Given a fixed parameter $\rho$ and locations of $\mathbf{b}_{i1}$ and $\mathbf{b}_{i2}$ on the 2D plane $\left(x_1,y_1\right)$ and $\left(x_2,y_2\right)$ respectively, 
the geometric locus of points that satisfy $r_{i1}\leq \rho \cdot r_{i2}$ for $\rho\in\left(0,1\right)$ 
is a \textbf{disc}, whose circle describing the boundary is given by the equation
\begin{eqnarray}
\left(x-\frac{x_1-x_2\rho^2}{1-\rho^2}\right)^2+ \left(y-\frac{y_1-y_2\rho^2}{1-\rho^2}\right)^2 & = & \left(\rho\frac{d\left(\mathbf{b}_{i1},\mathbf{b}_{i2}\right)}{1-\rho^2}\right)^2
\label{DiscCoop}
\end{eqnarray}
and has center $\left(\frac{x_1-x_2\rho^2}{1-\rho^2},\frac{y_1-y_2\rho^2}{1-\rho^2}\right)$ and radius $\rho \frac{d\left(\mathbf{b}_{i1},\mathbf{b}_{i2}\right)}{1-\rho^2}$, where 
$d\left(\mathbf{b}_{i1},\mathbf{b}_{i2}\right)$ is the Euclidean distance between $\mathbf{b}_{i1}$ and $\mathbf{b}_{i2}$.

For $\rho=1$ the locus of points degenerates to the line which passes over the 1-Voronoi boundary of the two cells
\begin{eqnarray}
x\left(x_2-x_1\right)+y \left(y_2-y_1\right) & = & \frac{1}{2}\left(x_2^2-x_1^2 + y_2^2-y_1^2\right).
\label{LineCoop}
\end{eqnarray}
\label{ProGL}
\end{Lem}

\begin{proof}
The expression $r_1=\rho r_2$ can be rewritten by using the expressions on the Euclidean distances (\ref{distanceM}) between a point $\left(x,y\right)$ of the 
2D plane and the positions of the BSs $\mathbf{b}_{i1}$ and $\mathbf{b}_{i2}$, 

\begin{eqnarray}
d\left(\mathbf{b}_{i1},\left(x,y\right)\right) & = & \rho d\left(\mathbf{b}_{i2},\left(x,y\right)\right)\nonumber\\
\left(x_1-x\right)^2 + \left(y_1-y\right)^2 & = & \rho^2 \left(\left(x_2-x\right)^2 + \left(y_2-y\right)^2\right)\nonumber\\
\left(1-\rho^2\right)x^2 - 2x\left(x_1-x_2\rho^2\right) + \left(1-\rho^2\right)y^2 - 2y\left(y_1-y_2\rho^2\right) & = & x_2^2\rho^2-x_1^2+ y_2^2\rho^2-y_1^2.\nonumber
\label{GL1}
\end{eqnarray}
Simply by setting $\rho = 1$ in the last expression, we reach (\ref{LineCoop}). For $\rho\neq 1$ we can divide by $1-\rho^2$. We further add at both sides of the 
equality the term $\left(\frac{x_1-x_2\rho^2}{1-\rho^2}\right)^2+ \left(\frac{y_1-y_2\rho^2}{1-\rho^2}\right)^2$ so that expressions of squares of sums are formed and we reach

\begin{eqnarray}
\left(x-\frac{x_1-x_2\rho^2}{1-\rho^2}\right)^2 + \left(y-\frac{y_1-y_2\rho^2}{1-\rho^2}\right)^2 & = & \rho^2\frac{\left(x_1-x_2\right)^2+\left(y_1-y_2\right)^2}{\left(1-\rho^2\right)^2}.\nonumber
\label{GL2}
\end{eqnarray}
This is obviously the equation describing a circle with center and radius given in the proposition. To show why the degenerate case for $\rho=1$ is a line which passes 
over the 1-Voronoi boundary we argue as follows. The 1-Voronoi boundary is by Def. \ref{DefV1} the perpendicular bisector of the line segment $\overline{\mathbf{b}_{i1}\mathbf{b}_{i2}}$, so that 
a point on the boundary is $\left(x_1+\frac{x_2-x_1}{2}, y_1+\frac{y_2-y_1}{2}\right)$. Take any other point $\left(x,y\right)$ on the boundary. The tangent of 
the vector defined by these two points, equals the tangent of the line segment $\overline{\mathbf{b}_{i1}\mathbf{b}_{i2}}$ rotated by $\pi/2$

\begin{eqnarray}
\frac{y-\frac{y_2+y_1}{2}}{x-\frac{x_2+x_1}{2}} & = & -\frac{1}{\frac{y_2-y_1}{x_2-x_1}} \nonumber\\
y\cdot\left(y_2-y_1\right) - \frac{y_2^2-y_1^2}{2} & = & -x\cdot\left(x_2-x_1\right) + \frac{x_2^2-x_1^2}{2}.\nonumber
\end{eqnarray}
The above expression coincides with the one in (\ref{LineCoop}).
\end{proof}

\newpage

\subsection{Proofs and Supplementary material for distance p.d.f.}
\label{AppF}

To derive the p.d.f. of the distances to the first and second BS neighbour of the typical location, we will make use of Prop. \ref{CircleV1} and Prop. \ref{CircleV2} from Appendix \ref{AppA}. Based on these, to determine the first closest neighbour, 
we should find the largest ball $\mathcal{B}\left(\mathbf{u}_o,r_1\right)$ of radius $r_1$ which is empty. This ball meets the first neighbour on its 
boundary. This implies by Prop. \ref{CircleV1} that $\mathbf{u}$ belongs to its 1-Voronoi cell. We further enlarge the ball until the second closest 
neighbour is met 
on the boundary of $\mathcal{B}\left(\mathbf{u}_o,r_2\right)$, $r_2\geq r_1$ and the first (and no other) is contained in its interior, so that by Prop. \ref{CircleV2} $\mathbf{u}_o$ belongs to the 2-Voronoi cell of these two BSs. In this way, the annulus starting from the ball of radius 
$r_1$ and reaching the ball of radius $r_2$ is the largest empty.

\begin{Lem}
Given a Poisson p.p. of intensity $\lambda$, the p.d.f. of the distance $r_1$ between the typical location $\mathbf{u}_o$ and its first closest neighbour, 
equals
\begin{eqnarray}
f_{r_1}\left(r_1\right) & = & 2\lambda\pi r_1 e^{-\lambda \pi r_1^2}.
\label{pdfN1f}
\end{eqnarray}
\label{pdfN1}
The expected value of the distance $r_1$ equals
\begin{eqnarray}
\mathbb{E}\left[r_1\right] & = & \frac{1}{2\sqrt{\lambda}}.
\label{ExpR1}
\end{eqnarray}
\label{Lemr1}
\end{Lem}

\begin{proof}
We can derive the p.d.f. from the cumulative distribution of $r_1$, in other words 
from the expression of the probability $\mathbb{P}\left[s\leq r_1\right]$. Here $s$ is the value of the stopping time when hitting the first neighbour.
\begin{eqnarray}
1 - \mathbb{P}\left[s\leq r_1\right] & = & \mathbb{P}\left[s>r_1\right]\nonumber\\
									 & = & \mathbb{P}\left[\phi\left(\mathcal{B}\left(\mathbf{u}_o,r_1\right)\right)=0\right]\nonumber\\
									 & = & e^{-\lambda\pi r_1^2}.\nonumber
\end{eqnarray}
So that, 
\begin{eqnarray}
\mathbb{P}\left[s\leq r_1\right] = 1 -e^{-\lambda\pi r_1^2}& \Rightarrow & f_{r_1}\left(r_1\right) = 2\lambda\pi r_1e^{-\lambda\pi r_1^2}.\nonumber
\end{eqnarray}
For the expected value, we have that
\begin{eqnarray}
\mathbb{E}\left[r_1\right] = \int_0^{\infty} \! \mathbb{P}\left[s> r_1\right] \, \mathrm{d} r_1 & = &  \int_0^{\infty} \! e^{-\lambda\pi r_1^2} \, \mathrm{d} r_1 \nonumber\\
& = & \left.\frac{\sqrt{\pi}}{2}\frac{1}{\sqrt{\lambda\pi}}\cdot erf\left(r_1\sqrt{\lambda\pi}\right)\right\vert_0^{\infty}\nonumber\\
& = & \frac{1}{2\sqrt{\lambda}}.\nonumber
\label{ExpR1c}
\end{eqnarray}
\end{proof}

\begin{proof}\textbf{[Lemma \ref{LemJointD}]} We denote by $s$ and $t$ the stopping times for hitting the first and second neighbour respectively. 
The cumulative distribution of $t$, given $s$, equals
\begin{eqnarray}
1 - \mathbb{P}\left[t\leq r_2|s= r_1,\ r_2\geq r_1\right] & = & \mathbb{P}\left[t>r_2|s= r_1,\ r_2\geq r_1\right]\nonumber\\
									 & = & \mathbb{P}\left[\phi\left(\mathbf{int}\left(\mathcal{B}\left(\mathbf{u}_o,r_1\right)\cap\mathcal{B}\left(\mathbf{u}_o,r_2\right)\right)\right)=0\right]\nonumber\\
									 & = & e^{-\lambda\pi \left(r_2^2-r_1^2\right)}.\nonumber
\end{eqnarray}
So that,
\begin{eqnarray}
\mathbb{P}\left[t\leq r_2|s= r_1,\ r_2\geq r_1\right] = 1 - e^{-\lambda\pi \left(r_2^2-r_1^2\right)} & \Rightarrow & f_{r_2|r_1}\left(r_2|r_1\right) = 2\lambda\pi r_2e^{-\lambda\pi \left(r_2^2-r_1^2\right)}\nonumber
\end{eqnarray}
and using Bayes' rule for the p.d.f.'s
\begin{eqnarray}
f_{r_1,r_2}\left(r_1,r_2\right) = f_{r_2|r_1}\left(r_2|r_1\right)\cdot f_{r_1}\left(r_1\right) & = & \left(2\lambda\pi\right)^2 r_1 r_2e^{-\lambda\pi r_2^2}.\nonumber
\end{eqnarray}
Integrating the joint p.d.f. over $r_2$ we get simply the p.d.f. of the distance $r_1$ given in (\ref{pdfN1f})
\begin{eqnarray}
\label{pdfN2s}
\int_{r_1}^{\infty} \! \left(2\lambda\pi\right)^2r_1 r_2 e^{-\lambda\pi r_2^2} \, \mathrm{d} r_2 & = & 2\lambda\pi r_1 \left(\int_{r_1}^{\infty} \! \frac{\mathrm{d} e^{-\lambda\pi r_2^2}}{\mathrm{d} r_2} \, \mathrm{d} r_2\right)\nonumber\\
& = & 2\lambda\pi r_1e^{-\lambda\pi r_1^2}.
\end{eqnarray}
Otherwise, integrating the joint p.d.f. over $r_1$ we get
\begin{eqnarray}
\label{pdfN2f2}
\int_{0}^{r_2} \! \left(2\lambda\pi\right)^2r_1 r_2 e^{-\lambda\pi r_2^2} \, \mathrm{d} r_1 & = & \left(2\lambda\pi\right)^2 \frac{r_2^3}{2}e^{-\lambda\pi r_2^2}.\nonumber
\end{eqnarray}
Now we can calculate the expected value of $r_2$

\begin{eqnarray}
\label{ExpR2}
\mathbb{E}\left[r_2\right] & = & \int_0^{\infty} \! r_2 \int_0^{r_2}\! \left(2\lambda\pi\right)^2r_1 r_2 e^{-\lambda\pi r_2^2} \, \mathrm{d} r_1 \,  \mathrm{d} r_2\nonumber\\
& = & \int_0^{\infty} \! r_2\left(2\lambda\pi\right)^2 \frac{r_2^3}{2}e^{-\lambda\pi r_2^2}\,  \mathrm{d} r_2\nonumber\\
& \stackrel{(a)}{=} & \frac{3}{2}\cdot \int_0^{\infty} \! e^{-\lambda\pi r_2^2} \, \mathrm{d} r_2\nonumber\\
& = & \frac{3}{2\sqrt{\lambda\pi}}\frac{\sqrt{\pi}}{2}\left.erf\left(r_2\sqrt{\lambda\pi}\right)\right\vert_0^{\infty} = \frac{3}{4\sqrt{\lambda}},\nonumber
\end{eqnarray}
where (a) results after consecutive applications of integration by parts (product rule). It is obvious that $\mathbb{E}\left[r_2\right]>\mathbb{E}\left[r_1\right]$.
\end{proof}

\begin{proof}\textbf{[Lemma \ref{Lemrho}]}
To calculate the probability, we will use the expression for the joint p.d.f. of the pair of distances $\left(r_1,r_2\right)$ from the 
typical location to the first and second closest atom of the point process. This is provided in (\ref{pdfN2f}), Lemma \ref{LemJointD}. The probability equals

\begin{eqnarray}
\label{CalcPnocoop}
\mathbb{P}\left[r_1\leq \rho r_2\right] 	& = & \int_0^{\infty}\! \int_{r_1}^{\infty}\! \mathbbm{1}_{\left\{r_1\leq \rho r_2\right\}} \left(2\lambda\pi\right)^2r_1r_2e^{-\lambda\pi r_2^2}\, \mathrm{d} r_2 \, \mathrm{d} r_1\nonumber\\
												& = & \int_0^{\infty}\! \int_{\frac{r_1}{\rho}}^{\infty}\! \left(2\lambda\pi\right)^2r_1r_2e^{-\lambda\pi r_2^2}\, \mathrm{d} r_2 \, \mathrm{d} r_1\nonumber\\
												& = & \int_0^{\infty}\! 2\lambda\pi r_1\int_{\frac{r_1}{\rho}}^{\infty}\! 2\lambda\pi r_2e^{-\lambda\pi r_2^2}\, \mathrm{d} r_2 \, \mathrm{d} r_1\nonumber\\
												& = & \int_0^{\infty}\! 2\lambda\pi r_1 e^{-\lambda\pi \left(\frac{r_1}{\rho}\right)^2} \, \mathrm{d} r_1\nonumber\\
												& = & \rho^2\int_0^{\infty}\! -\frac{d \left(e^{-\lambda\pi \left(\frac{r_1}{\rho}\right)^2}\right)}{d\left(\frac{r_1}{\rho}\right)} \, \mathrm{d} \left(\frac{r_1}{\rho}\right) = \rho^2\nonumber
\end{eqnarray}
\end{proof}
To intuitively understand the result, on Lem. \ref{Lemrho} observe that the ratios of the two balls $\frac{\mathcal{B}\left(\mathbf{u}_o,r_1\right)}{\mathcal{B}\left(\mathbf{u}_o,r_2\right)}=\frac{\pi r_1^2}{\pi r_2^2}\leq \rho^2$ is smaller or equal to the probability of no cooperation. This means that $\rho^2$ equals the spatial average of the area where no cooperation takes place, over the entire area where both events (no/full coop) occur, with regards to the typical location.

\newpage

\subsection{Proofs and Supplementary Material for Channel Fading Distributions}

\begin{proof}\textbf{[Lemma \ref{LemZ}]}
For $i=1,2$, set $\mu_i:=\frac{r_i^{\beta}}{p}$, so that the r.v. $X_i:=g_i r_i^{-\beta}$ follows 
the exponential distribution $X_i\sim \exp\left(\mu_i\right)$ with $\mathbb{E}\left[X_i\right]=\frac{1}{\mu_i}$. 
We can easily conclude that the r.v. $\sqrt{X_i}$ follows the Rayleigh distribution with p.d.f.
\begin{eqnarray}
\label{pdfR}
f_{\sqrt{X_i}}\left(u\right) & = & 2\mu_i u e^{-\mu_i u^2}\nonumber
\end{eqnarray}
and expected value $\mathbb{E}\left[\sqrt{X_i}\right]=\sqrt{\frac{\pi}{4\mu_i}}$. The p.d.f. of the sum of the 
two independent r.v.'s $\sqrt{Z} = \sqrt{X_1}+\sqrt{X_2}$ is given by the convolution of their individual p.d.f.'s, namely
\begin{eqnarray}
\label{ConvX1X2a}
f_{\sqrt{Z}}\left(v\right) & = & 4\mu_1\mu_2 \int_0^{v} \!  u\left(v-u\right) e^{-\mu_1 u^2} e^{-\mu_2 \left(v-u\right)^2} \, \mathrm{d} u\nonumber.
\end{eqnarray}
Since we are interested in the p.d.f. of $Z$ rather than of $\sqrt{Z}$ we make following observation. The first derivative of the 
cumulative distribution function (c.d.f.) of a r.v. $Y$ is its p.d.f. $f_{Y}\left(y\right)$. Then 
\begin{eqnarray}
\label{replaceZ}
\mathbb{P}\left[Y^2\leq y\right] = \mathbb{P}\left[Y\leq\sqrt{y}\right] & \Rightarrow & f_{Y^2}\left(y\right) = \frac{\mathrm{d} \mathbb{P}\left[Y\leq\sqrt{y}\right]}{\mathrm{d} y} = \frac{f_Y\left(\sqrt{y}\right)}{2\sqrt{y}}.\nonumber
\end{eqnarray}
With the above transformation, we can get the p.d.f. of $Z$ directly from the expression of the convolution
\begin{eqnarray}
\label{ConvX1X2}
f_Z\left(z\right) = \frac{f_{\sqrt{Z}}\left(\sqrt{z}\right)}{2\sqrt{z}} & = & \frac{4\mu_1\mu_2}{2\sqrt{z}} \int_0^{\sqrt{z}} \!  u\left(\sqrt{z}-u\right) e^{-\mu_1 u^2} e^{-\mu_2 \left(\sqrt{z}-u\right)^2} \, \mathrm{d} u.
\end{eqnarray}
It can be verified that the p.d.f. is square integrable. Finally, the LT of $Z$ equals
\begin{eqnarray}
\label{LTZf}
\mathcal{L}_Z\left(s\right) & := & \int_0^{\infty} \! e^{-s z}f_Z\left(z\right) \,  \mathrm{d} z\nonumber\\
& = & \frac{-s\sqrt{\frac{1}{\mu_1\mu_2}}\pi +s\sqrt{\frac{1}{\mu_1\mu_2}} \arctan\left(\sqrt{\frac{\mu_1}{\mu_2}}g\left(s\right)\right)+s\sqrt{\frac{1}{\mu_1\mu_2}} \arctan\left(\sqrt{\frac{\mu_2}{\mu_1}}g\left(s\right)\right)+ g\left(s\right)}{g\left(s\right)^3}\nonumber
\end{eqnarray}
where $g\left(s\right) := \sqrt{1+\left(\frac{1}{\mu_1}+\frac{1}{\mu_2}\right) s}$. We can check that the resulting p.d.f. of $Z$ is actually valid, since $\mathcal{L}_Z\left(0\right) = \int_0^{\infty} \! f_Z\left(z\right) \,  \mathrm{d} z = 1$. The expected value of $Z$ is found directly by the first derivative at $s=0$
\begin{eqnarray}
\label{Zpdf}
\mathbb{E}\left[Z\right] & = & -\left.\frac{\mathrm{d} \mathcal{L}\left(s\right)}{\mathrm{d} s}\right|_{s=0}.\nonumber
\end{eqnarray}
The expression is provided in (\ref{EZpdf}) and a direct evaluation for $\mu_1=\mu_2=1/2$ gives $\mathbb{E}\left[Z\right]=\pi+4$. We further plot the p.d.f. for values $\mu_1=\mu_2=2$  in Fig. \ref{fig:Zplot}.

\begin{figure}[h]
\centering
\includegraphics[trim = 0mm 0mm 0mm 0mm, clip, width=0.35\textwidth]{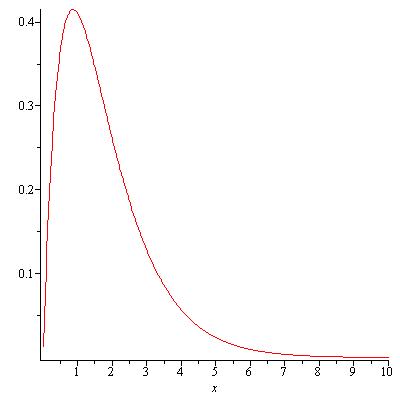}
\caption{The p.d.f. of the random variable $Z$ with $\mu_1=\mu_2=2$.}
\label{fig:Zplot}
\end{figure}
\end{proof}

\begin{Lem}
\label{StochOrd}
Given $\mu_1=\mu_2=\mu:=\frac{r^{\beta}}{p}$, and the two r.v.'s $G\sim\exp\left(\mu\right)$ and $Z_{r,r}/2$ from (\ref{RVz}), the following stochastic ordering inequality holds
\begin{eqnarray}
\label{SOineq}
\begin{tabular}{l l l l}
If $\mu< 1$ & $\Rightarrow$ & $G\leq_{st}\frac{Z_{r,r}}{2}$ & \\
If $\mu\geq 1$ & $\Rightarrow$ & $\mathbb{P}\left[Z_{r,r}>2 t_l\right] \leq \mathbb{P}\left[G>t_l\right]$, & $t_l\leq t^*$\\
& & $\mathbb{P}\left[Z_{r,r}>2 t_h\right] \geq \mathbb{P}\left[G>t_h\right]$, & $t_h> t^*$
\end{tabular}
\end{eqnarray}
for some $t^*<\infty$.
\end{Lem}

\begin{proof}
Based on the p.d.f. expression of $Z_{r,r}$ in (\ref{ConvX1X2}) and by integrating from $2t,\ldots,\infty$ we get
\begin{eqnarray}
\label{CDFz}
\mathbb{P}\left[Z_{r,r}>2t\right] & = & \frac{e^{-2\mu t}}{\mu^4}\left( e^{\mu t} \sqrt{\pi\mu t} \cdot erf\left(\sqrt{\mu t}\right)+ 1\right).\nonumber
\end{eqnarray}
Furthermore, the inverse c.d.f. for the r.v. $G$ equals
\begin{eqnarray}
\label{CDFg}
\mathbb{P}\left[G>t\right] & = & e^{-\mu t}.\nonumber
\end{eqnarray}
Let us take the difference of the above expressions which gives
\begin{eqnarray}
\label{DiffCDF}
\mathbb{P}\left[Z_{r,r}>2t\right] - \mathbb{P}\left[G>t\right] & = & e^{-\mu t}\left[\frac{e^{-\mu t}}{\mu^4}\left(e^{\mu t}\sqrt{\pi \mu t} \cdot erf\left(\sqrt{\mu t}\right)+1\right)-1\right].
\end{eqnarray}
For $t=0$ the difference equals $\frac{1}{\mu^4}-1$ and for $\mu\leq 1$ the inequality in (\ref{SOineq}) holds true. Furthermore, for $t\rightarrow \infty$ the above difference goes to $0$. We will show that under the same condition for the $\mu$ there are no level crossings, in other words, there exists no $t<\infty$ such that

\begin{eqnarray}
\label{LevC1}
\frac{e^{-\mu t}}{\mu^4}\left(e^{\mu t}\sqrt{\pi \mu t} \cdot erf\left(\sqrt{\mu t}\right)+1\right) = & 1 & \Rightarrow \nonumber\\
\sqrt{\pi \mu t} \cdot erf\left(\sqrt{\mu t}\right) = & \mu^4-e^{-\mu t}.\nonumber
\end{eqnarray}
The derivative of the function in the left handside (LHS) equals $\mu e^{-\mu t} + \frac{\sqrt{\pi}}{2\sqrt{\mu t}}\cdot erf\left(\sqrt{\mu t}\right)$ and for all $t$ is greater than the derivative of the function at the right handside (RHS) which equals $\mu e^{-\mu t}$. Both functions are strictly increasing. So if we show that the starting point of the RHS is strictly smaller than that of the LHS we have shown that the two curves do not meet. For $t=0$ the LHS equals $0$, whereas the RHS equals $\mu^4-1$ and under the condition $\mu\leq 1$ we have that indeed the initial point of the RHS is non-positive, so that there are no level crossings definitely for $\mu<1$. The difference in (\ref{DiffCDF}) is plotted as example in Fig \ref{fig:SO1} for $\mu<1$.

On the other hand, for $t=0$ and $\mu\geq 1$, the difference in (\ref{DiffCDF}) is non-positive and $\mathbb{P}\left[Z_{r,r}>0\right] \leq \mathbb{P}\left[G>0\right]$. Using the same arguments as above, since $\mu^4-1$ is positive for the $RHS$, there will be exactly one $t^*<\infty$ such that a level crossing occurs. That means for $t_l<t^*$: $\mathbb{P}\left[Z_{r,r}>2 t_l\right] \leq \mathbb{P}\left[G>t_l\right]$, whereas for $t_h>t^*$: $\mathbb{P}\left[Z_{r,r}>2 t_h\right] \geq \mathbb{P}\left[G>t_h\right]$. An example is plotted in Fig. \ref{fig:SO1} for $\mu>1$.

\begin{figure}[ht]    
\centering  
\label{fig:ProofSO}
	 		\subfigure[The stochastic ordering difference for $\mu=1/2<1$.]{          
           \includegraphics[trim = 0mm 0mm 0mm 0mm, clip, width=0.4\textwidth]{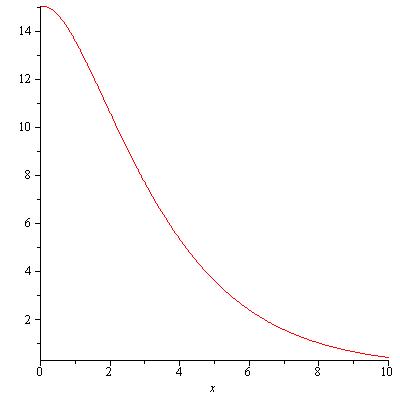}
           \label{fig:SO1}
           }
            \subfigure[The stochastic ordering difference for $\mu=1.2>1$.]{          
           \includegraphics[trim = 0mm 0mm 0mm 0mm, clip, width=0.4\textwidth]{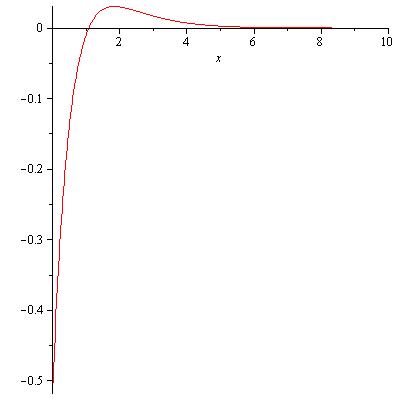}
           \label{fig:SO2}
           }
           \caption{Illustrative plots related to the Proof of Lemma \ref{StochOrd}.}
\end{figure}

%

\end{proof}

\begin{proof}\textbf{[Lemma \ref{LaplaceOrd}]}
We first establish a simple result. It can be easily shown that for the p.d.f. $f_{\frac{Z}{2}}\left(t\right)=2f_{Z}\left(2t\right)$. Taking the Laplace transform and using the 
above equality for the p.d.f.'s we get 
$\mathcal{L}_{\frac{Z}{2}}\left(s\right)=\mathbb{E}_{\frac{Z}{2}}\left[e^{-st}\right]=\mathbb{E}_Z\left[e^{-\frac{s}{2}2t}\right]=\mathcal{L}_Z\left(\frac{s}{2}\right)$.

The LT of the r.v. $G$ equals 

\begin{eqnarray}
\label{LTG1}
\mathcal{L}_G\left(s\right) & = & \frac{1}{1+sp}\nonumber
\end{eqnarray}
while, the LT of the r.v. $\frac{Z_{r,r}}{2}$ is found using the above observation and the form in (\ref{LTz})
\begin{eqnarray}
\label{LTZo2}
\mathcal{L}_{Z_{r,r}/2}\left(s\right) 	& = & \mathcal{L}_{Z}\left(\frac{s}{2},\mu,\mu\right)\nonumber\\
										& = & \frac{1}{1+sp}\left(1+sp\frac{-\pi/2+\arctan\left(\sqrt{1+sp}\right)}{\sqrt{1+sp}}\right)\nonumber.
\end{eqnarray}
The range of $\arctan\left(\sqrt{1+sp}\right)\in\left[\pi/4,\pi/2\right]$ and hence the term in the parenthesis is $(\cdot)\leq1$. This concludes the proof.

\end{proof}

\newpage 

\subsection{Proofs and Supplementary Material for the Interference random variable}

\begin{proof}\textbf{[Theorem \ref{THlaplI}]}
From the definition of the Laplace Transform, 
\begin{eqnarray}
\label{CalculLTi}
\mathcal{L}_{\mathcal{I}}\left(s,\rho,r_2\right) 	& = & \mathbb{E}_{\mathcal{I}}\left[e^{-s\mathcal{I}}\right]\nonumber\\
													& = & \mathbb{E}_{\phi,\left\{G_n\right\},\left\{G_{n1},G_{n2}\right\},\left\{B_n\right\}}\left[e^{-s\left(r_{2}^{-\beta} G_{2}B_2 + r_2^{-\beta}\frac{G_1+G_2}{2} \left(1-B_2\right)\right)} \prod_{\mathbf{z}_n\in\phi\setminus\left\{\mathbf{b}_1,\mathbf{b}_2\right\}}e^{-s\left(d_{n}^{-\beta}G_{n} B_n + d_n^{-\beta}\frac{G_{n1}+G_{n2}}{2} \left(1-B_n\right)\right)}\right]\nonumber
\end{eqnarray}
Observe now that the random variables $B_n$, $G_n$, $G_{n1},G_{n2}$ are independently distributed and independent from the point process $\phi$. Furthermore, the contribution of each atom outside the ball of radius $r_2$ is 
\begin{eqnarray}
\label{EachI}
\mathcal{L}_{\mathcal{J}}\left(s,\rho,d_n\right)& := & \mathbb{E}_{B_n}\left[\mathbb{E}_{G_n}\left[e^{-sd_{n}^{-\beta}G_{n} B_n}\right]\cdot\mathbb{E}_{G_{n1},G_{n2}}\left[e^{-sd_n^{-\beta}\frac{G_{n1}+G_{n2}}{2} \left(1-B_n\right)}\right]\right]\nonumber\\
& = & \rho^2\mathbb{E}_{G_n}\left[e^{-sd_{n}^{-\beta}G_{n}}\right] + \left(1-\rho^2\right)\mathbb{E}_{G_{n1},G_{n2}}\left[e^{-sd_n^{-\beta}\frac{G_{n1}+G_{n2}}{2}}\right]\nonumber\\
& = & \rho^2\mathcal{L}_G\left(sd_{n}^{-\beta}\right)+ \left(1-\rho^2\right)\left(\mathcal{L}_G\left(\frac{sd_{n}^{-\beta}}{2}\right)\right)^2\nonumber\\
& \stackrel{(a)}{=} & \rho^2\frac{1}{1+sd_n^{-\beta}p}+\left(1-\rho^2\right)\frac{1}{\left(1+sd_n^{-\beta}\frac{p}{2}\right)^2}\nonumber
\end{eqnarray}
where (a) comes fromt the LT of the random variable $G$, $G_n, G_{n1}, G_{n2} \sim\exp\left(1/p\right)$. Similarly for the atom on the boundary of the ball, the contribution is given from the above expression by replacing $d_n$ by $r_2$. The LT of the interference finally takes the expression
\begin{eqnarray}
\label{LTi2}
\mathcal{L}_{\mathcal{I}}\left(s,\rho,r_2\right) 	& = & \mathcal{L}_{\mathcal{J}}\left(s,\rho,r_2\right)\cdot\mathbb{E}_{\phi}\left[\prod_{\mathbf{z}_n\in\phi\setminus\left\{\mathbf{b}_1,\mathbf{b}_2\right\}}\mathcal{L}_{\mathcal{J}}\left(s,\rho,d_n\right)\right]\nonumber\\
&\stackrel{(b)}{=} & \mathcal{L}_{\mathcal{J}}\left(s,\rho,r_2\right)\cdot e^{-2\pi\lambda\int_{r_2}^{\infty} \!  \left(1-\mathcal{L}_{\mathcal{J}}\left(s,\rho,r\right)\right)r \, \mathrm{d} r }\nonumber
\end{eqnarray}
where (b) comes from applying the Laplace functional expression for the p.p.p. using polar coordinates \cite[Proposition 1.5, pp. 18-19]{BaccelliBookStoch}, for radius $r$ ranging from $r_2$ to $\infty$, since it has already been identified that the second closest neighbour lies at the boundary of the ball of radius $r_2$.
\end{proof}

\begin{Cor}
The LT of the Interference random variable $\mathcal{L}_{\mathcal{I}}\left(s,\rho,r_2\right)$ in (\ref{LTinterfD}) is an increasing (non-decreasing) function in $\rho$ and $r_2$ and decreasing (non-increasing) function in $s$. 
\end{Cor}

\begin{proof}
We first establish the monotonicity in $\rho$. Let us first consider the LT $\mathcal{L}_{\mathcal{J}}$ in (\ref{EachID}), found in the above proof

\begin{eqnarray}
\label{LJ1}
\mathcal{L}_{\mathcal{J}}\left(s,\rho,r\right) & = &  \rho^2\frac{1}{1+sd_n^{-\beta}p}+\left(1-\rho^2\right)\frac{1}{\left(1+sd_n^{-\beta}\frac{p}{2}\right)^2}\Rightarrow\nonumber\\
\frac{\partial \mathcal{L}_{\mathcal{J}}\left(s,\rho,r\right)}{\partial \rho} & = & 2\rho\left(\frac{1}{1+sd_n^{-\beta}p}-\frac{1}{\left(1+sd_n^{-\beta}\frac{p}{2}\right)^2}\right)\nonumber\\
& = & 2\rho\frac{s^2d_n^{-2\beta}\frac{p^2}{4}}{\left(1+sd_n^{-\beta}p\right)\left(1+sd_n^{-\beta}\frac{p}{2}\right)^2}\geq 0.\nonumber
\end{eqnarray}
Hence $\mathcal{L}_{\mathcal{J}}$ is increasing in $\rho$. From the expression in (\ref{LTinterfD}) and using the above result on $\mathcal{L}_{\mathcal{J}}$ it is easily seen that the partial derivative over $\rho$ is also non-negative.

Considering $s$, the monotone decreasing behavior is a direct result of the definition of the Laplace transform $\mathbb{E}_{\mathcal{I}}\left[e^{-s\mathcal{I}}\right]$.

The $r_2$ monotonicity comes from the following arguement. Variable $r_2$ is found in the $\mathcal{L}_{\mathcal{J}}\left(s,\rho,r_2\right)$ expression only as a product of the form $sr_2^{-\beta}$, hence $\mathcal{L}_{\mathcal{J}}\left(s,\rho,r_2\right)$ is increasing in $r_2$ (Laplace transform). In the exponential term, $r_2$ is the lower limit of the integral and does not appear at some other position. Hence, the larger the $r_2$ the smaller the value of the integral, since the quantity $1-\mathcal{L}_{\mathcal{J}}\left(s,\rho,r\right)$ can be  shown to be non-negative. As a result, the value of the negative exponential term is non-decreasing with increasing $r_2$ and this concludes the proof. 
\end{proof}
The LT of the interference takes its maximal value for $\rho=1$, where $\mathrm{No\ Coop}$ is applied in the entire Voronoi cell of each interferer. This is reasonable based on the Laplace-Stieltjes transform ordering. The argument is that for any $0\leq\rho_a<\rho_b\leq 1$, $\mathcal{L}_{\mathcal{I}}\left(s,\rho_a,r_2\right)\leq\mathcal{L}_{\mathcal{I}}\left(s,\rho_b,r_2\right)$ $\Rightarrow$ $\mathcal{I}\left(\rho_a,r_2\right)\geq_L \mathcal{I}\left(\rho_b,r_2\right)\geq_L\mathcal{I}\left(1,r_2\right)$. The larger the $\rho$ the less the interference. Similarly, the larger the distance to the second neighbour $r_2$, the less the interference, because a larger empty ball of interferers around the typical location is guaranteed.

\begin{Cor}
\label{CorLTIrho1}
The LT of the Interference random variable $\mathcal{L}_{\mathcal{I}}\left(s,\rho,r_2\right)$ in (\ref{LTinterfD}) for the case of $\rho=1$ $(\mathrm{No\ Coop})$ and path-loss exponent $\beta=4$ takes the expression
\begin{eqnarray}
\label{LTInocoop}
\left.\mathcal{L}_{\mathcal{I}}\left(s,1,r_2\right)\right|_{\beta=4} & = & \frac{1}{1+sr_2^{-4}p}e^{-\pi\lambda\sqrt{sp}\left(\pi/2-\arctan\left(\frac{r_2^2}{\sqrt{sp}}\right)\right)}
\end{eqnarray} 
\end{Cor}

\begin{proof}
For the case of $\rho=1$ we have that 
\begin{eqnarray}
\mathcal{L}_{\mathcal{J}}\left(s,1,r\right) & = & \frac{1}{1+sr^{-\beta}p}.\nonumber
\end{eqnarray}
Using this expression, the integral on the exponent in (\ref{LTinterfD}) can be written as
\begin{eqnarray}
\int_{r_2}^{\infty} \!  \left(1-\mathcal{L}_{\mathcal{J}}\left(s,1,r\right)\right)r \, \mathrm{d} r & = & \int_{r_2}^{\infty} \!  \frac{sr^{-\beta}p}{1+sr^{-\beta}p}r \, \mathrm{d} r \nonumber\\
& \stackrel{u=\left(\frac{r}{(sp)^{1/{\beta}}}\right)^2}{=} & \frac{\left(sp\right)^{2/\beta}}{2}\int_{\left(\frac{r_2}{(sp)^{1/\beta}}\right)^2}^{\infty} \!  \frac{u^{-\beta/2}}{1+u^{-\beta/2}} \, \mathrm{d} u\nonumber\\
& \stackrel{\beta=4}{=} & \frac{\left(sp\right)^{1/2}}{2}\left.\arctan{\left(u\right)}\right|_{\left(\frac{r_2}{(sp)^{1/4}}\right)^2}^{\infty}\nonumber\\
& = & \frac{\sqrt{sp}}{2}\left(\pi/2-\arctan\left(\frac{r_2^2}{\sqrt{sp}}\right)\right).\nonumber
\end{eqnarray}
\end{proof}

\begin{Cor}
\label{EI}
The expected value of the Interference random variable, whose LT is given in (\ref{LTinterfD}) is equal to
\begin{eqnarray}
\label{EIgeneral2}
\mathbb{E}\left[\mathcal{I}\left(\rho,r_2,\beta,p,\lambda \right)\right] & = & \frac{p}{\left(\beta-2\right)r_2^{\beta}}\left(\beta-2+2\pi\lambda r_2^2\right),
\end{eqnarray} 
and is \textbf{independent of the cooperation parameter $\rho$}. For the special case of path-loss exponent $\beta=4$, the expression simplifies to
\begin{eqnarray}
\label{EIspecial2}
\mathbb{E}\left[\mathcal{I}\left(\rho,r_2,4,p,\lambda \right)\right] & = & \frac{p}{r_2^{4}}\left(1+\pi\lambda r_2^2\right).
\end{eqnarray} 
\end{Cor}

\begin{proof}
By taking the derivative of the LT for the Interference r.v. in (\ref{LTinterfD}) we get $\mathbb{E}\left[\mathcal{I}\right]=-\left.\frac{\partial\mathcal{L}_{\mathcal{I}}\left(s,\rho,r_2\right)}{\partial s}\right|_{s=0}$. The expression in (\ref{EIspecial2}) results by direct substitution in (\ref{EIgeneral}) of $\beta=4$ or can alternatively be derived 
by $-\left.\frac{\partial\mathcal{L}_{\mathcal{I}}\left(s,\rho,r_2\right)}{\partial s}\right|_{s=0,\ \beta=4}$ and the LT expression in (\ref{LTInocoop}).
\end{proof}
Finally, observe from (\ref{EIgeneral2}) that $\lim_{r_2\rightarrow 0}\mathbb{E}\left[\mathcal{I}\left(\rho,r_2,\beta,p,\lambda \right)\right]=\infty$, for $\beta>2$ and irrespective of the power $p$.


\newpage

\subsection{Proof for General Coverage Probability}

\begin{proof}\textbf{[Theorem \ref{CovProb}]}
The probability of coverage in (\ref{CovDef}), given the $\mathrm{SINR}$ expression in (\ref{SINRo2}) can be broken into two separate integrals, which should 
be further calculated.

\begin{eqnarray}
\label{TwoInts}
q_c\left(\rho\right) 	& := & 	\mathbb{E}_{r_1,r_2}\left[\mathbb{P}\left[\mathrm{SINR}\left(\rho,r_1,r_2\right)>T|r_1,r_2\right]\right]\nonumber\\
						& = &	\int_0^{\infty} \! \int_{\frac{r_1}{\rho}}^{\infty} \! \mathbb{P}\left[G_1>r_1^{\beta}T\left(\sigma^2+\mathcal{I}\right)|r_1,r_2\right] f_{r_1,r_2}\left(r_1,r_2\right) \,  \mathrm{d} r_2 \,  \mathrm{d} r_1\nonumber\\
						& + & 	\int_0^{\infty} \! \int_{r_1}^{\frac{r_1}{\rho}} \! \mathbb{P}\left[Z_{r_1,r_2}>2T\left(\sigma^2+\mathcal{I}\right)|r_1,r_2\right] f_{r_1,r_2}\left(r_1,r_2\right) \,  \mathrm{d} r_2 \,  \mathrm{d} r_1\nonumber\\
						& := & q_{c,1}\left(\rho\right)	+ q_{c,2}\left(\rho\right),\nonumber
\end{eqnarray}
where the $Z_{r_1,r_2}$ r.v. is defined in (\ref{RVz}). For the first integral we derive the probability expression

\begin{eqnarray}
\label{qc1}
\mathbb{P}\left[G_1>r_1^{\beta}T\left(\sigma^2+\mathcal{I}\right)|r_1,r_2\right]  & \stackrel{(a)}{=} & \mathbb{E}_{\mathcal{I}}\left[\mathbb{P}\left[G_1>r_1^{\beta}T\left(\sigma^2+\mathcal{I}\right)|r_1,r_2,\mathcal{I}\right]\right]\nonumber\\
& \stackrel{(b)}{=} & \mathbb{E}_{\mathcal{I}}\left[e^{-\frac{r_1^{\beta}}{p}T\left(\sigma^2+\mathcal{I}\right)}|r_1,r_2\right]  \nonumber\\
& \stackrel{(c)}{=} & e^{-\frac{r_1^{\beta}}{P} T\sigma^2}\mathcal{L}_{\mathcal{I}}\left(\frac{r_1^{\beta}}{p} T,\rho,r_2\right).\nonumber
\end{eqnarray}
In the above, (a) comes from the law of total probability, (b) from the fact that $G_1\sim\exp\left(1/p\right)$ and (c) from the definition of the Laplace transform. 

For the second integral, we will use the Proposition in \cite[Prop.5.4]{BaccelliBookStoch} or alternatively in \cite[Prop.2.2]{BaccelliOpAloha} for the expression of the coverage probability for general fading. The conditions to apply the proposition is that:

\begin{itemize}
\item $Z_{r_1,r_2}$ has a finite first moment and admits a square integrable density, shown valid in Lemma \ref{LemZ}.
\item $\mathcal{I}$ admits a square integrable density.
\end{itemize}
Then we get that 
\begin{eqnarray}
\mathbb{P}\left[Z_{r_1,r_2}>2T\left(\sigma^2+\mathcal{I}\right)|r_1,r_2\right] & = & \int_{-\infty}^{\infty} \! e^{-2j\pi\sigma^2 s}\cdot \mathcal{L}_{\mathcal{I}}\left(2j\pi s,\rho,r_2\right)\frac{\mathcal{L}_{Z}\left(-j\pi s/T,\frac{r_1^{\beta}}{p},\frac{r_2^{\beta}}{p}\right)-1}{2j\pi s}\,  \mathrm{d} s,\nonumber
\label{GenFade}
\end{eqnarray}
where $\mathcal{L}_{Z}\left(s,\mu_1,\mu_2\right)$ is the LT of the general fading r.v. $Z_{r_1,r_2}$, whose expression is provided in (\ref{LTz}) and $\mathcal{L}_{\mathcal{I}}\left(s,\rho,r_2\right)$ is the LT of the interference r.v. $\mathcal{I}$, provided in (\ref{LTinterfD}). In the case of DPC, the LT for the interference r.v. is provided in (\ref{LTinterfDcancel}). Substituting the probability expressions in the above $q_{c,1}\left(\rho\right)$ and $q_{c,2}\left(\rho\right)$ respectively and using the joint distribution over $r_1,r_2$ from (\ref{pdfN2f}) in Lemma \ref{LemJointD}, we conclude the proof.
\end{proof}
Observe that for $\rho = 1$ the above expression of $q_c\left(\rho\right)$ in (\ref{Tint1}) and (\ref{Tint2}) equals $q_c\left(1\right) = q_{c,1}\left(1\right) + 0$ and coincides with the expression in \cite[eq. (2)]{AndrewsCoverage}, given the presence of an interferer at a distance $r_2$ away from the typical location.

\end{document}